\documentclass[11pt]{article}

\usepackage[square,sort,comma,numbers]{natbib}
\usepackage{amsmath}
\usepackage{amsthm}
\usepackage{amssymb}
\usepackage{thmtools}
\usepackage{thm-restate}

\usepackage{algorithm}

\usepackage{algpseudocode}

\usepackage{color}
\usepackage{xcolor}

\usepackage{babel}

\usepackage{graphicx}
\usepackage{caption}
\usepackage{subcaption}
\usepackage{tikz}
\usepackage{wrapfig,epsfig}
\usepackage{psfrag}
\usepackage{epstopdf}

\usepackage{comment}

\usepackage[bookmarksnumbered=true]{hyperref}

\usepackage[margin=1in, letterpaper]{geometry}

\usepackage{setspace}
\usepackage{tabularx}
\usepackage{longtable}

\newcommand{\wh}[1]{\widehat{#1}}

\newcommand{\dist}{\operatorname{dist}}

\DeclareMathOperator*{\E}{\mathbb{E}}
\newcommand{\poly}{\mathsf{poly}}
\newcommand{\supp}{\mathsf{supp}}

\newcommand{\rect}{\mathsf{rect}}
\newcommand{\sinc}{\mathsf{sinc}}
\newcommand{\comb}{\mathsf{Comb}}

\DeclareMathOperator*{\argmin}{arg\,min}

\newcommand{\R}{\mathbb{R}}

\newcommand{\Z}{\mathbb{Z}}

\newcommand{\cC}{\mathcal{C}}

\newcommand{\bi}{\mathbf{i}}

\newtheorem{theorem}{Theorem}[section]

\newtheorem{fact}[theorem]{Fact}
\newtheorem{lemma}[theorem]{Lemma}
\newtheorem{claim}[theorem]{Claim}

\newtheorem{conjecture}[theorem]{Conjecture}
\newtheorem{corollary}[theorem]{Corollary}

\newenvironment{proofof}[1]{\bigskip \noindent {\it Proof of #1.}\quad }
{\qed\par\vskip 4mm\par}

\title{Improved Algorithms for Learning Fourier-sparse Signals}
\author{Dongrun Cai\thanks{\tt{cdr@mail.ustc.edu.cn}, University of Science and Technology of China, Hefei 230026, China.} \and Xue Chen\thanks{\tt{xuechen1989@ustc.edu.cn}, University of Science and Technology of China, Hefei 230026, China and Hefei National Laboratory, Hefei 230088, China. Supported by NSFC 62372424 and Quantum Science and Technology-National Science and Technology Major Project 2021ZD0302901.}
\and Xiaowei Shao \thanks{\tt{shaoxiaowei@mail.ustc.edu.cn}, University of Science and Technology of China, Hefei 230026, China.}
\and Yile Wang\thanks{\tt{xortrue@mail.ustc.edu.cn}, University of Science and Technology of China, Hefei 230026, China.}
} \date{}

\begin{document}

\maketitle

\begin{abstract}
    A classical problem in sparse Fourier transforms, which dates back to the work by Prony in 1795 at least, is to learn a $k$-Fourier-sparse signal $x(t):=\sum_{j=1}^k \alpha_j e^{2 \pi \bi f_j t}$ with arbitrary frequencies $f_1,\ldots,f_k$. We study this problem of learning $x(t)$ in a fixed time window $[-T,T]$ under adversarial noise with bounded $\ell_2$ norm, where the frequencies $f_1,\ldots,f_k$ may be ``off-grid'' --- arbitrarily located in a given bandlimit $[-F,F]$. In particular, our goal is to output a sparse interpolation $\tilde{x}$ such that $\tilde{x}(t) \approx x(t)$ in the time window $[-T,T]$.
    
    \begin{enumerate}
        \item Our first result shows that the sample complexity of interpolation is $k^2 \cdot O(\log \frac{k FT}{\epsilon})^2$. While its running time is $(\frac{k FT}{\epsilon})^{O(k)}$, this improves the previous upper bound $k^{4} \cdot (\log FT)^{O(1)}$ on the sample complexity substantially and leaves a gap of about $k$ to the lower bound $\Omega(k \log FT)$. 

        \item Our second result provides efficient algorithms to interpolate $x(t)$. The first algorithm takes $m=k^{3.75} \cdot (\log FT)^{O(1)}$ samples and $m^{\omega+o(1)}$ time ($\omega$ is the matrix multiplication exponent). Assuming that the growth of any $k$-Fourier-sparse signal cannot be significantly larger than the growth of the degree-$(k-1)$ Chebyshev polynomial  --- specifically, $x(t) \le e^{k \cdot O\big( \sqrt{\frac{|t|}{T}-1} \big)} \cdot \underset{s \in [-1,1]}{\max} |x(s)|$ for any $t \notin [-T,T]$, the second algorithm further improves the sample complexity to $m'=k^{3} \cdot (\log FT)^{O(1)}$ and the time complexity to $(m')^{\omega+o(1)}$. Both algorithms improve the sample complexity $k^{4} \cdot (\log FT)^{O(1)}$ and time complexity $k^{4 \omega} \cdot (\log FT)^{O(1)}$ of the best known result by Song, Sun, Weinstein, and Zhang (FOCS'2023).
    \end{enumerate}

    Technically, we improve several tools and analyses in previous works by Chen, Kane, Price, and Song (FOCS'2016) and Song, Sun, Weinstein, and Zhang (FOCS'2023). Our technical contributions include an optimal bound on the relative error of shifting one frequency in Fourier-sparse signals and a new analysis to improve the error of coarse estimates of $f_1,\ldots,f_k$.
\end{abstract}

\thispagestyle{empty}
\clearpage

\pagenumbering{arabic}
\section{Introduction}
We consider the classical problem of learning a Fourier-sparse signal under noise in the continuous setting. Let $x(t):=\sum_{j=1}^k \alpha_j \cdot e^{2 \pi \bi f_j t}$ denote a $k$-Fourier-sparse signal with frequencies $f_1,\ldots,f_k$ and amplitudes $\alpha_1,\ldots,\alpha_k$. In this work, these $k$ frequencies $f_1,\ldots,f_k$ could be arbitrary real numbers in a given bandlimit $[-F,F]$. At the same time, many problems and applications of the continuous Fourier transform in engineering and computer science consider a fixed time window $[-T,T]$ rather than $(-\infty,+\infty)$. Thus, the basic problem is to learn $x(t)$ from a noisy observation $y(t):=x(t)+\eta(t)$ in the time window $[-T,T]$.

In this work, we study the recovery of $x(t)$ under \emph{adversarial} noise $\eta(t)$ with a bounded $\ell_2$ norm in the time window. Formally, let $\|z\|_{[-T,T]}:=(\int_{-T}^T |z(t)|^2 \mathrm{d} t)^{1/2}$ denote the $\ell_2$ norm in the window $[-T,T]$. This work assume that the adversarial noise satisfies  $\|\eta\|^2_{[-T,T]} \le \epsilon \cdot \|x\|_{[-T,T]}^2$ for a small fixed constant $\epsilon$. 

If the $k$ frequencies $f_1,\ldots,f_k$ in $x(t)$ are located on the discrete grid $\mathbb{Z}/2T$, then $x(t)$ is periodic of length $2T$. A long line of research has studied efficient algorithms to learn the signal $x(t)$ in this setting, including \cite{Man92,GGIMS,AGS,GMS,HIKP,IK}). 

If the frequencies are arbitrary real numbers that are not multiples of $1/2T$ (are ``off-grid''), the problem becomes much more challenging. In the noiseless setting, several methods (including Prony's classical method \cite{Prony}, Reed-Solomon decoding \cite{m69} and the matrix pencil algorithm \cite{BM86}) can still identify the frequencies. However, these algorithms are not robust to noise. For example, $e^{2 \pi \bi f t}$ and $e^{2 \pi \bi (f+\epsilon/T) t}$ are $O(\epsilon)$-close to each other, which become indistinguishable under adversarial noise. Moitra \cite{Moitra} further proved that the noise has to be exponentially small in $k$, in order to learn $k$ frequencies whose gap $\underset{i\neq j}{\min} |f_i-f_j|<\frac{1}{2T}$. When $k$ frequencies are separated by a gap $\ge \frac{1}{2T}$, Moitra \cite{Moitra} showed efficient algorithms to recover them under polynomially small noise. In fact, assuming that the frequency gap is $\frac{(\log k)^{\Omega(1)}}{2T}$, a variety of robust and efficient sparse Fourier transform algorithms have been developed to recover frequencies \cite{BCGLS,PS15,SSWZ22,JLS23}. In summary, a frequency gap $\frac{1}{2T}$ is necessary to learn each frequency accurately under noise \cite{Moitra}.

For arbitrary frequencies without any gap, Chen, Kane, Price, and Zhao \cite{CKPS17} showed efficient and robust algorithms for learning $x(t)$ as a whole in the time window $[-T,T]$. In another word, while it is impossible to learn these frequencies without a frequency gap, their result shows how to learn the signal in the time window.
Specifically, given a noisy observation of $x(t)$ with $k$ \emph{arbitrary} frequencies, their algorithms output a sparse representation $\tilde{x}(t) \approx x(t)$ in the time window $[-T,T]$. In particular, they call $\tilde{x}$ an \emph{interpolation} of $x$ because it is a combination of low-degree polynomials and coarse frequency estimates, instead of accurate estimates of each frequency in $x$. 
Subsequent works \cite{CP19_colt, CP19_ICALP,  SSWZ23} have improved the sample complexity of \cite{CKPS17} to $m=k^4 \cdot (\log kFT)^{O(1)}$ and achieved time complexity $\tilde{O}(m^{\omega})$, where $\omega<2.371339$ \cite{alman2025more} is the matrix multiplication constant. Furthermore, algorithms and techniques developed for learning $k$-Fourier-sparse signals without a frequency gap have found applications in reconstructing signals with simple Fourier spectra \cite{AKMMVZ18}.

However, many open questions remain in the study of learning Fourier-sparse signals without a frequency gap. The most immediate question is about the sample complexity of learning $x(t)$. The state of the art is $k^4 \cdot (\log kFT)^{O(1)}$ \cite{CP19_colt, SSWZ23}, ignoring the running time. This leaves a large gap to the lower bound $\Omega(k \log FT)$. On the other hand, in the discrete setting, both the restricted isometry property (RIP) \cite{RV,HR16} and sparse discrete Fourier transforms \cite{HIKP,IK,VZZ19} showed that the sample complexity is $k \cdot (\log N)^{O(1)}$ for any discrete domain of size $N$. In the continuous setting, A natural question is to close the gap between the upper bound $\tilde{O}(k^4)$\footnote{In the rest of this work, we use $\tilde{}$ to omit $(\log kFT/\epsilon)^{O(1)}$ factors.} and the lower bound $\tilde{O}(k)$. In particular, are $\tilde{O}(k)$ samples sufficient to interpolate a $k$-Fourier-sparse signal $x(t)$?

\subsection{Our Results}
In this work, we continue the study of interpolating Fourier-sparse signals and  make progress on the above question. We show several algorithms that improve the sample complexity $\tilde{O}(k^4)$ of previous results \cite{CP19_colt,SSWZ23}. For ease of exposition, this work focuses on the sample complexity, denoted by $m$, for learning $x(t)$. This is because (1) in many applications of sparse Fourier transforms, taking a sample is more expensive than computation; and (2) the running time of many algorithms, including some of our algorithms, is $m^{O(1)}$ (actually $m^{\omega+o(1)}$ with the matrix multiplication constant $\omega<2.371339$).

Our first result shows that the sample complexity of interpolating $k$-Fourier-sparse signals is $\tilde{O}(k^2)$. Recall that $[-T,T]$ is the time window and $\|y\|_{[-T,T]}^2:=\int_{-T}^T |y(t)|^2 \mathrm{d} t$.
\begin{theorem}\label{inform:sample_complexity}[Informal version of Corollary~\ref{cor:query_complexity_learning}]
    Given any $k$, $F$, $T$, and $\epsilon$, let $y(t):=x(t)+\eta(t)$, where $x(t):=\sum_{j=1}^k \alpha_j e^{2 \pi \bi f_j t}$ has $k$ arbitrary frequencies $f_1,\ldots,f_k \in [-F,F]$ and $\|\eta(t)\|_{[-T,T]}^2 \le \epsilon \cdot \|x(t)\|_{[-T,T]}^2$. There exists an algorithm that takes $\tilde{O}(k^2)$ samples to output  $\tilde{x}$ such that $\|\tilde{x}-x\|_{[-T,T]}^2=O(\epsilon) \cdot \|x\|_{[-T,T]}^2$.
\end{theorem}
This improves the previous upper bound $\tilde{O}(k^4)$ to $\tilde{O}(k^2)$, which leaves a gap of $\tilde{O}(k)$ to the lower bound $\Omega(k \log FT)$. The key technique behind Theorem~\ref{inform:sample_complexity} is an almost-optimal bound on the relative error when we shift one frequency, say $f_k$, to another frequency $f_{k+1}$. Specifically, Lemma~\ref{lem:shift_one_freq} in Section~\ref{sec:overview} shows that the relative error is only $O(|f_k-f_{k+1}|)$, instead of $k^{O(k^2)} \cdot |f_k-f_{k+1}|$ from \cite{CKPS17}. An important corollary is that we can round arbitrary frequencies in $[-F,F]$ to a finite frequency net $\mathcal{N}:=\frac{\mathbb{Z}}{C k^2} \cap [-F,F]$. We note that this reduces the size of the previous frequency net $\mathcal{N}:=\frac{\mathbb{Z}}{k^{C k^2}} \cap [-F,F]$ in \cite{CKPS17} from an exponential in $k$ to a polynomial. The tight bound on the error of shifting a frequency in $x$ and the new frequency net $\mathcal{N}:=\frac{\mathbb{Z}}{C k^2} \cap [-F,F]$ may be of independent interest, given the wide applications of sparse Fourier transforms.

Because the algorithm in Theorem~\ref{inform:sample_complexity} does not run in time $\tilde{O}(k)^{O(1)}$, our next results are two efficient algorithms with sample complexity $m=\tilde{o}(k^4)$ and running time $m^{\omega+o(1)}=\tilde{O}(k)^{O(1)}$. For convenience, we call a learning algorithm efficient only if its time complexity is $\tilde{O}(k)^{O(1)}$.

\begin{theorem}\label{inform:fast_algorithm}
    Given any $k$, $F$, $T$, and a small constant $\epsilon>0$, let $y(t):=x(t)+\eta(t)$ for $x(t):=\sum_{j=1}^k \alpha_j e^{2 \pi \bi f_j t}$ with $k$ arbitrary frequencies $f_1,\ldots,f_k \in [-F,F]$ and $\|\eta(t)\|_{[-T,T]}^2 \le \epsilon \cdot \|x(t)\|_{[-T,T]}^2$. There exists an algorithm that takes $m=\tilde{O}(k^{3.75})$ samples and $\tilde{O}(m^{\omega})$ time to output $\tilde{x}$ with $\|\tilde{x}-x\|_{[-T,T]}^2 \le O(\epsilon) \cdot \|x\|_{[-T,T]}^2$.    

    In particular, for some $\ell=O(k)$, $\tilde{x}(t):=\sum_{j=1}^{\ell} e^{2 \pi \bi \tilde{f}_j t} \cdot q_j(t)$ with $\ell$ frequency estimates  $\tilde{f}_1,\ldots,\tilde{f}_\ell$ and $\ell$ polynomials $q_1,\ldots,q_{\ell}$ of degree $\tilde{O}(k^{2.75})$.
\end{theorem}

While the algorithm of Theorem~\ref{inform:fast_algorithm} follows the same framework as the previous algorithms \cite{CKPS17,CP19_ICALP,SSWZ23}, its analysis is more involved. Our main technical contribution is an improved bound on the error of the estimates $\tilde{f}_1,\ldots,\tilde{f}_\ell$. Specifically, we prove that these estimates cover most frequencies in $x(t)$ within a covering radius $D:=\tilde{O}(k^{2.75}/T)$. This improves the previous bound $\tilde{O}(k^3/T)$ \cite{CKPS17,SSWZ23}. 

Our analysis is based on a pair of filter functions $(H,\wh{H})$ constructed by Chen and Price \cite{CP19_ICALP}, where $H(t)$ acts like a box function on the time window $[-T,T]$ and its Fourier transform $\wh{H}(f)$ is compact in $[-\tilde{O}(k^2/T),\tilde{O}(k^2/T)]$. This pair allows the learning algorithm to consider $H \cdot y$ over $\mathbb{R}$ intead of $[-T,T]$ and apply  the continuous Fourier transforms to obtain $\wh{H \cdot y}$. This leads to an efficient algorithm for one-cluster recovery \cite{CP19_ICALP}: if $k$ frequencies lie in a small cluster, say each $f_i \in [f-O(k^2/T),f+O(k^2/T)]$ for some $f \in [-F,F]$, it finds $\tilde{f}$ with estimation error $|f-\tilde{f}|=O(|\text{support size of } \wh{H}|)=\tilde{O}(k^2/T)$. Moreover, this is tight \cite{CP19_ICALP}: by sending $f_1,\ldots,f_k$ to $0$ and taking a Taylor expansion, $k$-Fourier-sparse signals can get arbitrarily close to any polynomial of degree $k-1$ on any interval. Then the extreme concentration of the Chebyshev polynomials implies that the estimation error of $\tilde{f}$ is $\tilde{\Omega}(k^2)$.

However, for $k$ arbitrary frequencies, previous analyses in \cite{CKPS17,SSWZ23} lose an extra factor of $k$ on the error of $\tilde{f}_1,\ldots,\tilde{f}_{\ell}$, compared with \cite{CP19_ICALP}. In this work, we showed that the error is $\tilde{O}(k^{2.75}/T)$ which improves previous analyses by a factor of $k^{0.25}$. Our approach is based on an algorithm partitioning $f_1,\ldots,f_k$ into clusters and a rigorous analysis that shows these clusters are almost orthogonal. We refer to Algorithm~\ref{alg:cluster1} for this partition algorithm and Theorem~\ref{thm:total_energy} for its guarantee.

An intriguing problem is to improve the error of these estimates $\tilde{f}_1,\ldots,\tilde{f}_{\ell}$ to $\tilde{O}(k^2)$, which matches the lower bound $\tilde{\Omega}(k^2)$ demonstrated by the Chebyshev polynomial \cite{CP19_ICALP}. Our last result shows that this is plausible if the growth rate of the Chebyshev polynomials outside $[-T,T]$ is \emph{asymptotically} the largest among all $k$-Fourier-sparse signals.

For ease of exposition, we discuss this part by fixing $T=1$ and the time window to be $[-1,1]$. As mentioned earlier, $k$-Fourier-sparse signals can get arbitrarily close to any polynomial of degree $k-1$ on any interval such as $[-2,2]$. From \cite{CKPS17,CP19_ICALP}, the error of frequency estimates depends on the magnitude of $x(t)$ just outside the interval $[-1,1]$. 
In particular, previous result \cite{CP19_ICALP} bounded 
\[
|x(t)| \le k^{O(1)} \cdot \underset{s \in [-1,1]}{\max} \{|x(s)|\} \cdot \min\{e^{\tilde{O}(k^2) \cdot (|t|-1)}, O(|t|)^k\} \text{ for }t \notin [-1,1].
\] The term $e^{\tilde{O}(k^2) \cdot (|t|-1)}$ turns out to be extremely useful in bounding the error of frequency estimates for one cluster in \cite{CP19_ICALP}. We show that if one can improve this term to $e^{O(k) \cdot \sqrt{|t|-1}}$\footnote{In fact, $e^{\tilde{O}(k) \cdot \sqrt{|t|-1}}$ is sufficient for our improvement. But we use $e^{O(k) \cdot \sqrt{|t|-1}}$ for ease of exposition.}, matching the Chebyshev polynomial of degree $k-1$ at $t=1+\delta$ for any $\delta \in (0,0.1)$, then the error of  $\tilde{f}_1,\ldots,\tilde{f}_\ell$ is $\tilde{O}(k^2)$ instead of $\tilde{O}(k^{2.75})$. We provide a formal statement of this conjecture\footnote{After submitting this work, we realized Zhang provided a proof of this conjecture in \cite{zhang2026optimalextrapolationboundssparse} during the preparation of this work.}.

\begin{conjecture}\label{conj:growth_rate}
 For any $x(t):=\sum_{j=1}^k \alpha_j e^{2 \pi \bi f_j t}$ with $k$ arbitrary frequencies $f_1,\ldots,f_k$, $|x(t)| \le k^{O(1)} \cdot \max_{s \in [-1,1]} {|x(s)|} \cdot \min\{e^{\sqrt{|t|-1} \cdot O(k)}, O(|t|)^k\}$ for any $t \notin [-1,1]$.
\end{conjecture}
Note that the Chebyshev polynomial of degree $k-1$ satisfies $q(t)=e^{\Theta(k \cdot \sqrt{t-1})}$ for $t \in (1,1.1)$ and $|q(t)| \le 1$ for $t \in [-1,1]$. Conjecture~\ref{conj:growth_rate} indicates that this polynomial has the largest growth asymptotically. Assuming this, we present a learning algorithm with sample complexity $m=\tilde{O}(k^3)$.

\begin{theorem}\label{inform:algorithm_conjecture}
    Given any $k$, $F$, and a small constant $\epsilon>0$, let $y(t):=x(t)+\eta(t)$ for $x(t):=\sum_{j=1}^k \alpha_j e^{2 \pi \bi f_j t}$ with $k$ arbitrary frequencies $f_1,\ldots,f_k \in [-F,F]$ and $\|\eta(t)\|_{[-1,1]}^2 \le \epsilon \cdot \|x(t)\|_{[-1,1]}^2$. If Conjecture~\ref{conj:growth_rate} holds for any $k$-Fourier sparse signals, there exists an efficient algorithm that takes $m=\tilde{O}(k^{3})$ samples and $\tilde{O}(m^{\omega})$ time to output $\tilde{x}$ with $\|\tilde{x}-x\|_{[-1,1]}^2 \le O(\epsilon) \cdot \|x\|_{[-1,1]}^2$.    
\end{theorem}

The last remark is that the learning algorithms in Theorem~\ref{inform:fast_algorithm} and Theorem~\ref{inform:algorithm_conjecture} apply linear regression to find the best fitting representation from the noisy samples. Essentially, the time complexity $\tilde{O}(m^{\omega})$ is the time complexity of applying linear regression to $m$ samples \cite{CP19_colt}.

\begin{table}[h]
\centering
\begin{tabular}{|c|c|c|}
     \hline
     Results &  sample complexity & time complexity\\
     \hline
     \cite{CP19_colt} &  $\tilde{O}(k^4)$ & $ (k^{O(k^2)} \cdot FT)^{O(k)}$ \\
     \cite{SSWZ23} & $\tilde{O}(k^4)$ & $\tilde{O}(k^{4\omega})$  \\
     Theorem~\ref{inform:sample_complexity} & $\tilde{O}(k^2)$ & $(k \cdot FT)^{O(k)}$ \\
     Theorem~\ref{inform:fast_algorithm} & $\tilde{O}(k^{3.75})$ & $\tilde{O}(k^{3.75 \omega})$ \\
     Theorem~\ref{inform:algorithm_conjecture} under Conjecture~\ref{conj:growth_rate} & $\tilde{O}(k^{3})$ & $\tilde{O}(k^{3 \omega})$ \\
     \hline
\end{tabular}
\caption{Summary on the sample complexity and time complexity of learning $k$-Fourier-sparse signals with a frequency gap.}\label{tab:comparison}
\end{table}
We summarize our results with previous bounds in Table~\ref{tab:comparison}. In the rest of this work, we assume that the relative error $\epsilon$ is a small constant and $T=1$.

\subsection{Related Works}
\paragraph{Sparse Fourier transforms in the discrete setting.} Sparse discrete Fourier transforms have a large literature, with rich connections to cryptography \cite{GL89} and coding theory \cite{AGS}. Its results can be separated into two lines. The first line carefully chooses samples (measurements) to allow sublinear time recovery (to name a few \cite{GGIMS,GMS,HIKP,Iw13,IKP,K16}). Our result is closely related to this line. Another line of research considers randomly chosen samples (measurements) and gives generic recovery algorithms such as $\ell_1$ minimization under the restricted isometry property \cite{RV,HR16}. While the first line has better sample complexity and running time, the second line has smaller failure probabilities. For a discrete domain of size $N$, the best known results achieve $O(k\log N)$ samples \cite{IK} or $O(k\log^2 N)$ time \cite{HIKP} separately.

However, algorithms in the discrete setting cannot be applied directly to the continuous problem studied in this work. If the continuous problem has frequencies ``off-the-grid'', the discrete approximation becomes $k/\epsilon$-sparse. More importantly, this approximation requires all frequencies to be well separated.

\paragraph{Sparse Fourier transforms in the continuous setting.} A line of research \cite{BCGLS,HK15,PS15,SSWZ22} has constructed sparse Fourier transform algorithms in the continuous setting directly. These algorithms learn each frequency $f_j$ and its amplitudes $\alpha_j$ in time $k \cdot (\log k FT)^{O(1)}$ like the discrete algorithms. However, all these algorithms require that the $k$ frequencies in $x$ have a gap $\min_{i \neq j}|f_i-f_j| \ge \frac{(\log k)^{\Omega(1)}}{T}$.

\paragraph{Super-resolution.} Learning $k$-Fourier-sparse signals is closely related to a fundamental task in imaging, called super-resolution. The task is to recover frequencies and amplitudes in $x(t):=\sum_{j=1}^{k} \alpha_j e^{2 \pi \bi f_j t}$. There are a variety of methods that work in the noiseless setting for $m=k$ samples, including Prony's method \cite{Prony}, Reed-Solomon decoding \cite{m69}, and the matrix pencil method \cite{BM86} (see more references in \cite{KATZ2024101687}). However, for exponentially small noise in the time window $[-T,T]$, Moitra \cite{Moitra} showed that it is impossible to recover each frequency accurately when the gap between frequencies is $<\frac{1}{2T}$. At the same time, Moitra provided an algorithm with $O(T)$ samples that tolerates polynomially small noise, when the gap between frequencies is at least $\frac{1}{2T}$. Various algorithms (to name a few \cite{FL12,TBSR,CF14,YX15}) based on convex optimization and compressed sensing have been developed in the last two decades. However, all these algorithms require the gap between frequencies to be at least $\frac{1}{2T}$ in order to recover frequencies.

While both sparse Fourier transform in the continuous setting and super-resolution study algorithms for recovering frequencies, their foci are different. The goal of sparse Fourier transform is to optimize the running time (and sample complexity). On the other hand, super-resolution is concerned with how the gap between frequencies affects other parameters such as the sample complexity, robustness, and the length of the time window.

\paragraph{Interpolating Fourier-sparse signals.} Chen, Kane, Price, and Song \cite{CKPS17} showed that the gap between frequencies is not necessary for learning the whole signal. Their algorithm provides an interpolation with sparsity $(k \log FT)^{O(1)}$ in $(k \log FT)^{O(1)}$ time. This result has been improved significantly by subsequent works \cite{CP19_colt,CP19_ICALP,SSWZ23}. Specifically, Chen and Price proposed a weighted sampling distribution to reduce the sample complexity in \cite{CP19_colt} and improved the construction of filter functions and sampling algorithms in \cite{CP19_ICALP}. Song, Sun, Weinstein, and Zhang \cite{SSWZ23} extended these techniques to provide an efficient interpolation in $m=\tilde{O}(k)^{4}$ samples and 0$m^{\omega+o(1)}$ time. Our algorithms are based on the techniques developed in these works.

Moreover, the techniques developed in these works have found applications beyond sparse Fourier transforms. The authors of \cite{AKMMVZ18} showed how to reconstruct signals with simple Fourier spectra. While their result provides almost optimal sample complexity to guarantee $\tilde{x}(t) \approx x(t)$ in the time window, it assumes that the positions of the spectra are given.

Two recent works \cite{LLM22, CDHNSY25} studied different approaches to interpolate Fourier-sparse signals. Li, Liu, and Moitra \cite{LLM22} showed how to efficiently interpolate $x(t)$ in an interval smaller than the time window using $\tilde{O}(k)$-sparse interpolations, which improves the interpolation sparsity of \cite{CKPS17,SSWZ23}. The authors of \cite{CDHNSY25} proposed an algorithm with running time $e^{\tilde{O}(k)}$ to reconstruct the Fourier spectrum of $\wh{x}$ with respect to the Wasserstein distance.


\paragraph{Exponential Sums.} Various properties of Fourier-sparse signals have been studied in approximation theory and Fourier analysis \cite{TuranBook} in terms of inequalities of exponential sums. In particular, Erd\'{e}lyi \cite{erdelyi2016inequalities} proved tight bounds on $\underset{k\text{-Fourier-sparse } x}{\max} \frac{|x(t)|}{\|x\|_{[-1,1]}}$. These bounds are extremely useful in designing the sampling distribution for samples from $[-1,1]$ (although weaker bounds were used in \cite{CKPS17,CP19_colt,CP19_ICALP}).
Moreover, for Conjecture~\ref{conj:growth_rate}, Borwein and Erd\'{e}lyi \cite{Remez_Bor_Erd} proved a stronger upper bound for a different family: $g(t):=\sum_{j=1}^k \alpha_j e^{\lambda_j t}$ with $\lambda_1,\ldots,\lambda_k \in \mathbb{R}$ satisfies $|g(t)| \le e^{O(k (|t|-1))} \cdot \underset{s \in [-1,1]}{\max} |g(s)|$ for $t \notin [-1,1]$. 


\subsection{Discussion}
In this work, we improve the sample complexity of learning $k$-Fourier-sparse signals. We show that the information theoretic upper bound is $m=\tilde{O}(k^2)$ and give efficient algorithms with $m=\tilde{O}(k^{3.75})$ samples and with $m=\tilde{O}(k^3)$ samples under Conjecture~\ref{conj:growth_rate}. Our work leaves many intriguing open questions, and we list some of them here.

\begin{enumerate}
    \item Essentially, the information theoretic upper bound $m=\tilde{O}(k^2)$ comes from the union bound over all possible choices of $k$ frequencies in the net 
    $\mathcal{N}:=[-F,F] \cap \frac{\mathbb{Z}}{C \cdot k^2}$. Can we apply the chaining arguments of the restricted isometry property from \cite{RV,HR16} to improve the union bound and reduce the sample complexity to $m=\tilde{O}(k)$?   

    \item Is Conjecture~\ref{conj:growth_rate} true? Moreover, how to use it to obtain efficient learning algorithms within $\tilde{o}(k^3)$ samples?
    
    \item Previous efficient learning algorithms \cite{CKPS17,SSWZ23}, including ours (for sparse Fourier transforms without a frequency gap), use combinations of low-degree polynomials and frequency estimates to interpolate $x(t)$ in the time window. Are there more efficient methods to interpolate $x(t)$?
    For example, Li, Liu, and Moitra \cite{LLM22} showed how interpolate $x(t)$ on a smaller interval $[-(1-c)T,(1-c)T]$ with sparser $\tilde{x}$; and very recent work by the authors of \cite{CDHNSY25} proposed an intriguing approach to reconstruct $x$ without learning each frequency accurately --- its output $x'$ has a small Wasserstein distance between $\wh{x'}$ and $\wh{x}$. 

    \item How to apply techniques developed for learning Fourier-sparse signals to learning signals with simple Fourier spectra \cite{AKMMVZ18}?
\end{enumerate}

\paragraph{Organization.}
The rest of this work is organized as follows. 
We introduce notations and properties of Fourier-sparse signals in Section~\ref{sec:prel}. 
We provide an overview of our algorithms in Section~\ref{sec:overview}.
Then we prove Theorem~\ref{inform:sample_complexity} using an improved frequency-shifting lemma in Section~\ref{sec:net_freqs}.
Next, Section~\ref{sec:heavy_region} proves the guaranty for the frequency estimates used in Theorem~\ref{inform:fast_algorithm}, while Section~\ref{sec:heavy-region-clustering-conjecture} strengthens this guaranty under Conjecture~\ref{conj:growth_rate} for Theorem~\ref{inform:algorithm_conjecture}. Finally, we combine these ingredients to prove Theorem~\ref{inform:fast_algorithm} and Theorem~\ref{inform:algorithm_conjecture} in Section~\ref{sec:main_proof}.

\section{Preliminaries} \label{sec:prel}
For ease of exposition, we fix the time window to be $[-1,1]$ and the bandlimit to be $[-F,F]$ in the rest of this work. We always treat the error $\epsilon$ as a fixed small constant, while we use $C$ to denote various constants in the proof. Also, we use $a=b \pm c$ to indicate $a \in [b-c,b+c]$.

For an interval $I \subset (-\infty,+\infty)$, let $\mathbf{1}_I$ denote the indicator function of interval $I$. So for a signal $x:\mathbb{R}\rightarrow \mathbb{C}$, $\mathbf{1}_I \cdot x$ denotes the truncation of $x$ in the interval $I$. 

Let $\|z\|_2:=(\int |z(t)|^2 \mathrm{d} t)^{1/2}$ and $\|z\|_I:=(\int_I |z(t)|^2 \mathrm{d} t)^{1/2}$ for any signal $z$. For convenience, we call $\|z\|_2^2$ the energy of $z$ and $\|z\|_{I}^2$ the energy of $z$ in $I$. For two integrable functions $x$ and $y$, we define the corresponding inner product $\langle x,y \rangle=\int x(t) \bar{y}(t) \mathrm{d} t$ and $\langle x , y \rangle_I = \int_I x(t) \bar{y}(t) \mathrm{d} t$. 

We review several facts about the Fourier transform. The Fourier transform $\wh{g}(f)$ of an integrable function $g$ is 
\[
\wh{g}(f)=\int_{-\infty}^{+\infty} g(t) e^{-2\pi \bi ft} \mathrm{d} t.
\]
We recall the classical Plancherel and Parseval identities for the inner product $\langle\cdot,\cdot\rangle$.
\begin{theorem}\label{thm:Plancherel_Parseval}
    For any integrable function $x$, $\|x\|_2=\|\wh{x}\|_2$. For any two integrable functions $x$ and $y$, $\langle x,y \rangle=\langle \wh{x}, \wh{y} \rangle$.
\end{theorem}

We use $g\cdot h$ to denote the point-wise dot product $g(t) \cdot h(t)$ and $g^k$ to denote $g(t) \cdot g(t) \cdots g(t) \cdot g(t)$. Similarly, $g*h$ denotes the convolution of $g$ and $h$: $\int g(x)h(t-x) \mathrm{d} x$ and $g^{*k}$ denotes the convolution $g*g*\cdots * g$. For a function $h$ whose Fourier transform $\wh{h}$ has a compact support, we define the Fourier support of $h$ as the set $\{f:\wh{h}(f) \neq 0\}$. 

\paragraph{Properties of Fourier-sparse signals}
We state several useful bounds for Fourier-sparse signals from \cite{erdelyi2016inequalities,CP19_ICALP}.
\begin{lemma} \label{lemma:bounds_Fourier_sparse_signals}
Any $k$-Fourier-sparse signal $x$ satisfies the following bounds:
\begin{enumerate}
    \item \label{item:uniform_bound_on_interval} $|x(t)| \le \frac{\pi k}{2} \cdot \| x \|_{[-1, 1]}$ for every $|t| \le 1$ (Theorem 2.3 in \cite{erdelyi2016inequalities}); 
    
    \item \label{item:leverage_bound_on_interval} $|x(t)| \le \sqrt{\frac{2k}{1-|t|}} \cdot \| x \|_{[-1, 1]}$ for every $|t| < 1$ (Theorem 7.1 in \cite{erdelyi2016inequalities}); 
    
    \item \label{item:polynomial_bound_outside_interval} $|x(t)| \le \left(e (|t| + 1) \right)^k \cdot \sup_{t \in [-1, 1]} |x(t)|$ for every $|t| > 1$ (Lemma 12.2 in \cite{erdelyi2016inequalities});
    
    \item \label{item:exponential_bound_outside_interval} $|x(t)| \le \poly(k) \cdot e^{O(t-1) \cdot k^2 \log k} \cdot \max_{s \in [-1,1]} |x(s)| $ for every $|t| > 1$ (Theorem 1.4 in \cite{CP19_ICALP}). 
\end{enumerate}
\end{lemma}
The last property is important in the construction of filter functions. 
Essentially, Conjecture~\ref{conj:growth_rate} says that the exponent of Property \ref{item:exponential_bound_outside_interval} can be improved to its square root, which matches the growth of the Chebyshev polynomial of degree $(k-1)$ asymptotically.


\section{Overview}\label{sec:overview}
We provide a high level overview of our methods in this section.

\paragraph{Frequency net $\mathcal{N}$.} To bound the sample complexity of learning a $k$-Fourier-sparse signal $x$, a standard method is a net argument, which constructs a frequency net $\mathcal{N}$ of bounded size here. Basically, this net $\mathcal{N} \subset [-F,F]$ has a finite size and guarantees that for any $x(t):=\sum_{j=1}^k \alpha_j e^{2 \pi \bi f_j t}$ with arbitrary frequencies $f_1,\ldots,f_k$, one can find $f'_1,\ldots,f'_k \in \mathcal{N}$ and $x' :=\sum_{j=1}^k \alpha'_j \cdot e^{2 \pi \bi f'_j t}$ such that $x'(t) \approx x(t)$ in the time window $t \in [-1,1]$ (see Theorem~\ref{thm:net_frequency} for a formal statement). In the off-grid setting, the challenge of this net argument is to bound the relative error of $\|x'-x\|_{[-1,1]}$ by $\|x\|_{[-1,1]}$ instead of the Fourier coefficients $|\alpha_1|,\ldots,|\alpha_k|$. This is because $k$ arbitrarily close frequencies could make $\|x\|_{[-1,1]}$ arbitrarily small compared to their Fourier coefficients. 

To bound the relative error between $x'$ and $x$, it suffices to bound the relative error of replacing one frequency $f_k$ in $x(t)$ by a frequency $f_{k+1}$ in $\mathcal{N}$. One of the key technical results in \cite{CKPS17} (Lemma 8.5) shows that the relative error of replacing $f_k$ by $f_{k+1}$ is at most $k^{O(k^2)} \cdot |f_k - f_{k+1}|$. This result has several important corollaries. At first, it shows $\mathcal{N}:=[-F,F] \cap \frac{1}{k^{O(k^2)}} \cdot \mathbb{Z} $ is a good frequency net. Secondly, this net implies an extra property of the approximation $x' :=\sum_{j=1}^k \alpha'_j \cdot e^{2 \pi \bi f'_j t}$ whose frequencies have a gap $\eta:=k^{-O(k^2)}$: $\|x\|_{[-1,1]}^2 \ge k^{-O(k^2)} \cdot (\eta)^k \cdot \sum_{j=1}^k |\alpha'_j|^2$. Finally, this relation between $\|x\|_{[-1,1]}^2$ and $\sum_{j=1}^k |\alpha'_j|^2$ implies that a degree-$\tilde{O}(k^3)$ Taylor expansion of $x'(t)$ is a good approximation.

When it is impossible to recover the $k$ frequencies in $x(t)$ and $x'(t)$ accurately, this suggests an efficient algorithm to interpolate $x'(t)$ as a summation of products of low-degree polynomials and wave functions: find coarse estimates $\tilde{f}_j$ for each $f'_j$ (and $f_j$) and interpolate $x(t)$ as
\begin{equation}\label{eq:approx_by_poly}
    x(t) \approx \sum_{j=1}^k e^{2 \pi \bi \tilde{f}_j t} \cdot q_j(t) \text{ with each polynomial $q_j$ of degree } O(|\tilde{f}_j-f_j|)+\tilde{O}(k^3).
\end{equation}

Our first technical result is an improvement in the relative error after replacing $f_k$ by any nearby frequency $f_{k+1}$. 
\begin{lemma}\label{lem:shift_one_freq}
    For any $x(t)=\sum_{j=1}^k \alpha_j e^{2 \pi \bi f_j t}$ with $k$ arbitrary frequencies $f_1,\ldots,f_k$ and any $f_{k+1} \notin \{f_1,\ldots,f_{k-1}\}$, there exists $x'(t)=\sum_{j=1}^{k-1} \alpha_j e^{2 \pi \bi f_j t} + \alpha_{k+1} e^{2 \pi \bi f_{k+1} t}$ such that
    \[
    \|x(t)-x'(t)\|_{[-1,1]} \le O(1) \cdot |f_k-f_{k+1}| \cdot \|x\|_{[-1,1]}.
    \]
\end{lemma}
Lemma~\ref{lem:shift_one_freq} shows that the relative error is $O(1) \cdot |f_{k+1}-f_k|$, much smaller than the previous bound $k^{O(k^2)} \cdot |f_{k+1}-f_k|$. The proof idea is to construct an approximation vector instead of estimating the Gram matrix of the wave functions $e^{2 \pi \bi f_1 t},\ldots,e^{2 \pi \bi f_k t},e^{2 \pi \bi f_{k+1} t}$. Specifically, let vectors $v_1:=(e^{2 \pi \bi f_1 t})_{t \in [-1,1]},\ldots,v_k:=(e^{2 \pi \bi f_k t})_{t \in [-1,1]},v_{k+1}:=(e^{2 \pi \bi f_{k+1} t})_{t \in [-1,1]}$. Let $v^{\bot}_{k}$ and $v^{\bot}_{k+1}$ be the components of $v_{k}$ and $v_{k+1}$ that are orthogonal to $\mathrm{span}\{v_1,\ldots,v_{k-1}\}$ respectively. Then the relative error of replacing $f_k$ by $f_{k+1}$ is the relative distance between the normalized unit vectors of $v^{\bot}_k$ and $v^{\bot}_{k+1}$. Equivalently, this is the ratio between the norm of the component of $v^{\bot}_{k+1}$ orthogonal to $v^{\bot}_k$ and $\|v^{\bot}_{k+1}\|_{[-1,1]}$ itself. Our new upper bound is obtained by presenting a decomposition of $v_{k+1}$ into two vectors $w$ and $(v_{k+1} - w)$ in $\mathrm{span}\{v_1,\ldots,v_k\}$. Since any upper bound on $\|w\|$ provides an upper bound on the component of $v^{\bot}_{k+1}$ orthogonal to $v^{\bot}_k$, $w \bot \mathrm{span}\{v_1,\ldots,v_k\}$ is no longer necessary. In Section~\ref{sec:net_freqs}, we present an integral operator of wave functions to construct $w$ and compare it with $v^{\bot}_{k+1}$ via the classical Poincare inequality.

As discussed above, this shows a much smaller net $\mathcal{N}:=[-F,F] \cap \frac{\epsilon}{k^{O(1)}} \cdot \mathbb{Z}$, which improves the sample complexity of learning $x$ (see Theorem~\ref{cor:query_complexity_learning} in Section~\ref{sec:net_freqs}) and the degree of the Taylor expansion of $x'(t)$ (see Lemma~\ref{lem:low_deg_approx} in Section~\ref{sec:net_freqs}). However, another bottleneck of previous learning algorithms \cite{CKPS17,SSWZ23} is the error of $|\tilde{f}_j-f_j|$ in \eqref{eq:approx_by_poly}. Before describing our methods, we review previous methods based on a pair of filter functions --- $H$ and its Fourier transform $\wh{H}$.

\paragraph{Filter functions $(H_{\ell,\delta},\wh{H_{\ell,\delta}})$.} This pair of filter functions from \cite{CKPS17, CP19_ICALP} has two properties: (1) $\wh{H}$ is compact and (2) $H$ acts like a box function on any $k$-Fourier-sparse signal $x$: $H \cdot x \approx \mathbf{1}_{[-1,1]} \cdot x$. Then $H \cdot y$ is a noisy approximation of $x([-1,1])$ and its Fourier transform $\wh{H}*\wh{y} \approx \wh{H}*\wh{x}$ preserves the structure of $\wh{x}$. The parameter $\ell$ bounds the number of frequencies and $\delta$ is the error in the approximation. 


We state the main properties  of these filter functions from \cite{CP19_ICALP} as follows. Let $\rect_s(t)$ denote the box function of width $s$: $\rect_s(t)=1/s$ if and only if $|t| \le s/2$; let $\sinc(sf)$ denote its Fourier transform $\frac{\sin(\pi s f)}{\pi s f}$.

\begin{lemma} \label{lemm:construction_H}
Given the sparsity $\ell$ and error $\delta$, 
let $C=O(1)$, $S=\ell^2 \log \ell$, $\alpha_H = 1 - \frac{\delta}{C \ell^2}$, and $s_0 = \Theta\left(\frac{\ell^2}{\delta} \sqrt{\log \frac{\ell}{\delta}}\right)$ be a normalizer such that
\begin{align*}
\wh{H_{\ell,\delta}}(f) & :=s_0\bigg( (\rect_{\frac{C \ell^2}{\delta}}(f)^{*C \log \ell/\delta}) *(\rect_{CS}(f)^{*C} *(\rect_{CS/2}(f)^{*2C}) * \cdots *(\rect_{C}(f)^{*CS}\bigg)\cdot \sinc(2 \alpha_H f), \\
H_{\ell,\delta}(t) & :=s_0 \bigg( (\sinc(\frac{C \ell^2}{\delta} \cdot t)^{C \log \ell/\delta}) \cdot (\sinc(CS \cdot t)^{C} *(\sinc(CS/2 \cdot t)^{2C}) \cdot \cdots \cdot(\rect(C t)^{CS}\bigg) * \rect_{2 \alpha_H}(t).    
\end{align*}
    Then $(H_{\ell,\delta},\wh{H_{\ell,\delta}})$ satisfies the following properties.
        \begin{enumerate}
        \item $supp(\wh{H_{\ell,\delta}}) =[-\Delta_{{\ell,\delta}},+\Delta_{{\ell,\delta}}]$ for $\Delta_{\ell,\delta}:=C^2 \cdot (\frac{\ell^2 \log(\ell/\delta)}{\delta} + S \log S)$. 

        \item $\|H_{\ell,\delta} \cdot z\|_2^2 = (1 \pm \delta) \cdot \|z\|_{[-1,1]}^2$ for any $\ell$-Fourier-sparse signal $z$.    

    \end{enumerate}
\end{lemma}
Since $k$ is fixed and $\epsilon$ is a fixed constant, this work uses $H_k$ to denote $H_{k,\epsilon}$ and $\Delta_k$ to denote $\Delta_{k,\epsilon}=\tilde{\Theta}(k^2)$. Roughly speaking, the main properties of $(H_k,\wh{H_k})$ are $H_k(t) = 1 \pm \delta$ for $|t| \le 1- \Omega(\frac{\delta}{\ell^2})$ and $H_k(t) \approx 0$ for $|t|>1$; and  $\wh{H}$ has a compact support in $[-\Delta_k,\Delta_k]$. Previous algorithms \cite{CKPS17,CP19_ICALP,SSWZ23}  reconstruct $H_k(t) \cdot y(t)$ for $t \in \mathbb{R}$ from the observation $y(t)=x(t)+\eta(t)$ on the time window and obtain coarse estimates of $\tilde{f}_1,\ldots,\tilde{f}_{\ell}$ from its continuous Fourier transform $\wh{H_k \cdot y}$. 

 Also, our algorithm applies $H_k$ to $y$ and use $\wh{H_k \cdot y}$ to obtain $\tilde{f}_1,\ldots,\tilde{f}_\ell$; but its analysis uses $H_{\ell,\delta}$ with a variety of parameters in several places. For completeness, we show the exact properties of $(H_{\ell,\delta},\wh{H_{\ell,\delta}})$ and a full proof of Lemma~\ref{lemm:construction_H} in Appendix~\ref{sec:proof_H_delta}. Because $\|y\|_{[-1,1]}=\|x\|_{[-1,1]} \pm \|\eta\|_{[-1,1]}$ and $\|H_k x\|_2 \approx \|x\|_{[-1,1]}$, the following three energies $\|y\|_{[-1,1]}^2$, $\|x\|_{[-1,1]}^2$, and $\|H_k x\|_2^2$ are very close, so we use them to denote the energy of the observation for convenience.


\paragraph{Frequency Recovery.} Our frequency estimation procedure (for interpolating \eqref{eq:approx_by_poly}) uses frequency estimation algorithms developed in \cite{CKPS17,CP19_ICALP,SSWZ23}, but our contribution here is a new analysis that shows smaller error bounds. In particular, we use the following procedure to obtain frequency estimates $\tilde{f}_1,\ldots,\tilde{f}_\ell$, which are rough estimations of $f_1,\ldots,f_k$ in $x$.
\begin{lemma}\label{lem:find_heavy_frequencies}
    Let $L:=\{\tilde{f}_1,\ldots,\tilde{f}_\ell \}$ be the list of frequencies output by Procedure~\textsc{FrequencyEstimationX} in Algorithm 3 of \cite{SSWZ23} with input signal $H_k \cdot y$ and $\Delta:=\Delta_k$ instead of $k \cdot \Delta_k$ as the length of the frequency interval. With probability $0.99$, $L$ satisfies the following properties:
    \begin{enumerate}
        \item $\ell = O(k/\epsilon)$;
        \item for any $f$ with $\int_{f-\Delta_k}^{f+\Delta_k} |\wh{H_k \cdot y}(u)|^2 \mathrm{d} u \ge \frac{\epsilon}{5k} \cdot \|y\|_{[-1,1]}^2$, $\exists \tilde{f} \in L$ such that $|\tilde{f}-f| = O(\Delta_k)$.
    \end{enumerate}

    Moreover, this procedure takes $O(\frac{k^2 (\log k/\epsilon)^2 (\log F) \log (k\log F)}{\epsilon})$ samples and $O(\frac{k^2 (\log k/\epsilon)^2 (\log F)^2}{\epsilon})$ time.
\end{lemma}
We remark that Procedure~\textsc{FrequencyEstimationX} in Algorithm 3 of \cite{SSWZ23} (including Lemma~L.1 and the analysis in Appendix K of \cite{SSWZ23}) can choose any $\Delta \ge \Delta_k$ and output a list $L$ of $\tilde{f}$ such that $\exists \tilde{f} \in L$ with $|\tilde{f}-f| = O(\Delta)$ as long as
\begin{equation}\label{eq:condition_freq_recover}
 \text{ $f$ satisfies } \int_{f-\Delta}^{f+\Delta} |\wh{H_k \cdot y}(u)|^2 \mathrm{d} u \ge \frac{\epsilon}{5k} \cdot \|y\|_{[-1,1]}^2.   
\end{equation}
However, it is highly non-trivial to guarantee condition \eqref{eq:condition_freq_recover} for arbitrary frequencies $f_1,\ldots,f_k$.

The analysis in \cite{SSWZ23} shows that $\Delta:=k \cdot 2\Delta_k=\tilde{O}(k^3)$ guarantees that \eqref{eq:condition_freq_recover} holds for most frequencies in $x$. We explain their choice $\Delta:=k \cdot 2\Delta_k$ as follows. For convenience, we call any pair of frequencies $f_i$ and $f_j$ in $x$ correlated if $|f_i-f_j| \le 2 \Delta_k$ because their Fourier spectra of $H_k \cdot (\alpha_i e^{2 \pi \bi f_i t})$ and $H_k \cdot (\alpha_j e^{2 \pi \bi f_j t})$ have a non-empty intersection. The analysis in \cite{SSWZ23} (including \cite{CKPS17}) partitions all frequencies into clusters by correlations --- $f_i$ and $f_j$ are in the same cluster when they are correlated that is, $|f_i-f_j| \le 2\Delta_k$. Then $\Delta:=k \cdot 2\Delta_k$ is an upper bound  on the length of a cluster because there are at most $k$ frequencies. Next, the analyses in \cite{CKPS17,SSWZ23} show that most clusters have a frequency satisfying Condition~\eqref{eq:condition_freq_recover} (under the noise) because these clusters have disjoint Fourier support (after convolution with $H_k$). On the other hand, counterexamples in \cite{CP19_ICALP} showed that $\Delta$ must be $\tilde{\Omega}(k^2)$.


\paragraph{Our approach.} We discuss how to get a smaller error $\tilde{O}(k^{2.75})$ instead of $\tilde{O}(k^3)$ for the frequency estimates in $L$ here. Together with the improved frequency net $\mathcal{N}=[-F,F] \cap \frac{\mathbb{Z}}{C \cdot k^2}$, this leads to a better learning algorithm in Theorem~\ref{inform:fast_algorithm}.

Our first observation is that for two signals $w(t)$ and $z(t)$ whose Fourier sparsities are $\ell$ and $r$, if $\ell$ and $r$ are small (compared to $k$), $H_k \cdot w$ and $H_k \cdot z$ are almost orthogonal even if they have a large intersection in the Fourier domain.

\begin{claim} \label{clm:almost_orthogonal_clusters}
For two signals of Fourier sparsity $\ell$ and $r$ respectively ($\ell \le r\le k$),
\begin{align*}
    w(t):=\sum_{j=1}^{\ell}\alpha_j e^{2\pi\bi f'_jt} \qquad \text{ and }
    \qquad
    z(t):=\sum_{j=1}^{r}\beta_j e^{2\pi\bi f_jt},
\end{align*}
if the distance between their frequencies $\min_{j,j'} |f_j-f'_{j'}| \ge \min\left\{ C_H\frac{\ell^2 \cdot (r + \log 1/\delta) \cdot \log^2 k}{\delta^2}, 2\Delta_k \right\}$ 
for some constant $C_H$, then
\begin{align*}
    |\langle H_k \cdot w,H_k \cdot z\rangle|
    \le \delta \cdot \|H_k \cdot w\|_2 \cdot \|H_k \cdot z\|_2
    \qquad \text{ and } \qquad
    |\langle w,z\rangle_{[-1,1]}|
    \le \delta \cdot \|w\|_{[-1,1]} \cdot \|z\|_{[-1,1]} .
\end{align*}
\end{claim}
If we set $\ell=k^{1/4}$ and $\delta=\epsilon$, this implies that $H_k \cdot w$  and $H_k \cdot z$ are almost orthogonal when their frequencies are separated by $\tilde{\Omega}(k^{1.5})$. However, $\Delta_k=\tilde{\Theta}(k^2)$ implies that $\wh{H_k \cdot w}$ and $\wh{H_k \cdot z}$ have a large intersection among their Fourier supports. The proof of Claim~\ref{clm:almost_orthogonal_clusters} is a modification of the proof of Lemma~\ref{lemm:construction_H} (essentially with different parameters), which is deferred to Appendix~\ref{sec:proof_clm_almost_orthogonal}.

Now we propose Algorithm~\ref{alg:cluster1} to partition $(f_1,\alpha_1),\ldots,(f_k,\alpha_k)$ in the support of $\wh{x}$ into clusters. 
Here are some \emph{definitions of clusters}. For a cluster $\cC=\{(f_1,a_1),\ldots,(f_\ell,a_\ell)\}$ in $\wh{x}$, let $x_{\cC}(t):=\sum_{j = 1}^{\ell} a_j e^{2\pi \bi f_j t}$.  Then we define $|\cC|:=\ell$ and call it the size (the number of frequencies) of $\cC$.
For two clusters, let $\operatorname{dist}(\cC,\cC') :=\underset{(f,a)\in\cC,\,(f',a')\in\cC'}{\min} |f-f'|$. 

We remark that Algorithm~\ref{alg:cluster1} is only used in the analysis of frequency estimates produced by Lemma~\ref{lem:find_heavy_frequencies}, because $f_1,\ldots,f_k$ are unknown. The goal of this algorithm is to partition $f_1,\ldots,f_k$ into as many clusters as possible while ensuring that any two clusters are almost orthogonal. So the distance threshold in the while loop of Algorithm~\ref{alg:cluster1}, $\min\bigg\{ d_{min} \cdot \min\{|\cC_i|^2,|\cC_j|^2\}, 2\Delta_k \bigg\}$ for $d_{min}:=\frac{2 C_H k^{1.5} \log^3 (k/\epsilon)}{\epsilon^2}$, is a relaxation of the distance $\min\left\{ C_H\frac{\ell^2 \cdot (r + \log 1/\delta) \cdot \log^2 k}{\delta^2}, 2\Delta_k \right\}$ in Claim~\ref{clm:almost_orthogonal_clusters}. The factor $k^{1.5}/\epsilon^2$ in $d_{min}$ comes from the facts that (1) the larger cluster of $\cC_i$ and $\cC_j$ may have size $\Omega(k)$ and (2) the correlation $\delta$ (in Claim~\ref{clm:almost_orthogonal_clusters}) needs to be less than $\epsilon/k^{1/4}$ for our proof.

\begin{algorithm} 
    \caption{Partition Frequencies into Clusters \label{alg:cluster1}} 
    \begin{algorithmic}
        \Procedure{}{frequencies $f_1,\ldots,f_k$ with amplitudes $\alpha_1,\ldots,\alpha_k$}.      
        \State Define $k$ clusters $\cC_i:=\{(f_i,\alpha_i)\}$ and $d_{min}:=\frac{2 C_H k^{1.5} \log^3 (k/\epsilon)}{\epsilon^2}$
        
        \While{ $\exists~\cC_i$ and $\cC_j$ such that $       
        \operatorname{dist}(\cC_i,\cC_j) \le \min\bigg\{ d_{min} \cdot \min\{|\cC_i|^2,|\cC_j|^2\}, 2\Delta_k \bigg\}
        $}
        \State  merge all clusters whose frequencies lie between $\cC_i$ and
    $\cC_j$ into one cluster
        \EndWhile
        \State Return all remaining clusters $\cC$
    \EndProcedure
    \end{algorithmic}
\end{algorithm}

Let $\cC_1,\ldots,\cC_n$ be the remaining clusters of Algorithm~\ref{alg:cluster1}, ordered by
their frequencies. By Algorithm~\ref{alg:cluster1}, the distance between any two different clusters $\cC_i$ and $\cC_j$ either satisfies $\dist(\cC_i,\cC_j) > 2 \Delta_k$ in which case they are orthogonal or lies in $(d_{min} \cdot \min\{|\cC_i|^2,|\cC_j|^2\},2 \Delta_k]$. In the second case, we call them correlated as in previous works \cite{CKPS17,SSWZ23}. Equivalently, $\cC_i$ and $\cC_j$ are correlated only if $\dist(\cC_i,\cC_j) \le 2\Delta_k$; otherwise, their Fourier supports (after the convolution with $\wh{H_k}$) are disjoint. For two correlated clusters $\cC_i$ and $\cC_j$, Claim~\ref{clm:almost_orthogonal_clusters} implies that 
\begin{equation}\label{eq:orthogonal}
|\langle H_k x_{\cC_i},H_k x_{\cC_j}\rangle| \le \frac{\epsilon \cdot \sqrt{\max\{|\cC_i|,|\cC_j|\}}}{k^{3/4}} \cdot \|H_k x_{\cC_i}\|_2 \cdot \|H_k x_{\cC_j}\|_2 \le \frac{\epsilon}{k^{1/4}} \cdot \|H_k x_{\cC_i}\|_2 \cdot \|H_k x_{\cC_j}\|_2.    \end{equation}

At the same time, each cluster is correlated with at most $2\frac{\Delta_k}{d_{min}} = \tilde{O}
(k^{0.5})$ clusters. While the correlation coefficient $\frac{\epsilon}{k^{1/4}}$ in \eqref{eq:orthogonal} is $\tilde{\omega}(1/k^{0.5})$, our key technical result (Theorem~\ref{thm:total_energy} in Section~\ref{sec:heavy_region}) shows that the clusters generated by Algorithm~\ref{alg:cluster1} satisfy
\begin{equation}\label{eq:sum_energies}
    \forall S \subseteq [n], \sum_{j \in S} \|H_k \cdot x_{\cC_j}\|_2^2 \approx (1\pm \epsilon) \cdot \|\sum_{j \in S} H_k \cdot x_{\cC_j}\|_2^2.
\end{equation}
In particular, for $S=[n]$, this implies that the energy contributed by every cluster $H_k \cdot x_{\cC_j}$ in $H_k \cdot x$ is about $\|H_k \cdot x_{\cC_j}\|_2^2$. 

Then we generalize the definition of heavy frequencies in \cite{CKPS17} to heavy clusters. We say that a cluster $\cC_i$ is heavy if and only if $\|H_k \cdot \cC_i\|_2^2 \ge \frac{\epsilon \cdot |\cC_i|}{k} \cdot \|H_k \cdot x\|_2^2$. \eqref{eq:sum_energies} implies that the total energy of light clusters is $\epsilon \cdot \|H_k x\|_2^2$. Therefore, it is safe to neglect light clusters and focus on heavy clusters.

Next, we show that the covering radius of $L$ is at most $\tilde{O}(k^{2.75})$ (see Theorem~\ref{thm:covering_radius_heavy_region} in Section~\ref{sec:heavy_region} for a formal statement) --- for most heavy clusters $\cC_i$, there exists $\tilde{f}$ in $L$ (output by Lemma~\ref{lem:find_heavy_frequencies}) with $\min_{f \in \cC_i} |f-\tilde{f}|=\tilde{O}(k^{2.75})$. Because we can not guarantee the estimation error of every $f_i$ under adversarial noise, it is more precise to use covering radius in the rest of this work. We \emph{define} $\operatorname{range}(\cC):= [\underset{(f,a)\in\cC}{\min} f,\underset{(f,a)\in\cC}{\max} f]$ and $\operatorname{range}(\cC)\pm\Delta := [\underset{(f,a)\in\cC}{\min} f-\Delta,\underset{(f,a)\in\cC}{\max} f+\Delta]$. 
If a heavy cluster $\cC$ has $|\cC| \ge \sqrt{k}$, one can extend the proof of \eqref{eq:sum_energies} to show
\begin{equation}\label{eq:large_clusters_energy}
\int_{\operatorname{range}(\cC)\pm\Delta_k} |\wh{H_k \cdot x}(f)|^2 \mathrm{d} f = \| \wh{H_k \cdot x} \|_{\operatorname{range}(\cC)\pm\Delta_k}^2 \ge (1-O(\epsilon)) \cdot \|H_k \cdot x_{\cC}\|_2^2.    
\end{equation}

In the noiseless setting, an averaging argument shows that at least one frequency of $\cC$ satisfies the condition in \eqref{eq:condition_freq_recover}. This is because $\|H \cdot x_{\cC}\|_2^2$ on the right-hand side of \eqref{eq:large_clusters_energy} is at least $\frac{\epsilon \cdot |\cC|}{k} \cdot \|H x\|_2^2$ and there are at most $|\cC|$ frequencies. So the length of $\operatorname{range}(\cC)$ plus $2 \Delta_k$ provides an upper bound on the covering radius. On the other hand, we can prove $\operatorname{range}(\cC) = \tilde{O}(k^{2.75})$ based on the merging condition in Algorithm~\ref{alg:cluster1}. Roughly speaking, $\tilde{O}(k^{2.75})$ comes from $k$ frequencies forming $k^{3/4}$ groups of size $k^{1/4}$ with distance $\Delta_k$ between any two adjacent groups. This provides an estimate of the covering radius in the noiseless setting. 

However, we need a finer estimate (see Corollary~\ref{cor:localized_theta_small_clusters}) of $\int_{\operatorname{range}(\cC)\pm\Delta_k} |\wh{H_k \cdot x}(f)|^2 \mathrm{d} f$ for adversarial noise $\eta$ with a bounded $\ell_2$ norm $\|\eta\|_{[-1,1]}^2 \le \epsilon \cdot \|x\|_{[-1,1]}^2$. In Section~\ref{sec:heavy_region}, we provide a formal proof on the covering radius under noise. In Section~\ref{sec:heavy-region-clustering-conjecture}, assuming Conjecture~\ref{conj:growth_rate}, we show the improvements to Claim~\ref{clm:almost_orthogonal_clusters} and Algorithm~\ref{alg:cluster1} and a better covering radius $\tilde{O}(k^2)$ underlying Theorem~\ref{inform:algorithm_conjecture}.

\section{Net of Frequencies}\label{sec:net_freqs}
We present the proof of Lemma~\ref{lem:shift_one_freq} in Section~\ref{sec:proof_shifting} and discuss two important corollaries of this result. A direct corollary of Lemma~\ref{lem:shift_one_freq} provides a net on ``off-grid" frequencies.


\begin{theorem}\label{thm:net_frequency}
Given any $\epsilon$, let $\mathcal{N}:=\frac{\epsilon}{C k^2} \cdot\mathbb{Z} \cap [-F,F]$ be the net of frequencies for a large constant $C$. For any $x(t):=\sum_{j=1}^k \alpha_j e^{2 \pi \bi f_j t}$ with $k$ arbitrary frequencies in $[-F,F]$, there exists $x'(t):=\sum_{j=1}^k \alpha'_j e^{2 \pi \bi f'_j t}$ whose frequencies $f'_1,\ldots,f'_k$ are in $\mathcal{N}$ such that
\[
\|x'-x\|_{[-1,1]} \le \epsilon \cdot \|x\|_{[-1,1]}.
\]
\end{theorem}

This net bounds the query complexity of learning $x'$ and $x$ because the total number of possible $k$ frequencies is $|\mathcal{N}|^k=(kF/\epsilon)^{O(k)}$.
\begin{corollary}\label{cor:query_complexity_learning}
    Given any $F$ and $\epsilon$, let $y(t):=x(t)+\eta(t)$ for $x(t):=\sum_{j=1}^k \alpha_j e^{2 \pi \bi f_j t}$ be our observation over $[-1,1]$ with $k$ arbitrary frequencies $f_1,\ldots,f_k \in [-F,F]$ and $\|\eta(t)\|_{[-1,1]}^2 \le \epsilon \cdot \|x(t)\|_{[-1,1]}^2$. There exists an algorithm that takes $O(k^2 \cdot \log k \cdot \log \frac{kF}{\epsilon})$ samples and $(\frac{kF}{\epsilon})^{O(k)}$ time to output a $k$-Fourier-sparse signal $\tilde{x}$ such that with probability 0.99,
    \[
    \|\tilde{x}-x\|_{[-1,1]}^2 = O(\epsilon \cdot \|x\|_{[-1,1]}^2 + \|\eta\|_{[-1,1]}^2).
    \]
\end{corollary}
Because the proofs of Theorem~\ref{thm:net_frequency} and Corollary~\ref{cor:query_complexity_learning} follow the same outline of Lemma 2.1 in \cite{CKPS17} and  Corollary~9.7 in \cite{CP19_colt} separately, we defer them to Appendix~\ref{sec:proofs}.

The second application of Theorem~\ref{thm:net_frequency} is an approximation of $x$ based on frequency estimations and low-degree expansions, which is the foundation of our efficient recovery algorithms in Theorem~\ref{inform:fast_algorithm} and Theorem~\ref{inform:algorithm_conjecture}. Plugging the frequency gap $\frac{\epsilon}{k^2}$ of $x'$ in Theorem~\ref{thm:net_frequency} to Lemma 8.7 in \cite{CKPS17}, we have the following approximation of $x'$.

\begin{lemma}\label{lem:low_deg_approx}
    Let $x(t):=\sum_{j=1}^k \alpha_j e^{2 \pi \bi f_j t}$ and $L:=\{\tilde{f}_1,\ldots,\tilde{f}_\ell\}$ such that the covering radius of $L$ is $\Delta$: $\forall f_i, \min_{\tilde{f}_j \in L} |f_i - \tilde{f}_j| \le \Delta$.
    Then there exist polynomials $q_1,\ldots,q_{\ell}$ of degree $D:=O(\Delta+k^2 \log k/\epsilon)$ such that
    $\|x - \sum_{j=1}^\ell e^{2 \pi \bi \tilde{f}_j t} \cdot q_j(t)\|_{[-1,1]} \le 2\epsilon \cdot \|x\|_{[-1,1]}$.
\end{lemma}
We remark that previous work \cite{CKPS17} showed a net of frequency gap $k^{-O(k^2)}$ such that the degree $D=O(\Delta+k^3 \log k/\epsilon)$. Theorem~\ref{thm:net_frequency} improves the second term from $k^3 \log k/\epsilon$ to $k^2 \log k/\epsilon$. Theorem~\ref{thm:covering_radius_heavy_region} in Section~\ref{sec:heavy_region} and Theorem~\ref{thm:covering_radius_conjecture} will show smaller covering radii $\Delta$. In Section~\ref{sec:main_proof}, we will finish the proof of Theorem~\ref{inform:fast_algorithm} and Theorem~\ref{inform:algorithm_conjecture} .

\subsection{Proof of Lemma~\ref{lem:shift_one_freq}}
\label{sec:proof_shifting}




Following the proof of \cite{CKPS17}, let $v_j(t) = e^{2\pi\bi f_jt}$ and $V_i=\mathrm{span}\{v_1,\ldots,v_i\}$, with $V_0=\{\vec{0}\}$.
We define $v_k^{\parallel}$ and $v_{k+1}^{\parallel}$ to be the  projections of $v_k$ and $v_{k+1}$ onto $V_{k-1}$. Then $v_k^{\perp}:=v_k-v_k^{\parallel}$ and $v_{k+1}^{\perp}:=v_{k+1}-v_{k+1}^{\parallel}$ are the orthogonal parts.



Furthermore, let $w$ be the component of $v_{k+1}^{\perp}$ orthogonal to $v_k^{\perp}$ such that $v_{k+1}^{\perp} - w$ is parallel to $v_k^{\perp}$.  
Because both $v_k^{\perp}$ and $v_{k+1}^{\perp}$ lie in $V_{k-1}^{\perp}$, $w$ is orthogonal to both $V_{k-1}$ and $v_k^{\perp}$. So it is the component of $v_{k+1}$ orthogonal to $V_k=V_{k-1}+\mathrm{span}\{v_k^{\perp}\}$.

As shown in Lemma 8.5 of \cite{CKPS17}, the relative error of replacing $f_k$ by $f_{k+1}$ is at most $\frac{\|w\|_{[-1,1]}^2}{\|v_{k+1}^{\bot}\|^2_{[-1,1]}}$. Explicitly,
\begin{align}
w = v_{k+1}^{\perp} -\frac{\langle v_{k+1}^{\perp},v_k^{\perp}\rangle_{[-1,1]}} {\|v_k^{\perp}\|_{[-1,1]}^2}v_k^{\perp} \quad \text{and} \quad \|w\|_{[-1,1]}^2 = \|v_{k+1}^{\perp}\|_{[-1,1]}^2 -\frac{|\langle v_k^{\perp},v_{k+1}^{\perp}\rangle_{[-1,1]}|^2} {\|v_k^{\perp}\|_{[-1,1]}^2 }. \label{eq:w-orthogonality-bound}
\end{align}
Observe that for $b:=\frac{|\langle v_k^{\bot},v_{k+1}^{\bot} \rangle|}{\|v_{k+1}^\bot\|^2_{[-1,1]}}$,
\[\min_{z \in \mathbb C} \|v_k^{\perp}-z v_{k+1}^{\perp}\|^2_{[-1,1]}
=\|v_k^{\perp}-b v_{k+1}^{\perp}\|^2_{[-1,1]}
=\|v_k^{\perp}\|_{[-1,1]}^2 - \frac{|\langle v_k^{\bot},v_{k+1}^{\bot} \rangle|^2}{\|v_{k+1}^\bot\|^2_{[-1,1]}} 
.\] So we set $x'= \sum_{j=1}^{k-1} \alpha'_j v_j + \alpha_k b v_{k+1}^{\perp}$ to approximate $x= \sum_{j=1}^k \alpha_j v_j$, which replaces $f_k$ by $f_{k+1}$ and keeps the component in $V_{k-1}$ the same. Hence,
\begin{align}
\|x'-x\|_{[-1,1]}^2 = |\alpha_k|^2 \cdot \|v_k^{\perp}-b v_{k+1}^{\perp}\|^2_{[-1,1]}
\le \frac{\|v_k^{\perp}-b v_{k+1}^{\perp}\|^2_{[-1,1]}}{\|v_k^{\perp}\|_{[-1,1]}^2} \cdot \|x\|^2_{[-1,1]}, \label{eq:one-shift-first-bound}
\end{align}
where we use that $|\alpha_k|\|v_k^{\perp}\|_{[-1,1]}\le\|x\|_{[-1,1]}$ in the last step. Moreover, by the guarantee of $b$,
\begin{align}
\frac{\|v_k^{\perp}-b v_{k+1}^{\perp}\|_{[-1,1]}^2}{\|v_k^{\bot}\|_{[-1,1]}^2}
= 1  - \frac{|\langle v_k^{\perp},v_{k+1}^{\perp}\rangle_{[-1,1]}|^2} {\|v_k^{\perp}\|_{[-1,1]}^2 \cdot \|v_{k+1}^{\perp}\|_{[-1,1]}^2} = \frac{\|w\|_{[-1,1]}^2}{\|v_{k+1}^\bot\|_{[-1,1]}^2}. \label{eq:b-orthogonality-bound}
\end{align}


Different than Lemma 8.5 in \cite{CKPS17}, the rest of this proof provides a new bound on the quotient $\|w\|^2_{[-1,1]} / \|v_{k+1}^{\perp}\|^2_{[-1,1]}$, based on the Poincare inequality via an integral operator. From now on, we fix $f_1,\ldots,f_k$ and $f_{k+1}$. For any $g:[-1,1] \to \mathbb{C}$ and $t \in [-1,1]$, we define 
\begin{equation}\label{eq:def_operator_S}
(Sg)(t)=e^{2\pi\bi f_kt} \int_0^t e^{-2\pi\bi f_ks}g(s)\,\mathrm{d}s. 
\end{equation}
For any $j \neq k$ with $v_j(t)=e^{2 \pi \bi f_j t}$,
\begin{equation}
S v_j(t) = \frac{e^{2 \pi \bi f_j t} - e^{2 \pi \bi f_k t} }{2 \pi \bi(f_j-f_k)}=(v_j-v_k)/(2\pi\bi(f_j-f_k)). \label{eq:S_v_i}
\end{equation}
So $S(u)\subseteq V_k$ for any $u \in V_{k-1}$.

Since \eqref{eq:S_v_i} also holds for $v_{k+1}$, rearranging \eqref{eq:S_v_i} for $v_{k+1}$ shows
\begin{align}
    v_{k+1} = & 2 \pi \bi (f_{k+1} - f_k) \cdot S v_{k+1} +  v_k \notag \\
    = & 2 \pi \bi (f_{k+1} - f_k) \cdot S v_{k+1}^{\perp} + 2 \pi \bi (f_{k+1} - f_k) \cdot S v_{k+1}^{\parallel} + v_k \notag \\
    = & 2 \pi \bi (f_{k+1} - f_k) \cdot ( S v_{k+1}^{\perp} - r \cdot v_k) +  \underbrace{2 \pi \bi (f_{k+1} - f_k) \cdot S v_{k+1}^{\parallel} + [ 2 \pi \bi (f_{k+1} - f_k) r + 1 ] \cdot v_k}_{\in V_k}, \label{eq:decom_v_k_1}
\end{align}
where $r$ is a parameter chosen later.

Because $v_{k+1}- \underbrace{2 \pi \bi (f_{k+1} - f_k) \cdot ( S v_{k+1}^{\perp} - r \cdot v_k)}_{\text{ first part in \eqref{eq:decom_v_k_1}}} \in V_k$ by \eqref{eq:decom_v_k_1}, $w$ as the component of $v_{k+1}$ orthogonal to $V_k$ satisfies
\begin{equation}
    \| w \|_{[-1,1]} \le \| 2 \pi \bi (f_{k+1} - f_k) \cdot ( S v_{k+1}^{\perp} - r \cdot v_k) \|_{[-1,1]} \label{eq:w-orthogonality-bound2}.
\end{equation}
For convenience, let $(Tg)(t):=\int_0^t e^{-2\pi\bi f_ks}g(s)\,\mathrm{d}s$ be the integral opeartor in \eqref{eq:def_operator_S} such that \eqref{eq:w-orthogonality-bound2} becomes
\[
\| w \|_{[-1,1]}  \le 2 \pi | f_{k+1} - f_k | \cdot \left\| T v_{k+1}^{\perp} - r \right\|_{[-1,1]}.
\]

Finally, we apply Poincar\'{e}'s inequality with a proper $r$.
\begin{lemma}[Poincar\'{e}'s inequality \cite{stein2011fourier}]
    If $f$ is continuously differentiable on $[-1, 1]$, 
    \begin{align*}
        \int_{-1}^{1} \left| f(t) - \frac{1}{2} \int_{-1}^{1} f(s) ds \right|^2 dt \le \frac{4}{\pi^2} \int_{-1}^{1} |f'(t)|^2 dt.
    \end{align*}
\end{lemma}

Applying Poincar\'{e} inequality with $f:=T v_{k+1}^{\perp}$ and $r := \frac{1}{2} \int_{-1}^{1} f(s) ds$, we have
\begin{align*}
    \left\| T v_{k+1}^{\perp} - r \right\|_{[-1,1]} \le \frac{2}{\pi} \| v_{k+1}^{\perp} \|_{[-1,1]}.
\end{align*}

Therefore, with \eqref{eq:w-orthogonality-bound2},
\begin{align*}
    \| w \|_{[-1,1]} \le 4 | f_{k+1} - f_k | \cdot \| v_{k+1}^{\perp} \|_{[-1,1]}.
\end{align*}

\section{Heavy Frequency Recovery}
\label{sec:heavy_region}

Our main result in this section provides a strong guarantee on the list of frequencies in Lemma~\ref{lem:find_heavy_frequencies}. Let  $\cC_1,\ldots,\cC_n$ denote clusters returned by Algorithm~\ref{alg:cluster1} in this section, while our recovery algorithms do not know the frequencies $f_1,\ldots,f_k$. We still use $H_k$ to denote the filter function $H_{k,\epsilon}$ constructed in Lemma~\ref{lemm:construction_H} with a support $\wh{H_k}=[-\Delta_k,\Delta_k]$ and recall $x_{\cC}(t):=\sum_{j = 1}^{\ell} a_j e^{2\pi \bi f_j t}$ for a cluster $\cC=\{(f_1,a_1),\ldots,(f_\ell,a_\ell)\}$ in $\wh{x}$.



\begin{theorem}\label{thm:covering_radius_heavy_region}
     For $y(t)=x(t)+\eta(t)$ with $x(t):=\sum_{j=1}^k \alpha_j e^{2 \pi \bi f_j t}$ and $\|\eta\|_{[-1,1]}^2 \le \epsilon \cdot \|x\|_{[-1,1]}^2$ for a fixed small constant $\epsilon$, let $L$ be the list of frequencies from Lemma~\ref{lem:find_heavy_frequencies} on $H_k \cdot y$. Then for covering radius $D:=\frac{k^{2.75}}{\epsilon^{1.5}} \cdot (\log k)^{O(1)}$, 
     \[
    \mathcal{R}:=\{\cC_i: \exists f \in \cC_i \text{ with } \min_{\tilde{f}_i \in L} |\tilde{f}_i - f| \le D \}\] covered by $L$ within the distance $D$ 
    satisfies $\|H_k \cdot (\sum_{\cC \in \mathcal{R}} x_{\cC})  - H_k x\|_{2}^2 = O(\epsilon) \cdot \|H_k x\|_2^2$.
    
\end{theorem}

We finish the proof of Theorem~\ref{thm:covering_radius_heavy_region} in this section. The key technical result of this proof is the following theorem, which shows that the clusters output by Algorithm~\ref{alg:cluster1} (in Section~\ref{sec:overview}) are almost orthogonal given Claim~\ref{clm:almost_orthogonal_clusters}(in Section~\ref{sec:overview}).

\begin{theorem} \label{thm:total_energy}
Let $d_{min}:=\frac{2 C_H k^{1.5} \log^3 (k/\epsilon)}{\epsilon^2}$ and $\cC_1,\ldots,\cC_n$ be $n$ clusters with $\dist(\cC_i,\cC_j) \ge  \min\bigg\{ d_{min} \cdot \min\{|\cC_i|^2,|\cC_j|^2\}, 2\Delta_k \bigg\}$ for any two $\cC_i$ and $\cC_j$.
For every $S\subseteq[n]$,
\begin{align*}
    \left\|H_k \cdot\sum_{i\in S}x_{\cC_i}\right\|_2^2
    = \left(1\pm O(\epsilon)\right) \sum_{i\in S}\|H_k \cdot x_{\cC_i}\|_2^2.
\end{align*}
In particular, $
    \|H_k \cdot x\|_2^2 
    = \left(1\pm O(\epsilon)\right) \sum_{i=1}^n\|H_k \cdot x_{\cC_i}\|_2^2 $.
\end{theorem}
Recall that a cluster $\cC_i$ is heavy iff $\|H_k \cdot x_{\cC_i}\|_2^2 \ge \frac{\epsilon \cdot |\cC_i|}{k} \cdot \|H_k x\|_2^2$. Theorem~\ref{thm:total_energy} implies that it is safe to neglect all light clusters with $\|H \cdot x_{\cC_j}\|_2^2 \le \frac{\epsilon \cdot |\cC_j|}{k} \cdot \|H \cdot x\|_2^2$ and focus on heavy clusters. This is because 
\begin{align*}
\| H_k \cdot \sum_{\text{heavy } \cC_i} x_{\cC_i}\|_2^2 & \ge (1-O(\epsilon)) \cdot \sum_{\text{heavy } \cC_i}  \|H_k \cdot x_{\cC_i}\|_2^2 \\
& = (1-O(\epsilon)) (\sum_{i=1}^n \|H_k \cdot x_{\cC_i}\|_2^2 - \sum_{\text{light } \cC_i} \|H_k \cdot x_{\cC_i}\|_2^2) \\
& = (1-O(\epsilon))^2 \|H_k \cdot x\|_2^2 - (1-O(\epsilon))\epsilon \cdot \|H_k \cdot x\|_2^2 = (1- O(\epsilon)) \|H_k \cdot x\|_2^2.
\end{align*}

We rewrite $H_k \cdot y=H_k \cdot (x + \eta) = H_k \cdot (\sum_{\text{heavy } \cC_i} x_{\cC_i}) + H_k \cdot \eta + H_k \cdot (\sum_{\text{light } \cC_i} x_{\cC_i})$. In the rest of this section, we reset $x:=\sum_{\text{heavy } \cC_i} x_{\cC_i}$ and consider the recovery of (heavy) clusters in $x$ under noise $\eta_H:=H_k \cdot \eta + H_k \cdot (\sum_{\text{light } \cC_i} x_{\cC_i})$ with $\|\eta_H\|_2^2 =O(\epsilon) \cdot \|H_k \cdot x\|_2^2$. 

The next observation is that for a (heavy) cluster $\cC_i$ with $|\cC_i|=O(k^{1/4})$, most of the energies of $\wh{H \cdot \cC_i}$ concentrate around $range(\cC_i) \pm \Delta$ for some $\Delta:=O(k^{1.5})$ much smaller than $\Delta_k$. Let $\theta_s:=\tilde{\Theta}(k^{1/4})$ be the largest integer with $\Delta_{\theta_s,\epsilon^2/k} < d_{min}/2$ (defined in Lemma~\ref{lemm:construction_H}). In the rest of this section, we call a cluster $\cC_i$ small iff $|\cC_i| \le \theta_s$; otherwise we call it large.

To be more precise, for any small cluster $\cC_i$ of size $\le \theta_s$, basic properties of $H_{\theta_s,\delta}$ and $H_k$ imply that $H_{\theta_s,\delta} \cdot x_{\cC_i} \approx H_k \cdot x_{\cC_i}$. By the Fourier transform, $\wh{H_{\theta_s,\delta} \cdot x_{\cC_i}} \approx \wh{H_k \cdot x_{\cC_i}}$ such that most energies of $\wh{H_k \cdot x_{\cC_i}}$ are concentrated in $range(\cC_i) \pm \Delta_{\theta_s,\delta}$. We refer to Claim~\ref{clm:cluster_concentration} for a formal statement. After choosing $\delta$ carefully, the following fact provide a good approximation on the Fourier spectrum of $H_k \cdot (\sum_{\text{small clusters}} x_{\cC_j})$ of small clusters. Since $\theta_s$ is fixed, we define $\Delta_{\theta}:=\Delta_{\theta_s,\epsilon^2/k}$, which is less than $d_{min}/2$.

\begin{corollary}
\label{cor:localized_theta_small_clusters}
Let $T \subseteq \{j:|\cC_j| \le \theta_s\}$ be a subset of small clusters and 
$H_\theta:=H_{\theta_s,\epsilon^2/k}$ with support $[-\Delta_\theta,\Delta_\theta]$. 
Then $\supp(\wh{H_\theta\cdot x_{\cC_j}})$ is contained in
$\operatorname{range}(\cC_j)\pm\Delta_\theta$ for every $j\in T$ and
\begin{align*}
    \left\| \sum_{j\in T}\wh{H_k \cdot x_{\cC_j}} - \sum_{j\in T}\wh{H_\theta\cdot x_{\cC_j}} \right\|_2
    \le \epsilon \left(\sum_{j\in T} \|\wh{H_k \cdot x_{\cC_j}}\|_2^2 \right)^{1/2}.
\end{align*}
\end{corollary}
Our proof relies on the following fact: small clusters have disjoint Fourier supports in $\wh{H_\theta\cdot x_{\cC_j}}$ by the definition of $\theta_s$. 

For large clusters, we use the following bound on the length of their ranges. For convenience, we state it for all possible sizes. 
\begin{claim}
\label{clm:bound_range_clusters}
From Algorithm~\ref{alg:cluster1}, a cluster with $\ell$ frequencies has a range of length at most (Recall $\Delta_k:=\tilde{O}(\frac{k^2}{\epsilon} )$ and $d_{min}:=\tilde{O}(\frac{k^{1.5}}{\epsilon^2})$)
\begin{align*}
    r_\ell \le
    \begin{cases}
  d_{min}\cdot O(\ell^2)  = \tilde{O}(\frac{k^{1.5}}{\epsilon^2} \cdot \ell^2),     &   \ell\le \sqrt{\frac{2\Delta_k}{d_{min}}},\\[1.2ex]
    d_{min}\cdot\sqrt{\frac{2\Delta_k}{d_{min}}} \cdot O(\ell) = \tilde{O}(\frac{k^{1.75}}{\epsilon^{1.5}} \cdot \ell),
        &       \ell >\sqrt{\frac{2\Delta_k}{d_{min}}}.
    \end{cases}
\end{align*}
Moreover, $\sum_{i \in [n]} |range(\cC_i)| = \tilde{O}(\frac{k^{2.75}}{\epsilon^{1.5}})$. 
\end{claim}

Now we are ready to finish the proof of Theorem~\ref{thm:covering_radius_heavy_region}.
The proofs of Theorem~\ref{thm:total_energy}, Corollary~\ref{cor:localized_theta_small_clusters}, and Claim~\ref{clm:bound_range_clusters} are deferred to Section~\ref{subsec:heavy_region_clustering}, Section~\ref{subsec:heavy_region_localized_filters}, and Section~\ref{subsec:heavy_region_range} separately. 


\begin{proofof}{Theorem~\ref{thm:covering_radius_heavy_region}}
First of all, we approximate $H_k \cdot x$ and $H_k \cdot y$ as follows. We split all clusters in $x$ into small ones of size $\le  \theta_s$ and larger ones of size $> \theta_s$. For $H_{\theta}$ defined in Corollary~\ref{cor:localized_theta_small_clusters} with Fourier support $[-\Delta_\theta,\Delta_\theta]$ of $\Delta_\theta:=C_H \cdot \frac{k^{3/2} \cdot \log^2 k}{2 \epsilon^2}$, let
\begin{equation}\label{eq:def_z}
    z(t):=\sum_{\cC_i:|\cC_i| \le \theta_s} H_{\theta}(t) \cdot x_{\cC_i}(t) + \sum_{\cC_j:|\cC_j| > \theta_s} H_{k}(t) \cdot x_{\cC_j}(t).
\end{equation}
By Corollary~\ref{cor:localized_theta_small_clusters}, $\|z-H_k \cdot x\|_2^2 \le O(\epsilon) \cdot \|H_k x\|_2^2$. Recall that $\eta_H:=H_k \cdot \eta + H_k \cdot (\sum_{\text{light } \cC_i} x_{\cC_i})$ with $\|\eta_H\|_2^2 =O(\epsilon) \cdot \|H \cdot x\|_2^2$. We consider $H_k \cdot y=z + (H_k \cdot x - z + \eta_H) = z + \eta'$ for noise $\eta':=H_k \cdot x - z + \eta_H$ with  $\|\eta'\|_2^2 = O(\epsilon) \cdot \|H_k \cdot x\|_2^2$.

Another useful property is that $\Delta_\theta < \dist(\cC_i,\cC_j)/2$ for any two clusters (from the definition of $\theta_s$ and $\Delta_{\theta}$) such that $\wh{H_{\theta} \cdot x_{\cC_i}}$ and $\wh{H_{\theta} \cdot x_{\cC_{i'}}}$ are disjoint for any two \emph{small} clusters $\cC_i$ and $\cC_{i'}$. However, a large cluster $\cC_j$ may have $\wh{H_k \cdot x_{\cC_j}}$ intersecting with small clusters. So we consider the following approach.

Now we define $n$ intervals to be the Fourier support of each $\cC_i$ in $z$ (defined in \eqref{eq:def_z}):
\begin{align*}
    I_i =
    \begin{cases}
    range(\cC_i) \pm \Delta_{ \theta},
        &|\cC_i| < \theta_s,\\[1.2ex]
    range(\cC_i) \pm \Delta_k,
        & |\cC_i| \ge \theta_s.
    \end{cases}
\end{align*}
Then we keep merging intervals as long as there exist $I_i$ and $I_j$ with $I_i \cap I_j \neq \emptyset$. For convenience, let $J_1,\ldots,J_m$ be the remaining disjoint intervals. For a cluster $\cC$ and interval $J_j$, we use $\cC \subset J_j$ to indicate that each frequency $f \in \cC$ satisfies $f \in J_j$ and $R_j:=\{i: \cC_i \subset J_j\}$ to denote the clusters in $J_j$.



By the definition of $J_a$ and $R_a$, $J_a$ is the union of supports of clusters in $R_a$ as
\begin{equation}\label{eq:def_I_a}
J_a:=\left( \cup_{i \in R_a: |\cC_i| < \theta_s} \supp(\wh{H_{\theta} \cdot x_{\cC_i}}) \right) \cup \left( \cup_{j \in R_a: |\cC_j| \ge \theta_s} \supp(\wh{H_{k} \cdot x_{\cC_j}}) \right). \end{equation}

We apply Theorem~\ref{thm:total_energy} to clusters in $R_a$: 
\begin{equation}\label{eq:sum_R_a_1}
     \|\sum_{i \in R_a} \wh{H_k \cdot x_{\cC_i}}\|_{2}^2 = (1 \pm O(\epsilon)) \cdot \sum_{i \in R_a} \|\wh{H_k \cdot x_{\cC_i}}\|_{2}^2.
\end{equation}
At the same time, the signal constituted by clusters in $R_a$ is $\sum_{i \in R_a:|\cC_i| < \theta_s} \wh{H_{\theta} \cdot x_{\cC_i}} + \sum_{j \in R_a:|\cC_j| \ge \theta_s} \wh{H_{k} \cdot x_{\cC_j}}$. 
Corollary~\ref{cor:localized_theta_small_clusters} bounds its difference to $\sum_{i \in R_a} \wh{H_k \cdot x_{\cC_i}}$ as
\begin{equation}\label{eq:sum_R_a_2}
     \left\|\sum_{i \in R_a:|\cC_i| < \theta_s} \wh{H_{\theta} \cdot x_{\cC_i}} + \sum_{j \in R_a:|\cC_j| \ge \theta_s} \wh{H_{k} \cdot x_{\cC_j}} - \sum_{i \in R_a} \wh{H_k \cdot x_{\cC_i}} \right\|_{2} \le \epsilon \cdot \left( \sum_{i \in R_a:|\cC_i| < \theta_s} \|\wh{H_k \cdot x_{\cC_i}}\|_{2}^2 \right)^{1/2}.
\end{equation}
Eq~\eqref{eq:sum_R_a_1} and Eq~\eqref{eq:sum_R_a_2} imply that
\begin{equation}\label{eq:energy_in_I_a}
\left\|\sum_{i \in R_a:|\cC_i| < \theta_s} \wh{H_{\theta} \cdot x_{\cC_i}} + \sum_{j \in R_a:|\cC_j| \ge \theta_s} \wh{H_{k} \cdot x_{\cC_j}} \right\|_2^2 = (1 \pm O(\epsilon)) \sum_{i \in R_a} \|\wh{H_k \cdot x_{\cC_i}}\|_{2}^2.
\end{equation}
Recall that $J_a$ defined in \eqref{eq:def_I_a} is the union of supports of $\wh{H_{\theta} \cdot x_{\cC_i}}$ over small clusters in $R_a$ and $\wh{H_{k} \cdot x_{\cC_j}}$ over large clusters in $R_a$. Because $J_1,\ldots,J_m$ are disjoint, $\wh{z}(f) \cdot \mathbf{1}_{J_a}(f)=\sum_{i \in R_a:|\cC_i| < \theta_s} \wh{H_{\theta} \cdot x_{\cC_i}}(f) + \sum_{j \in R_a:|\cC_j| \ge \theta_s} \wh{H_{k} \cdot x_{\cC_j}}(f)$ from the definition of $z$ in \eqref{eq:def_z}. 
\eqref{eq:energy_in_I_a} shows 
\begin{equation}\label{eq:energy_Z_J_a}
\|\wh{z}\|^2_{J_a} = (1 \pm O(\epsilon)) \cdot \sum_{i \in R_a} \|\wh{H_k \cdot x_{\cC_i}}\|_{2}^2.    
\end{equation}

Next, we consider $z+\eta'$ for noise  $\eta':=H_k \cdot x - z + \eta_H$ with  $\|\eta'\|_2^2 = O(\epsilon) \cdot \|x\|_{[-1,1]}^2$. We say that an interval $J_a$ in $z$ is bad if $\|\wh{\eta'}\|_{J_a}^2 \ge \|\wh{z}\|_{J_a}^2/16$; otherwise $J_a$ is good. 
 Since $J_1,\ldots,J_m$ are disjoint, 
 \[
 \|\eta'\|_2^2 \ge \sum_{a \in [m]:J_a \text{ is bad}} \|\wh{\eta'}\|^2_{J_a} \ge \sum_{a \in [m]:J_a \text{ is bad}} \|\wh{z}\|^2_{J_a}/16.
 \]
 This implies
 \begin{equation}\label{eq:good_clusters}
\sum_{a \in [m]:J_a \text{ is good}} \|\wh{z}\|^2_{J_a} \ge (1-O(\epsilon)) \cdot \sum_{i} \|\wh{H_k \cdot x_{\cC_i}}\|_{2}^2 - 16 \|\eta'\|_2^2 \ge (1-O(\epsilon)) \cdot \|\wh{H_k \cdot x}\|_{2}^2.      
 \end{equation}
So $\mathcal{R}:=\cup_{a: J_a \text{ is good}} R_a$ in this theorem. Then \eqref{eq:good_clusters} shows that the total energy of the clusters in $\mathcal{R}$ is at least $(1-O(\epsilon)) \cdot \|H_k x\|_2^2$. This indicates 
\[ 
\sum_{\cC \notin \mathcal{R}} \|H_k \cdot x_{\cC}\|_2^2 = O(\epsilon) \cdot \|H_k x\|_2^2  \text{ and } \|H_k \cdot (\sum_{\cC \notin \mathcal{R}} x_{\cC})\|_2^2 = O(\epsilon) \cdot \|H_k x\|_2^2
\] by applying Theorem~\ref{thm:total_energy} twice to $[n]$ and $[n] \setminus \mathcal{R}$ separately. 
Finally, $\|H_k \cdot (\sum_{\cC \in \mathcal{R}} x_{\cC})  - H_k x\|_{2}^2 = O(\epsilon) \cdot \|H_k x\|_2^2$ follows the above bound in $H_k \cdot (\sum_{\cC \notin \mathcal{R}} x_{\cC})$.
 
Finally, we bound the covering radius $D$. For each good $J_a$, there exists $f_a$ in some cluster of $R_a$ with $\int_{f_a-\Delta_k}^{f_a+\Delta_k} |\wh{H_k \cdot y}(f)|^2 \mathrm{d} f \ge \frac{\epsilon}{5k} \cdot \|x\|_{[-1,1]}^2$. This is because the number of frequencies in $J_a$ is  $(\sum_{i \in R_a} |\cC_i|)$ and 
\begin{align*}
\|H_k \cdot y\|^2_{J_a} & = \|\wh{z+\eta'}\|^2_{J_a}\\
& \ge \frac{1}{2} |\wh{z}\|_{J_a}^2 \tag{$J_a$ is good} \\
& \ge \frac{1-O(\epsilon)}{2} \sum_{i \in R_a} \|\wh{H_k \cdot x_{\cC_i}}\|_{2}^2 \tag{by \eqref{eq:energy_Z_J_a}}\\
& \ge \frac{1-O(\epsilon)}{2} \cdot \frac{\epsilon \cdot (\sum_{i \in R_a} |\cC_i|)}{k} \cdot \|H_k x\|_{[-1,1]}^2. \tag{by the definition of heavy clusters}\\    
& \ge \frac{1}{5} \cdot \frac{\epsilon \cdot (\sum_{i \in R_a} |\cC_i|)}{k} \cdot \|y\|_{[-1,1]}^2.
\end{align*}
 So the covering radius $D$ (of $\mathcal{R}$) is the length of $J_a$ plus $O(\Delta_k)$. The former is at most 
\begin{align*}
\sum_{i=1}^n |I_i| & = \sum_{i: |\cC_i|<\theta_s} (|range(\cC_i)| +  2 \Delta_{\theta}) + \sum_{i: |\cC_i| \ge \theta_s} (|range(\cC_i)| +  2 \Delta_{k}) \\
& = \sum_i |range(\cC_i)| + k \cdot 2 \Delta_{\theta} + 2 \Delta_k \cdot \frac{k}{\theta_s} = \tilde{O}(k^{2.75}/\epsilon^{1.5}+ k^{2.5}/\epsilon^2).    \tag{by Claim~\ref{clm:bound_range_clusters}} 
\end{align*}

From the discussion above, for $D=\frac{k^{2.75}}{\epsilon^1.5} \cdot (\log k)^{O(1)}$, for every cluster $\cC$ in $\mathcal{R}$, there exists $\tilde{f} \in L$ (output by Lemma~\ref{lem:find_heavy_frequencies}) such that $\max_{f \in \cC} |\tilde{f}-f| \le D$. 
\end{proofof}


\subsection{Proof of Theorem~\ref{thm:total_energy}}
\label{subsec:heavy_region_clustering}
We finish the proof of Theorem~\ref{thm:total_energy} in this section. We assume $\epsilon$ is a small constant such that $d_{min}<\Delta_k$.


Recall $\Delta_k:=C^2 (\frac{k^2 \log(k/\epsilon)}{\epsilon} + k^2 \log k \cdot \log (k^2 \log k) )\le \frac{3C^2k^{2}\log^2(k/\epsilon)}{\epsilon}$ and $d_{min}:=\frac{2 C_H k^{1.5} \log^3 (k/\epsilon)}{\epsilon^2}$. Therefore,

\begin{equation}
\label{eq:max D in correlation}
\frac{\Delta_k}{d_{min}}\le \frac{3C^2k^2\log^2(k/\epsilon)/\epsilon}{2C_Hk^{1.5}\log^3(k/\epsilon)/\epsilon^2}\le \frac{3C^2}{2C_H} k^{0.5}\epsilon=O(k^{0.5}\epsilon).
\end{equation}

In this proof, we assume that $\cC_1,\ldots,\cC_n$ are sorted by their frequencies.
For correlated clusters $\cC_i$ and $\cC_j$ with $\dist(\cC_i,\cC_j) \le 2 \Delta_k$, we define
\begin{align}
    D_{i,j}:= & \frac{k\cdot \dist(\cC_i,\cC_j)} {d_{min}}.
\end{align}

Because $D_{i,j}$ is defined only when $\cC_i$ and $\cC_j$ are correlated and $\dist(\cC_i,\cC_j) \ge d_{min} \cdot \min\{|\cC_i|^2,|\cC_j|^2\}$ in this case, we have the following bounds on $D_{i,j}$:

\begin{equation}
\label{eq:bound_D}
D_{i,j} \in \left[ k \cdot \min\{|\cC_i|^2,|\cC_j|^2\}, \frac{2k\Delta_k}{d_{min}} \right].
\end{equation}

\begin{claim}
\label{clm:satisbility_of_delta_construction}
For correlated $\cC_i$ and $\cC_j$, let 
\begin{equation}
    \label{eq:selection_on_delta}
    \delta_{i,j}:=
    \frac{\epsilon}{k^{1/4}}
    \sqrt{
    \frac{\min\{|\cC_i|^2,|\cC_j|^2\} \max\{|\cC_i|,|\cC_j|\}}
    {D_{i,j}}}
\end{equation}
denote the correlation of $\cC_i$ and $\cC_j$; and let $\delta_{i,j}:=0$ for uncorrelated $\cC_i$ and $\cC_j$. 
Then for any $i$ and $j$, $\delta_{i,j}$ satisfies the condition of Claim~\ref{clm:almost_orthogonal_clusters}: 
\begin{equation}
\label{eq:weighted_pair_correlation}
    |\langle H_k \cdot x_{\cC_i},H_k \cdot x_{\cC_j}\rangle|
    \le \delta_{i,j} \|H_k \cdot x_{\cC_i}\|_2\|H_k \cdot x_{\cC_j}\|_2 .
\end{equation}
\end{claim}

\begin{proof}

Applying Claim~\ref{clm:almost_orthogonal_clusters} to $\dist(\cC_i,\cC_j)$, it is clear that any $\delta$ satisfies the inequality below would meet the condition of Claim~\ref{clm:almost_orthogonal_clusters}.

\begin{equation}
\label{eq: limit_on_delta}
\dist(\cC_i,\cC_j)=D_{i,j}\cdot \frac{d_{min}}{k} \ge C_H\cdot \frac{\min\{|\cC_i|^2,|\cC_j|^2\}\big(\max\{|\cC_i|,|\cC_j|\}+\log(1/\delta)\big)}{\delta^2}\cdot \log^2 k.
\end{equation}

We show that $\delta_{i,j}$ defined in \eqref{eq:selection_on_delta} satisfies \eqref{eq: limit_on_delta}. We begin by bounding the $\log\left(1/\delta_{i,j}\right)$ term in \eqref{eq: limit_on_delta} as follows: 

\begin{align}
\log\left(\frac{1}{\delta_{i,j}}\right)
&=\log\left(\frac{k^{1/4}}{\epsilon}\sqrt{\frac{D_{i,j}}{\min\{|\cC_i|^2,|\cC_j|^2\}\max\{|\cC_i|,|\cC_j|\}}}\right)\tag{by \eqref{eq:selection_on_delta}}\nonumber\\
&\le \log\left(\frac{k^{1/4}}{\epsilon}\sqrt{\frac{2k\Delta_k}{d_{min}}}\right)\tag{by the upper bound of $D_{i,j}$ in \eqref{eq:bound_D} and denominator $\ge 1$}\nonumber\\
&\le \log\left(\frac{k^{1/4}}{\epsilon}\sqrt{\frac{3C^2 k^{1.5}\epsilon}{C_H}}\right)\tag{by \eqref{eq:max D in correlation}}\nonumber\\ 
&=\log\left(\sqrt{\frac{3C^2}{C_H}}\cdot \frac{k}{\epsilon^{0.5}}\right)\nonumber\\
&\le \log \frac{k^2}{e\epsilon^2}\tag{suppose $\sqrt{\frac{3C^2}{C_H}}\le \frac{k}{e\epsilon^{1.5}}$}\nonumber\\
&=2 \log(k/\epsilon)-1. \label{eq:upper_bound_partial_term} 
\end{align}

Now we are ready to finish the proof of \eqref{eq: limit_on_delta}.

\begin{align*}
&C_H\cdot \frac{\min\{|\cC_i|^2,|\cC_j|^2\}\big(\max\{|\cC_i|,|\cC_j|\}+\log(1/\delta_{i,j})\big)}{\delta_{i,j}^2}\cdot \log^2 k\\
=&C_H\frac{\min\{|\cC_i|^2,|\cC_j|^2\}\max\{|\cC_i|,|\cC_j|\}}{\delta_{i,j}^2}\log^2 k\cdot \left(1+\frac{\log(1/\delta_{i,j})}{\max\{|\cC_i|,|\cC_j|\}}\right)\\
\le& C_H \frac{k^{0.5}D_{i,j}}{\epsilon^2}\cdot\log^2 k \left(1+\log\frac{1}{\delta_{i,j}}\right)\tag{plug definition of $\delta_{i,j}$ and denominator $\ge 1$}\\
\le&\frac{d_{min}}{2k\log k/\epsilon}D_{i,j}\cdot 2\log  (k/\epsilon)\tag{because $d_{min}:=\frac{2 C_H k^{1.5} \log^3 (k/\epsilon)}{\epsilon^2}$ and \eqref{eq:upper_bound_partial_term} on $\log(1/\delta_{i,j})$}\\
=&\frac{d_{min}}{k} D_{i,j}.
\end{align*}

\end{proof}

Our plan is to show that $\sum_{j: \text{ correlated with } i} \delta_{i,j}=O(\epsilon)$ for any $i$ in this section. We first bound the total number of frequencies in the correlated clusters.

\begin{lemma}
\label{lem:bounding_on_D_weighted}
For two correlated clusters $\cC_x$ and $\cC_y$ with indices $x<y$, 
\begin{align*}
D_{x,y}\ge\frac{k}{2} \left(\sum_{a=x}^{y}|\cC_a| - \max_{x\le a\le y}|\cC_a|\right).
\end{align*}
\end{lemma}

\begin{proof}
For any threshold of size $t \ge 1$, let $\cC_{i_1},\ldots,\cC_{i_p}$ be the clusters in $\{\cC_x,\ldots,\cC_y\}$ whose size is at least $t$.
As \eqref{eq:bound_D} shows, 
all $D_{i_1,i_2},\ldots,D_{i_{p-1},i_p}$ are larger than $kt^2$. 
So $D_{x,y}\ge (p - 1) k t^2$.
Summing this bound over thresholds $t$ shows
\begin{align*}
\sum_{a=x}^{y}|\cC_a|
&=\sum_{t=1}^{\max_{x\le a\le y}|\cC_a|} \sum_{j = x}^y \mathbb{I}[|\cC_j| \ge t] \tag{$|\cC_j| = \sum_{t=1}^{\max_{x\le a\le y}|\cC_a|} \mathbb{I}[|\cC_a| \ge t]$}\\
& \le \sum_{t=1}^{\max_{x\le a\le y}|\cC_a|} \left(\frac{D_{x,y}}{kt^2}+1\right) \tag{$\sum_{j = x}^y \mathbb{I}[|\cC_j| \ge t] \le \frac{D_{x,y}}{kt^2}+1$}\\
& \le \frac{2D_{x,y}}{k} + \max_{x\le a\le y}|\cC_a|.
\end{align*}

\end{proof}

One more step towards bounding $\sum_{j: \text{ correlated with } i} \delta_{i,j}$ is to bound the summation of \\ $\sqrt{
    \frac{\min\{|\cC_i|^2,|\cC_j|^2\} \max\{|\cC_i|,|\cC_j|\}}
    {D_{i,j}}}$ in \eqref{eq:selection_on_delta} (the definition of $\delta_{i,j}$).
\begin{lemma}
\label{lem:one_sided_weight_sum_weighted}
Let $i < q$. Suppose that $\cC_i$ and $\cC_q$ are correlated and $|\cC_i|=\max_{i\le a\le q}|\cC_a|$. 
Then
\begin{align*}
\sum_{j=i+1}^{q} \sqrt{\frac{|\cC_i||\cC_j|^2}{D_{i,j}}}
\le 4 \sqrt{\sum_{j=i+1}^{q}|\cC_j|}.
\end{align*}
The same bound holds for $i>q$.
\end{lemma}

\begin{proof}
We prove this bound by induction on $q$.  
Base case $q = i + 1$: Because $|\cC_i|\le k$, $D_{i,i+1}\ge k|\cC_{i+1}|^2$ by \eqref{eq:bound_D}. 
Thus $\sqrt{|\cC_i||\cC_{i+1}|^2/D_{i,i+1}}\le 1\le 4\sqrt{|\cC_{i+1}|}$, proving the base case.

For the induction step, denote $A:=\sum_{j=i+1}^{q-1}|\cC_j|$. 
Lemma~\ref{lem:bounding_on_D_weighted} gives $ D_{i,q}\ge kA/2$. 
Thus, $D_{i,q}\ge k|\cC_q|^2$ from \eqref{eq:bound_D} and $D_{i,q} \ge kA/2$ imply
\begin{align*}
\sqrt{\frac{|\cC_i||\cC_q|^2}{D_{i,q}}}\le \min\left\{1,\frac{\sqrt2|\cC_q|}{\sqrt A}\right\}.
\end{align*}
\begin{enumerate}
\item If $|\cC_q|\le \sqrt{\frac{A}{2}}\le A$, we obtain:
\begin{align*}
4\sqrt{A}+\frac{\sqrt{2}|\cC_q|}{\sqrt{A}} 
\le \sqrt{A}\left(4\sqrt{1+\frac{|\cC_q|}{A}}\right) = 4\sqrt{A+|\cC_q|}.
\end{align*}
\item If $|\cC_q|>\sqrt{\frac{A}{2}}$, then:
\begin{align*}
4\sqrt{A}+1
\le 4\sqrt{A+\sqrt{\frac{A}{2}}}
\le 4\sqrt{A+|\cC_q|}.
\end{align*}
\end{enumerate}

In both cases, the induction hypothesis extends from $q-1$ to $q$.
Finally, we have
\[
\sum\limits_{j=i+1}^q\sqrt{\frac{|\cC_i||\cC_j|^2}{D_{i,j}}}\le 4\sqrt{\sum\limits_{j=i+1}^q |\cC_j|}.
\]
The case $q < i$ works similarly.
\end{proof}

With Lemma~\ref{lem:one_sided_weight_sum_weighted}, we obtain an important lemma for Theorem~\ref{thm:total_energy} as follows:

\begin{lemma}
\label{lem:total_energy_weighted_row_sum}
For every cluster $\cC_i$,

\begin{align}
\sum_{j:\,\cC_j\text{ correlated with }\cC_i} \delta_{i,j}= O(\epsilon).\label{eq:summation_correlation}
\end{align}
\end{lemma}

\begin{proof}
Let $p$ and $q$ be the leftmost and rightmost indices such that $\cC_p$ and $\cC_q$ are correlated with $\cC_i$.
Lemma~\ref{lem:bounding_on_D_weighted} and upper bound of \eqref{eq:bound_D} imply
\begin{equation}
\label{eq:left_nonlargest_mass}
    \sum_{a=p}^{i}|\cC_a|-\max_{p\le a\le i}|\cC_a|
    \le \frac{2}{k}D_{p,i}\le 4\frac{\Delta_k}{d_{min}}.
\end{equation}
With the same argument on the right, we have that if $q>i$,
\begin{equation}
\label{eq:right_nonlargest_mass}
    \sum_{a=i}^{q}|\cC_a|-\max_{i\le a\le q}|\cC_a|
    \le  4\frac{\Delta_k}{d_{min}}.
\end{equation}

First suppose $|\cC_i|\ge 2 \sqrt{\Delta_k/d_{min}}$.
If a correlated neighbor $\cC_j$ had $|\cC_j|\ge|\cC_i|$, then lower and upper bound of \eqref{eq:bound_D} imply
\begin{align*}
    k|\cC_i|^2 \le D_{i,j} < \frac{2 \Delta_k}{d_{min}},
\end{align*}
contradicting $|\cC_i|\ge 2\sqrt{\Delta_k/d_{min}}$.
Hence $\cC_i$ is the largest cluster on each correlated side.
For every correlated $j\ne i$, the summand in \eqref{eq:summation_correlation} is then $\sqrt{|\cC_i||\cC_j|^2/D_{i,j}}$. Therefore,
\begin{align*}
\sum_{\substack{j:\,j\ne i\\ \cC_j\text{ correlated with }\cC_i}} \delta_{i,j}
&= \sum_{\substack{j:\,j\ne i\\ \cC_j\text{ correlated with }\cC_i}} \frac{\epsilon}{k^{1/4}}\sqrt{\frac{\min\{|\cC_i|^2,|\cC_j|^2\}\max\{|\cC_i|,|\cC_j|\}}{D_{i,j}}}\\
&\le\frac{\epsilon}{k^{1/4}}\left( 4\sqrt{\sum_{a=p}^{i-1}|\cC_a|} +4\sqrt{\sum_{a=i+1}^{q}|\cC_a|}\right)\tag{By Lemma~\ref{lem:one_sided_weight_sum_weighted} on two sides}\\
&\le \frac{\epsilon}{k^{1/4}}\cdot 8\sqrt{4\frac{\Delta_k}{d_{min}}}\tag{by \eqref{eq:left_nonlargest_mass} and \eqref{eq:right_nonlargest_mass}}\\
&= 16\frac{\epsilon}{k^{1/4}}(\frac{\Delta_k}{d_{min}})^{1/2}.
\end{align*}

It remains to consider $|\cC_i| < 2 \sqrt{\Delta_k/d_{min}}$.
By lower bound of \eqref{eq:bound_D}, each summand is at most
\begin{align}
\sqrt{\frac{\min\{|\cC_i|^2,|\cC_j|^2\}\max\{|\cC_i|,|\cC_j|\}}{D_{i,j}}}
\le \sqrt{\frac{\max\{|\cC_i|,|\cC_j|\}}{k}}.
\label{eq:small_center_each_summand}
\end{align}
Since  every cluster has size at most $k$, the two possible largest clusters on each side of $\cC_i$ contribute at most $2$ in total by \eqref{eq:small_center_each_summand}.
For the remaining correlated clusters, \eqref{eq:left_nonlargest_mass} and \eqref{eq:right_nonlargest_mass} show that their total size is at most $8\Delta_k/d_{\min}$.
Each of these remaining clusters has size at most $4\Delta_k/d_{\min}$ by \eqref{eq:left_nonlargest_mass} or \eqref{eq:right_nonlargest_mass}.
Using \eqref{eq:small_center_each_summand} and the fact that the number of remaining clusters is at most their total size (i.e., the total number of frequencies among remaining clusters),
\begin{align*}
\sum_{\substack{j:\,j\ne i\\ \cC_j\text{ correlated with }\cC_i}} \delta_{i,j}&=
\sum_{\substack{j:\,j\ne i\\ \cC_j\text{ correlated with }\cC_i}} \frac{\epsilon}{k^{1/4}}\sqrt{\frac{\min\{|\cC_i|^2,|\cC_j|^2\}\max\{|\cC_i|,|\cC_j|\}}{D_{i,j}}} \\
&\le \frac{\epsilon}{k^{1/4}}\sum_{\substack{j:\,j\ne i\\ \cC_j\text{ correlated with }\cC_i}} \sqrt{\frac{\max\{|\cC_i|,|\cC_j|\}}{k}} \tag{by lower bound of \eqref{eq:bound_D}}\\
&\le \frac{\epsilon}{k^{1/4}}\cdot \left(8\frac{\Delta_k}{d_{\min}} \sqrt{ \frac{\max\left\{4\frac{\Delta_k}{d_{\min}},\ 2\sqrt{\frac{\Delta_k}{d_{min}}}\right\}}{k}}
+2\right) \\
&\le 16\frac{\epsilon}{k^{3/4}}(\frac{\Delta_k}{d_{min}})^{3/2}+2\frac{\epsilon}{k^{1/4}}.
\end{align*}
\end{proof}

Recall in $\eqref{eq:max D in correlation}$ that $\Delta_k/d_{min}=O(k^{0.5}\epsilon)$. Combining the two cases together makes:

\begin{align*}
\sum_{\substack{j:\,j\ne i\\ \cC_j\text{ correlated with }\cC_i}} \delta_{i,j}\le& \max\left\{16\frac{\epsilon}{k^{1/4}}(\frac{\Delta_k}{d_{min}})^{1/2},16\frac{\epsilon}{k^{3/4}}(\frac{\Delta_k}{d_{min}})^{3/2}+2\frac{\epsilon}{k^{1/4}}\right\}\\
\le& O(\epsilon^{1.5})+O(\epsilon^{2.5})+2\frac{\epsilon}{k^{1/4}}
\le O(\epsilon).\\
\end{align*}

\begin{proofof}{Theorem~\ref{thm:total_energy}}
For $i,j \in S$ with $i \neq j$ and $\cC_i$ being correlated with $\cC_j$, Claim~\ref{clm:satisbility_of_delta_construction} implies
\begin{align}
|\langle H_k \cdot x_{\cC_i}, H_k \cdot x_{\cC_j}\rangle|
&\le \delta_{i,j}\cdot \|H_k \cdot x_{\cC_i}\|_2 \cdot \|H_k \cdot x_{\cC_j}\|_2\nonumber\\
&\le \delta_{i,j}\left(\frac{\|H_k \cdot x_{\cC_i}\|_2^2}{2}+\frac{\|H_k \cdot x_{\cC_j}\|_2^2}{2}\right).\label{eq:bound inner product by delta}
\end{align}

If $\cC_i$ and $\cC_j$ are not correlated, then their filtered Fourier supports are disjoint, so their inner product is zero.
Hence, by Lemma~\ref{lem:total_energy_weighted_row_sum},
\begin{align*}
\left\|H_k \cdot \sum_{i\in S} x_{\cC_i} \right\|_2^2
&= \left\| \sum_{i\in S} H_k \cdot x_{\cC_i} \right\|_2^2 \\
&= \sum_{i\in S} \|H_k \cdot x_{\cC_i} \|_2^2
+ \sum_{\substack{i,j\in S\\ i \neq j}} \langle H_k \cdot x_{\cC_i}, H_k \cdot x_{\cC_j}\rangle \\
&= \sum_{i\in S} \|H_k \cdot x_{\cC_i} \|_2^2
+ \sum_{i \in S} \sum_{\substack{j\in S\\ i \neq j}} \delta_{i,j} \|H_k \cdot x_{\cC_i}\|_2^2 \\
&=\sum_{i\in S}\left(1 \pm O(\epsilon)\right) \|H_k \cdot x_{\cC_i} \|_2^2\\
&=\left(1 \pm O(\epsilon)\right)\sum_{i\in S} \|H_k \cdot x_{\cC_i} \|_2^2.
\end{align*}
\end{proofof}

\subsection{Localized Filters for Small Clusters}
\label{subsec:heavy_region_localized_filters}

\begin{claim}\label{clm:cluster_concentration}
    For any $\ell$, $\delta$ such that $\Delta_{\ell,\delta}\le \Delta_k$, and any cluster $\cC$ with $|\cC| \le \ell$, $\wh{H_{\ell,\delta}\cdot x_{\cC}}$ is a good approximation of $\wh{H \cdot x_{\cC}}$: 
    \begin{align*}
        \| \wh{H_{\ell,\delta}\cdot x_{\cC}} - \wh{H_k \cdot x_{\cC}}\|_2^2 \le O(\delta) \cdot \|\wh{H_k \cdot x_{\cC}}\|_2^2.
    \end{align*}

    Since $\wh{H_{\ell,\delta} \cdot x_{\cC}}$ is supported in $R:=range(\cC) \pm \Delta_{\ell,\delta}$, this implies that a $1-O(\delta)$ fraction of the energy of $\|\wh{H_k \cdot x_{\cC}}\|_2^2$ is concentrated in $R$:
    \begin{align*}
        \|\wh{H_k \cdot x_{\cC}}\|_{R}^2 \ge (1-O(\delta)) \cdot \|\wh{H_k \cdot x_{\cC}}\|_2^2.
    \end{align*}
\end{claim}
\begin{proof}
We may assume that $\Delta_{\ell, \delta} \le \Delta_k$ holds; otherwise we can simply replace all $H_{\ell, \delta}$ with $H_k$ in the above statement.
We first compare $H_k$ and $H_{\ell,\delta}$ in the time domain.  
Let $I_{\ell,\delta} := \left[-1+\frac{C\delta}{\ell^2}, 1-\frac{C\delta}{\ell^2} \right]$.
Applying Claim~\ref{claim:H-bounds} to $H_{\ell,\delta}$ shows $H_{\ell,\delta}(t)=1\pm O(\delta)$ for $t\in I_{\ell,\delta}$.
By assumption, we have $\delta = \Omega (\frac{\ell^2}{k^2} \epsilon)$. 
So $I_{\ell,\delta}$ lies inside the region on which the filter $H_k$ is also equal to $1 + \left(\frac{\varepsilon}{k}\right)^{\Omega(1)} \subset 1\pm O(\delta)$. 
Hence, $|H_k(t)-H_{\ell,\delta}(t)| \le O(\delta)$. 
Therefore,
\begin{align*}
    \int_{I_{\ell,\delta}} |H_k(t)-H_{\ell,\delta}(t)|^2 |x_{\cC}(t)|^2 dt \le O(\delta) \cdot \int_{-1}^{1}|x_{\cC}(t)|^2 dt.
\end{align*}
On the boundary layer $[-1,1]\setminus I_{\ell,\delta}$, whose length is $O(\delta / \ell^2)$,  we simply use the bound $|H_k|,|H_{\ell,\delta}|\le 2$.
This implies
\begin{align*}
    \int_{[-1,1]\setminus I_{\ell,\delta}} |H_k(t)-H_{\ell,\delta}(t)|^2 |x_{\cC}(t)|^2 dt
    \le O\left(\frac{\delta}{\ell^2}\right) \cdot \sup_{|t|\le1}|x_{\cC}(t)|^2
    \le O(\delta) \cdot \int_{-1}^{1}|x_{\cC}(t)|^2 dt,
\end{align*}
where the last inequality is by Property~\ref{item:uniform_bound_on_interval} of Lemma~\ref{lemma:bounds_Fourier_sparse_signals}.

It remains to control the tails outside $[-1,1]$.  
Since $|\cC| \le \ell$,
\begin{align*}
    \int_{\R\setminus[-1,1]} |H_k(t)-H_{\ell,\delta}(t)|^2 |x_{\cC}(t)|^2 dt
    \le \int_{\R\setminus[-1,1]} (|H_k(t)|^2 + |H_{\ell,\delta}(t)|^2)\cdot |x_{\cC}(t)|^2 dt
    \le O(\delta) \cdot \int_{-1}^{1}|x_{\cC}(t)|^2 dt,
\end{align*}
by Claim~\ref{claim:H-bounds} applied with sparsity $\ell$.
Summing the above results, we obtain
\begin{align*}
    \int_{\R} |H_k(t)-H_{\ell,\delta}(t)|^2 |x_{\cC}(t)|^2 dt \le O(\delta) \cdot \int_{-1}^{1}|x_{\cC}(t)|^2 dt.
\end{align*}

With Plancherel,
\begin{align*}
    \| \wh{H_{\ell,\delta}\cdot x_{\cC}} - \wh{H_k \cdot x_{\cC}} \|_2^2
    \le O(\delta) \cdot \int_{-1}^{1}|x_{\cC}(t)|^2 dt
    \le O(\delta) \cdot \| \wh{H_k \cdot x_{\cC}} \|_2^2,
\end{align*}
where the last inequality uses the inside-energy guarantee of $H_k$.

Finally, since $\supp (\wh{H_{\ell,\delta}\cdot x_{\cC}}) \subseteq R$,
\begin{align*}
    \| \wh{H_k \cdot x_{\cC}} \|_{[-\infty, \infty]\setminus R}^2
    = \|\wh{H_{\ell,\delta}\cdot x_{\cC}} - \wh{H_k \cdot x_{\cC}}\|_{[-\infty, \infty]\setminus R}^2
    \le O(\delta) \cdot \| \wh{H_k \cdot x_{\cC}} \|_2^2.
\end{align*}
Equivalently,
\begin{align*}
    \|\wh{H_k \cdot x_{\cC}} \|_{R}^2 \ge (1-O(\delta)) \cdot\| \wh{H_k \cdot x_{\cC}} \|_2^2.
\end{align*}
\end{proof}

\begin{proofof}{Corollary~\ref{cor:localized_theta_small_clusters}}
Recall that $\Delta_{\theta}=\Delta_{\theta_s,\epsilon^2/k}$. Because $\theta_s$ is chosen to satisfy $\Delta_{\theta_s,\epsilon^2/k} < \dist(\cC_i,\cC_j)/2$, it indicates all clusters in $T$ have disjoint support sets in $\wh{H_{\theta} \cdot x_{\cC_j}}$.
By Claim~\ref{clm:cluster_concentration}, for each $j\in T$,
\begin{align*}
    \|\wh{H_k \cdot x_{\cC_j}}-\wh{H_\theta\cdot x_{\cC_j}}\|_2^2
    \le \frac{\epsilon^2}{k}\|\wh{H_k \cdot x_{\cC_j}}\|_2^2 .
\end{align*}
Hence,
\begin{align*}
    \|\sum_{j\in T}\wh{H_k \cdot x_{\cC_j}} - \sum_{j\in T}\wh{H_\theta\cdot x_{\cC_j}}\|_2 
    & \le \sum_{j \in T} \| \wh{H_k \cdot x_{\cC_j}} - \wh{H_{\theta} \cdot x_{\cC_j}}\|_2 \\
    & \le \frac{\epsilon}{\sqrt{k}} \cdot \sum_{j \in T} \| \wh{H_k \cdot x_{\cC_j}} \|_2 \\
    & \le \frac{\epsilon}{\sqrt{k}} \sqrt{|T|} \cdot \sqrt{\sum_{j \in T}\|\wh{H_k \cdot x_{\cC_j}}\|_2^2} \\
    & \le \epsilon \cdot \sqrt{\sum_{j \in T}\|\wh{H_k \cdot x_{\cC_j}}\|_2^2},
\end{align*}
where the third step is by Cauchy-Schwarz inequality and the last step follows from that the number of clusters is at most $k$.
\end{proofof}

\subsection{Proof of Claim~\ref{clm:bound_range_clusters}}
\label{subsec:heavy_region_range}

Let $r_\ell$ denote the maximum possible length of the range of any cluster with at most $\ell$ frequencies. Since a cluster is generated by merging two smaller clusters, we have
\[
r_\ell=\max\limits_{i=1}^{\ell-1} \left\{r_i+r_{\ell-i}+\min\left\{ d_{min}\cdot \min\{i^2,(\ell-i)^2\}, 2\Delta_{k}\right\}\right\}
\]
for any $\ell \ge 2$. Based on symmetry, the upper bound of $i$ can be replaced with $\frac{\ell}{2}$.


Let $i_0:= \sqrt{\frac{2\Delta_k}{d_{min}}} $ such that $d_{min}\cdot i_0^2 = 2\Delta_k$.

We prove the following hypothesis of $\ell$ by induction: 
\[
r_{\ell} \le
\begin{cases}
d_{min}\cdot \frac{1}{2}(\ell^2-\ell),   &\ell\le i_0,\\
d_{min}\cdot  (\frac{3}{2}i_0\cdot \ell  - \frac{1}{2}\cdot \ell -i_0^2),    & \ell>i_0.
\end{cases}
\]

The base case $\ell=1$ follows from $r_{\ell}=0$. 

Let $p_{\ell,i}:=\left(r_i+r_{\ell-i}+\min\big\{d_{min}\cdot \min\{i^2,(\ell-i)^2\}, 2\Delta_k\}\right) /  d_{min}$. Since $r_\ell = d_{min}\cdot  \max_{i=1}^{\ell/2} p_{\ell,i}$, it is enough to bound $p_{\ell,i}$ in difference cases for the inductive step of $r_\ell$. 

\begin{enumerate}
\item If $0<i \le  \ell -i \le i_0$, we have $\ell \le 2i_0$, then
\begin{align*}
p_{\ell,i}
&\le \frac{1}{2}(i^2-i)+ \frac{1}{2}((\ell-i)^2-(\ell-i)) + i^2 \\
&= \frac{1}{2}(\ell^2 - \ell)- (\ell-2i)i\\
&\le \frac{1}{2}(\ell^2 - \ell).
\end{align*}
$\ell \in [i_0,2i_0]$ guarantees $\frac{1}{2}(i_0 - \ell)(2 i_0 - \ell)\le 0$, which implies $\frac{1}{2}(\ell^2 -\ell)\le \frac{3}{2}i_0\cdot \ell  - \frac{1}{2}\cdot \ell -i_0^2 $ for $\ell \in [i_0,2i_0]$. This implies 
$$
p_{\ell,i} \le
\begin{cases}
\frac{1}{2}(\ell^2-\ell), &\ell \le i_0,\\
\frac{3}{2}i_0\cdot \ell - \frac{1}{2}\cdot \ell - i_0^2, &\ell > i_0.
\end{cases}
$$
    
\item If $0< i\le i_0 < \ell -i$, we have $\ell > i_0$, then
\begin{align*}
p_{\ell,i}&\le \frac{1}{2}(i^2-i)+(\frac{3}{2}i_0\cdot (\ell-i)-\frac{1}{2}\cdot (\ell-i) -  i_0^2)+ i^2\\
&= \frac{3}{2}i_0\cdot \ell - \frac{1}{2}\cdot \ell -\frac{3}{2}(i_0-i)\cdot i - i_0^2     \\
&\le \frac{3}{2}i_0\cdot \ell - \frac{1}{2}\cdot \ell - i_0^2. 
\end{align*}    

\item If $i_0 < i \le \ell -i $, we have $\ell > i_0$, then
\begin{align*}
p_{\ell,i}&\le   (\frac{3}{2}i_0\cdot i  - \frac{1}{2}\cdot i -i_0^2) +  (\frac{3}{2}i_0\cdot (\ell-i)  - \frac{1}{2}\cdot (\ell-i) -i_0^2) +  i_0^2\\
&= \frac{3}{2}i_0\cdot \ell - \frac{1}{2}\cdot 
\ell -i_0^2. 
\end{align*}
\end{enumerate}
Combining all of the above cases with $r_\ell=d_{min} \cdot \max\limits_{i=1}^{\ell/2}p_{\ell,i}$, we finish the proof of $r_\ell$.

The next observation is that $\sum_{i \in [n]} |range(\cC_i)|$ is upper bounded by $r_k$ shown above, where $k=\sum_{i\in[n]} |\cC_i|$. Because for any pair of clusters $\cC_i,\cC_j$, $|range(\cC_i)|+ |range(\cC_j)| \le r_{|\cC_i|}+r_{|\cC_j|} \le r_{|\cC_i|+|\cC_j|}$, 
\[\sum_{i \in [n]} |range(\cC_i)| \le r_{\sum_{i\in [n]}|\cC_i|} \le r_k.\]

The last part is to bound $r_k$. Recall $\Delta_k=C^2 (\frac{k^2 \log k/\epsilon}{\epsilon} + k^2 \log k \cdot \log (k^2 \log k) )$ and $d_{min}:=\frac{2 C_H k^{1.5} \log^3 k/\epsilon}{\epsilon^2} $. Since $\epsilon$ is a small constant, we have
\begin{align}
    \Delta_{k} &= C^2 (\frac{k^2 \log k/\epsilon}{\epsilon} + k^2 \log k \cdot \log (k^2 \log k) ) = \Theta(\frac{k^2\log k}{\epsilon}+k^2\log^2 k), \label{eq:Theta_Delta_k}\\
    d_{min} &=\frac{2 C_H k^{1.5} \log^3 k/\epsilon}{\epsilon^2} = \Theta(\frac{k^{1.5}\log^3 k}{\epsilon^2}). \label{eq:Theta_d}
\end{align}
Combining \eqref{eq:Theta_Delta_k} and \eqref{eq:Theta_d}, we have a bound of $r_k$ and $\sum_{i \in [n]} |range(\cC_i)|$:
\begin{align*}
\sum_{i \in [n]} |range(\cC_i)| &\le r_k \le \frac{3}{2} d_{min} \sqrt{\frac{2\Delta_k}{d_{min}}} \cdot k = \frac{3}{2}  \sqrt{2\Delta_k\cdot d_{min}}\cdot k\\
&= \Theta \left( k^{2.75}\log^2 k \cdot   \epsilon^{-1}\cdot \sqrt{ \epsilon^{-1}+ \log k} \right ) = \tilde{O}(\frac{k^{2.75}}{\epsilon^{1.5}}).
\end{align*}

\section{Heavy Frequency Recovery under Conjecture~\ref{conj:growth_rate}}
\label{sec:heavy-region-clustering-conjecture}
Assuming Conjecture~\ref{conj:growth_rate}, we show a better guarantee on the list of frequencies returned from Lemma~\ref{lem:find_heavy_frequencies}. In fact, Conjecture~\ref{conj:growth_rate} could improve the construction of $H_k$ and the query complexity by a $\log k$ factor. For ease of exposition, we focus on the improvement of the main term $k^{O(1)}$ in this work and omit that part.

First, we improve Claim~\ref{clm:almost_orthogonal_clusters} under Conjecture~\ref{conj:growth_rate}. 
\begin{claim}
\label{clm:almost_orthogonal_clusters_under_conjecture}
For two signals of Fourier sparsity $\ell$ and $r$ separately ($\ell \le r\le k$)
\begin{align*}
    w(t):=\sum_{j=1}^{\ell}\alpha_j e^{2\pi\bi f'_jt} \qquad \text{ and }
    \qquad
    z(t):=\sum_{j=1}^{r}\beta_j e^{2\pi\bi f_jt},
\end{align*}
if the distance between their frequencies $\min_{j,j'} |f_j-f'_{j'}| \ge \min\left\{ C_H (\frac{\ell r}{\delta} + \frac{\ell^2 \log (r/\delta)}{\delta^2}) \log^2 k, 2\Delta_H \right\}$ 
for some constant $C_H$,
then
\begin{align*}
    |\langle H_k w,H_k z\rangle|
    \le \delta \cdot \|H_k w\|_2 \cdot \|H_k z\|_2
    \qquad \text{ and } \qquad
    |\langle w,z\rangle_{[-1,1]}|
    \le \delta \cdot \|w\|_{[-1,1]} \cdot \|z\|_{[-1,1]} .
\end{align*}
\end{claim}
The proof of Claim~\ref{clm:almost_orthogonal_clusters_under_conjecture} is very similar to the proof of Claim~\ref{clm:almost_orthogonal_clusters}, which is deferred to Appendix~\ref{sec:almost_orthogonal_clusters_under_conjecture}.

In this section, We reset $d'_{min}:=\frac{C_H k \log^4 (k/\epsilon)}{\epsilon^2}$ and use Algorithm~\ref{alg:cluster2} to partition frequencies into clusters. The only difference compared to Algorithm~\ref{alg:cluster1} is that the distance becomes $d'_{min} \cdot \min\{|\cC_i|,|\cC_j|\}$.
\begin{algorithm} 
    \caption{Partition Frequencies into Clusters \label{alg:cluster2}}
    \begin{algorithmic}
        \Procedure{}{frequencies $f_1,\ldots,f_k$ with amplitudes $\alpha_1,\ldots,\alpha_k$}      
        \State Define $k$ clusters $\cC_i:=\{(f_i,\alpha_i)\}$ and $d'_{min}:=\frac{C_H k \log^4 (k/\epsilon)}{\epsilon^2}$
        
        \While{ $\exists~\cC_i$ and $\cC_j$ such that $       
        \operatorname{dist}(\cC_i,\cC_j) \le \min\bigg\{ d'_{min} \cdot \min\{|\cC_i|,|\cC_j|\}, 2\Delta_k \bigg\}
        $}
        \State  merge all clusters whose frequencies lie between $\cC_i$ and
    $\cC_j$ into one cluster
        \EndWhile
        \State Return all remaining clusters $\cC$
    \EndProcedure
    \end{algorithmic}
\end{algorithm}

Now we state the main guarantee and finish its proof in the rest of this section.
\begin{theorem}\label{thm:covering_radius_conjecture}
     Let $D:=\frac{k^{2}}{\epsilon^2} \cdot (\log k)^{O(1)}$ be the covering radius and $L$ be the list of frequencies from Lemma~\ref{lem:find_heavy_frequencies}. Then \[
    \mathcal{R}:=\{\cC_i: \exists f \in \cC_i \text{ with } min_{\tilde{f}_i \in L} |\tilde{f}_i - f| \le D \}\] covered by $L$ within the distance $D$ satisfies $\|H_k \cdot (\sum_{\cC \in \mathcal{R}} x_{\cC})  - H_k x\|_{2}^2 = O(\epsilon) \cdot \|H_k x\|_2^2$. 
\end{theorem}

In the rest of this section, we finish the proof of Theorem~\ref{thm:covering_radius_conjecture} under Conjecture~\ref{conj:growth_rate}. The proof strategy is almost the same as the outline of Theorem~\ref{thm:covering_radius_heavy_region} with the following two improvements.
\begin{theorem}
\label{thm:total_energy_conjecture}
Let $d'_{min}:=\frac{C_H k \log^4 (k/\epsilon)}{\epsilon^2}$ and $\cC_1,\ldots,\cC_n$ be $n$ clusters with $\dist(\cC_i,\cC_j) \ge  \min\bigg\{ d'_{min} \cdot \min\{|\cC_i|,|\cC_j|\}, 2\Delta_k \bigg\}$ for any two $\cC_i$ and $\cC_j$.
For every $S\subseteq[n]$,
\begin{align*}
    \left\|H_k \cdot\sum_{i\in S}x_{\cC_i}\right\|_2^2
    = \left(1\pm O(\epsilon)\right) \sum_{i\in S}\|H_k \cdot x_{\cC_i}\|_2^2.
\end{align*}
In particular, $
    \|H_k \cdot x\|_2^2 
    = \left(1\pm O(\epsilon)\right) \sum_{i=1}^n\|H_k \cdot x_{\cC_i}\|_2^2 $.
\end{theorem}

\begin{claim} \label{clm:bound_range_clusters_conjecture}
From Algorithm~\ref{alg:cluster2}, a cluster with $\ell$ frequencies has a range of length at most 
\[
r_\ell \le
\begin{cases}
d'_{min}\cdot O( \ell \log \ell),   &\ell \le \frac{2\Delta_k}{d'_{min}},\\
d'_{min}\cdot \log \frac{2\Delta_k}{d'_{min}} \cdot O( \ell ),    & \ell >\frac{2\Delta_k}{d'_{min}}.
\end{cases}
\]
Moreover, $\sum_{i \in [n]} |range(\cC_i)| = \tilde{O}(k) \cdot d'_{min} = \tilde{O}(\frac{k^2}{\epsilon^2})$. 
\end{claim}

We are ready to finish the proof of Theorem~\ref{thm:covering_radius_conjecture}. The proofs of Theorem~\ref{thm:total_energy_conjecture} and Claim~\ref{clm:bound_range_clusters_conjecture} are deferred to Section~\ref{sec:proof_total_energy_conjecture} and Section~\ref{heavy_region_range_conjecture} separately.

\begin{proofof}{Theorem~\ref{thm:covering_radius_conjecture}}
Let $\theta_s=\Theta(k)$ be the smallest integer with $\Delta_{\theta_s, \frac{\epsilon^2 \cdot \theta_s}{k}} \ge \Delta_k$. We consider 
\begin{equation}\label{eq:def_z_conjecture}
    z(t):=\sum_{i: |\cC_i| < \theta_s} H_{|\cC_i|, \frac{\epsilon^2 |\cC_i|}{k}}(t) \cdot x_{\cC_i}(t) + \sum_{i: |\cC_i| \ge \theta_s} H_{k}(t) \cdot x_{\cC_i}(t).
\end{equation}
Similar to Corollary~\ref{cor:localized_theta_small_clusters}, we bound the error between $z$ and $H_k \cdot x$ as follows:
\begin{align*}
    \|z - H_k \cdot x\|_2 & =  \|\sum_{i: |\cC_i| < \theta_s} (H_{|\cC_i|, \frac{\epsilon^2 |\cC_i|}{k}} - H_k) \cdot x_{\cC_i}\|_2 \\
    & \le \sum_{i: |\cC_i| < \theta_s} \| (H_{|\cC_i|, \frac{\epsilon^2 |\cC_i|}{k}} - H_k) \cdot x_{\cC_i}\|_2 \\
    & \le \sum_{i: |\cC_i| < \theta_s} (\frac{\epsilon^2 |\cC_i|}{k})^{1/2} \cdot \|H_k \cdot x_{\cC_i}\|_2 \tag{by Claim~\ref{clm:cluster_concentration}} \\
    & \le (\sum_i \frac{\epsilon^2 |\cC_i|}{k})^{1/2} \cdot (\sum_i \|H_k \cdot x_{\cC_i}\|^2_2)^{1/2} \tag{the Cauchy-Schwartz inequality} \\
    & \le \epsilon \cdot (1+O(\epsilon)) \|H_k x\|_2. \tag{by Theorem~\ref{thm:total_energy_conjecture}}
\end{align*}
Also, the above calculation implies that for any subset $T$ of clusters, 
\begin{equation}\label{eq:error_replacing_H_k}
\|\sum_{i \in T} (H_{|\cC_i|, \frac{\epsilon^2 |\cC_i|}{k}} - H_k) \cdot x_{\cC_i}\|^2_2 \le \frac{\epsilon^2 \sum_{i \in T} |\cC_i|}{k} \cdot (\sum_i \|H_k \cdot x_{\cC_i}\|^2_2).    
\end{equation}

Now we define $n$ intervals corresponding to the Fourier support of each $\cC_i$ in $z$ defined above:
\begin{align*}
    I_i =
    \begin{cases}
    range(\cC_i) \pm \Delta_{|\cC_i|, \frac{\epsilon^2 |\cC_i|}{k}},
        &|\cC_i| < \theta_s,\\[1.2ex]
    range(\cC_i) \pm \Delta_k,
        & |\cC_i| \ge \theta_s.
    \end{cases}
\end{align*}

Similar to the proof of Theorem~\ref{thm:covering_radius_heavy_region}, we keep merging intervals as long as there exist $I_i$ and $I_j$ with $I_i \cap I_j \neq \emptyset$. For convenience, let $J_1,\ldots,J_m$ be the remaining disjoint intervals. For a cluster $\cC$ and interval $J_i$, we use $\cC \subset J_j$ to indicate that each frequency $f \in \cC$ satisfies $f \in J_j$.

By Theorem~\ref{thm:total_energy_conjecture}, for each $J_j$,
\begin{equation}\label{eq:sum_clusters_J}
    \| \sum_{i: \cC_i \subset J_j} H_k x_{\cC_i} \|_2^2 \ge (1-O(\epsilon)) \sum_{i: \cC_i \subset J_j} \|  H_k x_{\cC_i} \|_2^2.
\end{equation}
\eqref{eq:error_replacing_H_k} implies that 
\begin{equation}\label{eq:error_J}
     \| \sum_{i: |\cC_i|<\theta_s \text{ and }\cC_i \subset J_j} (H_{|\cC_i|, \frac{\epsilon^2 |\cC_i|}{k}} - H_k) x_{\cC_i} \|_2 \le \epsilon \cdot (\sum_{i: \cC_i \subset J_j} \|  H_k x_{\cC_i} \|_2^2)^{1/2}
\end{equation}
So \eqref{eq:sum_clusters_J} and \eqref{eq:error_replacing_H_k} imply that 
\[
\| \sum_{i: |\cC_i|<\theta_s \text{ and }\cC_i \subset J_j} H_{|\cC_i|, \frac{\epsilon^2 |\cC_i|}{k}} \cdot x_{\cC_i} + \sum_{i: |\cC_i| \ge \theta_s \text{ and }\cC_i \subset J_j} H_k \cdot x_{\cC_i} \|_2^2  \ge (1- O(\epsilon)) \cdot (\sum_{i: \cC_i \subset J_j} \|  H_k x_{\cC_i} \|_2^2).
\]
Because $J_j$ is disjoint with the rest $J_1,\ldots,J_m$ and the Fourier supports of $H_{|\cC_i|, \frac{\epsilon^2 |\cC_i|}{k}} \cdot x_{\cC_i}$ and $H_k \cdot x_{\cC_i}$ in the LHS of the above inequality are in $J_j$, this is equivalent to 
\begin{equation}
    \|\wh{z}\|_{J_j}^2 \ge (1- O(\epsilon)) \cdot (\sum_{i: \cC_i \subset J_j} \|  H_k \cdot x_{\cC_i} \|_2^2).
\end{equation}
The rest of this proof is identical to the proof of Theorem~\ref{thm:covering_radius_heavy_region} except the calculation of the covering radius
\begin{align*}
    \sum_i |I_i| & = \sum_{i=1} |range(\cC_i)| + \sum_{i:|\cC_i|<\theta_s} 2\Delta_{|\cC_i|,\frac{\epsilon^2 |\cC_i|}{k}}+ 2 \Delta_k \cdot k/\theta_s \\
    & \le \sum_{i:|\cC_i|<\theta_s} 2C^2 \cdot \left( \frac{|\cC_i|^2 \log \frac{|\cC_i|}{\frac{\epsilon^2 |\cC_i|}{k}}}{\frac{\epsilon^2 |\cC_i|}{k}} + |\cC_i|^2 \log |\cC_i| \cdot \log (|\cC_i|^2 \log |\cC_i| ) \right)  + O(k \cdot d'_{min} + \Delta_k) \tag{by the definition of $\Delta_{\ell,\delta}$ in Lemma~\ref{lemm:construction_H}}\\
    & \le \sum_{i:|\cC_i|<\theta_s} 2C^2 \cdot \left(\frac{|\cC_i| \cdot k \log \frac{k}{\epsilon^2}}{\epsilon^2} + |\cC_i|^2 \log |\cC_i| \cdot \log (|\cC_i|^2 \log |\cC_i| ) \right) + O(k \cdot d'_{min} + \Delta_k) \\
    & = O(\frac{k^2 \log \frac{k}{\epsilon^2}}{\epsilon^2}) + O(k^2 \log^2 k) + O(k \cdot d'_{min} + \Delta_k)=\tilde{O}(k^2/\epsilon^2).
\end{align*}
\end{proofof}

\subsection{Proof of Theorem~\ref{thm:total_energy_conjecture}}\label{sec:proof_total_energy_conjecture}


Recall $\Delta_k:=C^2 (\frac{k^2 \log k/\epsilon}{\epsilon} + k^2 \log k \cdot \log (k^2 \log k) )\le \frac{3C^2k^{2}\log^2 (k/\epsilon)}{\epsilon}$ and $d'_{min}:=\frac{C_H k \log^4 (k/\epsilon)}{\epsilon^2}$. Therefore,

\begin{equation}
\label{eq:max_D_in_correlation_conjecture}
\frac{\Delta_k}{d'_{min}}\le \frac{3C^2k^2\log^2(k/\epsilon)/\epsilon}{C_Hk\log^4(k/\epsilon)/\epsilon^2}\le \frac{3C^2\epsilon k}{C_H\log^2(k/\epsilon)}.
\end{equation}

We assume that $\cC_1,\ldots,\cC_n$ are sorted by their frequencies in this proof. Similarly to Section~\ref{subsec:heavy_region_clustering}, for correlated clusters $\cC_i$ and $\cC_j$ with $\dist(\cC_i,\cC_j) \le 2 \Delta_k$, we define
\begin{align*}
    D_{i,j}:= \frac{k\log k\cdot \dist(\cC_i,\cC_j)} {d'_{min}}.
\end{align*}

Because $D_{i,j}$ is defined only when $\cC_i$ and $\cC_j$ are correlated and $\dist(\cC_i,\cC_j) \ge d'_{min} \cdot \min\{|\cC_i|,|\cC_j|\}$ in this case, we have the following bounds on $D_{i,j}$:
\begin{equation}
\label{eq:bound_D_conjecture}
D_{i,j} \in \left[ k \log k \cdot \min\{|\cC_i|,|\cC_j|\}, k \log k \cdot \frac{2\Delta_k}{d'_{min}} \right].
\end{equation}


Claim~\ref{clm:almost_orthogonal_clusters_under_conjecture} improves Claim~\ref{clm:satisbility_of_delta_construction} to the following bound.
\begin{claim}
\label{clm:satisbility_of_delta_construction_conjecture}

Let $\delta_{i,j}=0$ for uncorrelated $\cC_i$ and $\cC_j$. For correlated $\cC_i$ and $\cC_j$, let 
\begin{equation}
\label{eq:selection_on_delta_conjecture}
    \delta_{i,j}:=
    \frac{\epsilon^2}{\log k}\cdot \frac{|\cC_i|\cdot|\cC_j|}{D_{i,j}}+\epsilon\sqrt{\frac{3\log (k/\epsilon)}{\log k}}\cdot \frac{\min\{|\cC_i|,|\cC_j|\}}{\sqrt{D_{i,j}}}.
\end{equation}
Then for any $i$ and $j$, $\delta_{i,j}$ satisfies the condition of Claim~\ref{clm:almost_orthogonal_clusters_under_conjecture}: 
\begin{equation}
\label{eq:weighted_pair_correlation_conjecture}
    |\langle H_k \cdot x_{\cC_i},H_k \cdot x_{\cC_j}\rangle|
    \le \delta_{i,j} \|H_k \cdot x_{\cC_i}\|_2\|H_k \cdot x_{\cC_j}\|_2 .
\end{equation}
\end{claim}

\begin{proof}
For convenience, we denote $\ell:=\min\{|\cC_i|,|\cC_j|\}$, $r:=\max\{|\cC_i|,|\cC_j|\}$, $D:=D_{i,j}$, and $\delta:=\delta_{i,j}$. 
Then $\delta = \frac{\epsilon^2 \ell r}{D\log k} + \sqrt{\frac{3\epsilon^2\log(k/\epsilon)}{D \log k}} \cdot \ell$.
By Claim~\ref{clm:almost_orthogonal_clusters_under_conjecture}, the result follows if 
\begin{equation} \label{eq:limit_on_delta_conjecture}
    \frac{D d'_{min}}{k\log k} = \dist(\cC_i,\cC_j) \ge C_H (\frac{\ell r}{\delta} + \frac{\ell^2 \log (r/\delta)}{\delta^2}) \log^2 k.
\end{equation}
Hence, it remains to show that \eqref{eq:limit_on_delta_conjecture} holds for our choice of $\delta$.

By equation \eqref{eq:max_D_in_correlation_conjecture} and the upper bound in \eqref{eq:bound_D_conjecture}, 
\begin{align*}
    D \le k \log k \cdot \frac{2\Delta_k}{d'_{min}} \le \frac{6C^2}{C_H} \frac{\epsilon k^2\log k}{\log^2(k/\epsilon)} \le k^2,
\end{align*}
for sufficiently large $C_H$. Thus, the second term of $\delta$ implies
\begin{align}
    \frac{r}{\delta} \le \frac{r}{\ell} \sqrt{\frac{D \log k}{3\epsilon^2\log(k/\epsilon)}} \le \frac{k^2}{\epsilon} \quad \text{and} \quad \log\left(r/\delta\right)\le 3\log(k/\epsilon). \label{eq:upper_bound_partial_term1_conjecture}
\end{align}

The first term of $\delta$ shows
\[
    \frac{D\log k}{\epsilon^2} \cdot \frac{\epsilon^2\ell r}{D\log k}\cdot\delta = \ell r\,\delta,
\]
and the second term gives
\[
    \frac{D\log k}{\epsilon^2} \cdot \left( \sqrt{\frac{3\epsilon^2\log(k/\epsilon)}{D \log k}} \cdot \ell \right)^2
    = 3\ell^2\log(k/\epsilon).
\]
Hence,
\[
    \frac{D\log k}{\epsilon^2}\delta^2 \ge \ell r\,\delta+3\ell^2\log(k/\epsilon).
\]
Dividing by $\delta^2$ and using \eqref{eq:upper_bound_partial_term1_conjecture} yields
\[
    \frac{D\log k}{\epsilon^2}
    \ge \frac{\ell r}{\delta} + \frac{\ell^2\log(r/\delta)}{\delta^2},
\]
which satisfies the separation condition in Claim~\ref{clm:almost_orthogonal_clusters_under_conjecture}.
\end{proof}

Similarly to the framework of Section \ref{subsec:heavy_region_clustering}, we would like to show that $\sum_{j: \text{ correlated with } i} \delta_{i,j}=O(\epsilon)$ for any $i$. Just like Lemma~\ref{lem:bounding_on_D_weighted}, we first bound the total number of frequencies in correlated clusters.



\begin{lemma}
\label{lem:bounding_on_D_weighted_conjecture}
Suppose that $x < y$ and $\cC_x$ is correlated with $\cC_y$.
Let $s_x,\ldots,s_y$ be positive integers with $s_a \le |\cC_a|$ for every $a$ and $s_x=\max_{x\le a\le y}s_a$.
\begin{align*}
D_{x,y}\ge \frac{k}{2} \sum_{a=x+1}^{y}s_a.
\end{align*}
\end{lemma}

\begin{proof}
For a threshold $t \ge 1$, let $\cC_{i_1},\ldots,\cC_{i_p}$ be the clusters in $\{\cC_x,\ldots,\cC_y\}$ whose size is at least $t$.
As \eqref{eq:bound_D_conjecture} shows, 
all $D_{i_1,i_2},\ldots,D_{i_{p-1},i_p}$ are larger than $k\log k\cdot t$.  So $(p - 1) k\log k\cdot t \le D_{x,y}$.
Summing this bound over thresholds $t$ shows
\begin{align*}
\sum_{a=x}^{y}s_a
&=\sum_{t=1}^{s_x} \sum_{j = x}^y \mathbb{I}[s_j \ge t] \tag{$s_j = \sum_{t=1}^{s_x} \mathbb{I}[s_j \ge t]$}\\
& \le \sum_{t=1}^{s_x} \left(\frac{D_{x,y}}{k\log k\cdot t}+1\right) \tag{$\sum_{j = x}^y \mathbb{I}[s_j \ge t] \le \sum_{j = x}^y \mathbb{I}[|\cC_j| \ge t] \le \frac{D_{x,y}}{k\log k\cdot t}+1$}\\
& \le \frac{2D_{x,y}}{k} + s_x. \tag{$s_x \le |C_x| \le k$}
\end{align*}
\end{proof}

Unlike Lemma~\ref{lem:one_sided_weight_sum_weighted}, there are two terms in \eqref{eq:selection_on_delta_conjecture}. So we have two lemmas for two terms.

\begin{lemma}
\label{lem:one_sided_weight_sum_weighted1_conjecture}
Suppose that $i < q$ and $\cC_i$ is correlated with $\cC_q$.
Let $s_i,\ldots,s_q$ be positive integers with $s_a \le |\cC_a|$ for every $a$ and $s_i=\max_{i\le a\le q}s_a$.
Then
\[
    \sum_{j=i+1}^{q}\frac{ks_j}{D_{i,j}} \le 3\log\left(\sum_{j=i+1}^{q}s_j\right)+1.
\]
The same bound holds on the left side of $i$.
\end{lemma}

\begin{proof}
We prove this by induction on $q$.  
The base case $q = i + 1$ follows from $D_{i,i+1} \ge k\log k\cdot |\cC_{i+1}| \ge k\log k \cdot s_{i+1}$. 
For the induction step, denote $A:=\sum_{j=i+1}^{q-1}s_j$. 
Lemma~\ref{lem:bounding_on_D_weighted_conjecture} gives $D_{i,q}\ge kA/2$. 
Also, lower bound of \eqref{eq:bound_D_conjecture} gives $D_{i,q}\ge k\log k\cdot |\cC_q| \ge k\log k\cdot s_q$.  
Thus,
\begin{align*}
    \frac{k s_q}{D_{i,q}}
    \le \min\left\{\frac{1}{\log k},\frac{2 s_q}{A}\right\}\le \min\left\{1,\frac{2 s_q}{A}\right\}.
\end{align*}
\begin{enumerate}
\item If $s_q\le \frac{A}{2}$, by $2 x \le 3 \log (1+x)$ for $x \in (0,1/2]$, we have
\begin{align*}
3\log(A)+1+\frac{2 s_q}{A}
\le 3\log(A)+1+3\log \left(1+\frac{ s_q}{A}\right)=3\log (A+ s_q)+1.
\end{align*}
\item If $s_q>\frac{A}{2}$, by $1 \le 3 \log(3/2)$, we obtain
\begin{align*}
3\log(A)+1+1\le 3\log\left(A+\frac{A}2\right)+1
\le 3\log (A+ s_q)+1.
\end{align*}
\end{enumerate}

In both cases, the induction hypothesis extends from $q-1$ to $q$.
Finally, we have
\[
\sum\limits_{j=i+1}^q\frac{k s_j}{D_{i,j}}\le 3\log\left(\sum\limits_{j=i+1}^q s_j\right)+1.
\]
The case $q < i$ works similarly.
\end{proof}

\begin{lemma}
\label{lem:one_sided_weight_sum_weighted2_conjecture}
Suppose that $i < q$ and $\cC_i$ is correlated with $\cC_q$.
Let $s_i,\ldots,s_q$ be positive integers with $s_a \le |\cC_a|$ for every $a$ and $s_i=\max_{i\le a\le q}s_a$.
Then
\begin{align*}
\sum_{j=i+1}^{q} \frac{s_j}{\sqrt {D_{i,j}}}
\le 4\sqrt{\frac{\sum\limits_{j=i+1}^q s_j}{k}}.
\end{align*}
The same bound holds on the left side.
\end{lemma}

\begin{proof}
Again we prove the right-sided bound by induction on $q$.
The case $q = i + 1$ follows from $D_{i,i+1}\ge k\log k\cdot |\cC_{i+1}|$. 
For the induction step, we denote $A:=\sum_{j=i+1}^{q-1}s_j$. 
Lemma~\ref{lem:bounding_on_D_weighted_conjecture} implies $D_{i,q} \ge  kA/2$. 
In addition, the lower bound of \eqref{eq:bound_D_conjecture} shows $D_{i,q}\ge k\log k\cdot |\cC_q|\ge k\log k\cdot s_q$.  
Thus,
\begin{align*}
    \frac{s_q}{\sqrt{D_{i,q}}}
    \le \min\left\{\frac{\sqrt{2}s_q}{\sqrt{kA}},\sqrt{\frac{s_q}{k}}\right\}.
\end{align*}
\begin{enumerate}
\item If $s_q/A \le \frac12$, we obtain:
\begin{align*}
4\sqrt{\frac{A}{k}} + \frac{\sqrt{2}s_q}{\sqrt{kA}}
\le 4\sqrt{\frac{A}{k}}\left(1 + \frac{\sqrt{2} s_q}{4A}\right)
\le 4\sqrt{\frac{A}{k}} \sqrt{1+\frac{s_q}{A}}
= 4\sqrt{\frac{A+s_q}{k}}.
\end{align*}
\item If $s_q/A > \frac12$, we have:
\begin{align*}
4\sqrt{\frac{A}{k}} + \sqrt{\frac{s_q}{k}}
\le 4\sqrt{\frac{A}{k}} \left(1+\sqrt{\frac{s_q}{16A}}\right)
\le 4\sqrt{\frac{A}{k}} \sqrt{1+\frac{s_q}{A}}
= 4\sqrt{\frac{A+s_q}{k}}.
\end{align*}
\end{enumerate}

In both cases, the induction hypothesis extends from $q-1$ to $q$.
Finally, we have
\[
\sum\limits_{j=i+1}^q\frac{ks_j}{\sqrt{D_{i,j}}}
\le 4\sqrt{\frac{\sum\limits_{j=i+1}^q s_j}{k}}
\le 4.
\]
The case $q < i$ works similarly.
\end{proof}

With Lemma~\ref{lem:one_sided_weight_sum_weighted1_conjecture} and Lemma~\ref{lem:one_sided_weight_sum_weighted2_conjecture}, we obtain an important lemma for Theorem~\ref{thm:total_energy_conjecture} as follows:

\begin{lemma}
\label{lem:total_energy_weighted_row_sum_conjecture}
For every cluster $\cC_i$,

\begin{align}
\sum_{j \ne i:\,\cC_j\text{ correlated with }\cC_i} \delta_{i,j}= O(\epsilon).\label{eq:summation_correlation_conjecture}
\end{align}
\end{lemma}

\begin{proof}
In this proof, we will bound two terms of $\delta_{i,j}$ in \eqref{eq:selection_on_delta_conjecture} separately.

Let $s_j:=\min\{|\cC_i|,|\cC_j|\}$.  
Then $s_i=\max_j s_j$, so the above two lemmas can be applied separately to the correlated clusters to the left and to the right of $i$.  
Since $\sum_j s_j\le \sum_j |\cC_j| \le k$, Lemma~\ref{lem:one_sided_weight_sum_weighted1_conjecture} implies
\begin{equation} \label{eq:bound_on_delta_part1}
    \sum_{\substack{j\ne i:\\ \cC_j\text{ correlated with }\cC_i}}
    \frac{ks_j}{D_{i,j}} \le 8\log k.
\end{equation}
Similarly,
Lemma~\ref{lem:one_sided_weight_sum_weighted2_conjecture} gives
\begin{equation} \label{eq:bound_on_delta_part2}
    \sum_{\substack{j\ne i:\\ \cC_j\text{ correlated with }\cC_i}}
    \frac{s_j}{\sqrt{D_{i,j}}} \le 8.
\end{equation}
Recall that $\delta_{i,j} = \frac{\epsilon^2}{\log k}\cdot \frac{|\cC_i|\cdot|\cC_j|}{D_{i,j}}+\epsilon\sqrt{\frac{3\log (k/\epsilon)}{\log k}}\cdot \frac{\min\{|\cC_i|,|\cC_j|\}}{\sqrt{D_{i,j}}}$. 
Because $|\cC_i||\cC_j|\le k s_j$,
\begin{align} \label{eq:bound_on_delta_part3}
    \delta_{i,j} \le \frac{\epsilon^2}{\log k}\cdot \frac{ks_j}{D_{i,j}}+\epsilon\sqrt{\frac{3\log (k/\epsilon)}{\log k}}\cdot \frac{s_j}{\sqrt{D_{i,j}}}.
\end{align}
By plugging \eqref{eq:bound_on_delta_part1} and \eqref{eq:bound_on_delta_part2} into \eqref{eq:bound_on_delta_part3}, we have
\begin{align*}
    \sum_{\substack{j\ne i:\\ \cC_j\text{ correlated with }\cC_i}} \delta_{i,j} 
    \le \frac{\epsilon^2}{\log k} \cdot 8\log k
    + \epsilon\sqrt{\frac{3\log (k/\epsilon)}{\log k}} \cdot 8
    \le O(\epsilon).
\end{align*}
\end{proof}

\begin{proofof}{Theorem~\ref{thm:total_energy_conjecture}}
For $i,j \in S$ with $i \neq j$ and $\cC_i$ being correlated with $\cC_j$, Claim~\ref{clm:satisbility_of_delta_construction_conjecture} and \eqref{eq:weighted_pair_correlation_conjecture} implies
\begin{align}
|\langle H_k \cdot x_{\cC_i}, H_k \cdot x_{\cC_j}\rangle|
&\le \delta_{i,j}\cdot \|H_k \cdot x_{\cC_i}\|_2 \cdot \|H_k \cdot x_{\cC_j}\|_2\nonumber\\
&\le \delta_{i,j}\left(\frac{\|H_k \cdot x_{\cC_i}\|_2^2}{2}+\frac{\|H_k \cdot x_{\cC_j}\|_2^2}{2}\right).\label{eq:bound_inner_product_by_delta_conjecture}
\end{align}
Then we finish this proof by applying Lemma~\ref{lem:total_energy_weighted_row_sum_conjecture} to \eqref{eq:bound_inner_product_by_delta_conjecture}.

\end{proofof}

\subsection{Proof of Claim~\ref{clm:bound_range_clusters_conjecture}}
\label{heavy_region_range_conjecture}

Similarly to the proof of Claim~\ref{clm:bound_range_clusters}, we have:
\[
r_\ell=\max_{1 \le i \le \ell/2} \left\{r_i+r_{\ell-i}+\min\left\{ d'_{min}\cdot \min\{i,\ell-i\}, 2\Delta_{k}\right\}\right\}
\]
for any $\ell \ge 2$. 

Let $i_0:= \frac{2\Delta_k}{d'_{min}}$ such that $d'_{min}\cdot i_0 = 2\Delta_k$, and $C' := \frac{1}{2\log 2}$.

We still prove the following hypothesis by induction in $\ell$:
\[
r_\ell \le
\begin{cases}
d'_{min}\cdot C' \ell\log \ell,   &\ell\le i_0,\\
d'_{min}\cdot (C' \ell\log i_0 +\ell-i_0),    & \ell>i_0.
\end{cases}
\]

The base case $\ell=1$ follows from $r_{\ell}=0$. 

Let $p_{\ell,i}:=\left(r_i+r_{\ell-i}+\min\big\{d'_{min}\cdot \min\{i,\ell-i\}, 2\Delta_k\}\right) / d'_{min}$. It is sufficient to bound different cases of $p_{\ell,i}$ for the inductive step of $r_\ell$.

\begin{enumerate}
\item If $0<i \le  \ell -i \le i_0$, we have $\ell \le 2i_0$, then
\begin{align*}
p_{\ell,i}
&\le C' i\log i + C'(\ell -i)\log(\ell -i)+ i\\
&= \frac{i}{2}\log_2 i + \frac{\ell-i}{2}\log_2(\ell-i) +i\\
&= \frac{x\ell}{2}\log_2 (x\ell) + \frac{(1-x)\ell}{2}\log_2{(1-x)\ell} +x\ell \tag{Assume $i=x\ell$ with $0<x\le 1/2$}\\
&= \frac{1}{2}\log_2 \ell^{\ell} + \ell(\frac{1}{2}(x\log_2 x+(1-x)\log_2{(1-x)}) +x)\\
&\le \frac{1}{2}\ell \log_2 \ell=C' \ell \log \ell. 
\end{align*}

Since $x\ln x +(1-x)\ln(1-x)\le 4\ln2\cdot x(x-1)$ with $0<x<1$ and $2x(x-1)+x = 2x(x-1/2)\le 0$ with $0< x <1/2$, the last inequality holds.

Since $1-\frac{1}{x} \ge \frac{\log_2 {x}}{2}$ with $1\le x\le 2$, we have $C'\ell\log i_0 + \ell - i_0\ge C'\ell\log \ell$ with $i_0\le \ell\le 2i_0$, then
$$
p_{\ell,i} \le
\begin{cases}
C'\ell \log \ell, &\ell \le i_0,\\
C'\ell\log i_0 +\ell  -i _0, &\ell > i_0.
\end{cases}
$$
    
\item If $0< i\le i_0 < \ell -i$, we have $\ell > i_0$, then
\begin{align*}
p_{\ell,i}&\le C' i\log i + (C'(\ell-i)\log i_0 + (\ell - i) - i_0) + i \\
&\le C'\ell \log i_0 +\ell - i_0.
\end{align*}    

\item If $i_0 < i \le \ell -i $, we have $\ell > i_0$, then
\begin{align*}
p_{\ell,i}
&\le  (C'i\log i_0 + i - i_0) +  (C'(\ell-i)\log i_0 + (\ell - i) - i_0)  + i_0 \\
&= C'\ell \log i_0 + \ell - i_0.
\end{align*}
\end{enumerate}
Combining all of the above cases with $r_\ell=d'_{min} \cdot \max\limits_{i=1}^{\ell/2}p_{\ell,i}$, we finish the proof of $r_{\ell}$.

As same as Claim~\ref{clm:bound_range_clusters}, $\sum_{i \in [n]} |range(\cC_i)|$ is upper bounded by $r_{k}$.Combining $\Delta_k=\Theta(\frac{k^2\log k}{\epsilon}+k^2\log^2 k)$ and $d'_{min}:=\frac{C_H k \log^4 k/\epsilon}{\epsilon^2} $, we have a bound of $r_k$ and $\sum_{i \in [n]} |range(\cC_i)|$:
\begin{align*}
\sum_{i \in [n]} |range(\cC_i)| &\le r_k \le d'_{min}\cdot(\frac{k}{2\log 2} \cdot  \log \frac{2\Delta_k}{d'_{min}}+ k) \\
&= \Theta \left(\frac{k^2\log^4 k}{\epsilon^2} 
\cdot \log(k(\frac{\epsilon}{ \log^3 k}+\frac{\epsilon^2}{\log^2 k}))\right ) = \tilde{O}(\frac{k^{2}}{\epsilon^{2}}).
\end{align*}

\section{Main Results}\label{sec:main_proof}
We prove Theorem~\ref{inform:fast_algorithm} and Theorem~\ref{inform:algorithm_conjecture}. Since their proofs are very similar, we combine them as follows.

\begin{theorem}\label{thm:fast_algorithm}
    Given any $F$ and a small constant $\epsilon>0$, let $y(t):=x(t)+\eta(t)$ for $x(t):=\sum_{j=1}^k \alpha_j e^{2 \pi \bi f_j t}$ be the observation over the time window $[-1,1]$ with $k$ arbitrary frequencies $f_1,\ldots,f_k \in [-F,F]$ and $\|\eta(t)\|_{[-1,1]}^2 \le \epsilon \cdot \|x(t)\|_{[-1,1]}^2$. There exists an algorithm that takes $m:=\tilde{O}(k^{3.75})$ samples and $\tilde{O}(m^{\omega})$ time to output $\tilde{x}$ with $\|\tilde{x}-x\|_{[-1,1]}^2 \le O(\epsilon \cdot \|x\|_{[-1,1]}^2 + \|\eta\|_{[-1,1]}^2)$.    

    If Conjecture~\ref{conj:growth_rate} is correct, the same guaranty holds for algorithms with $m_c:=\tilde{O}(k^{3})$ samples and $\tilde{O}(m_c^{\omega})$ time.
\end{theorem}

\begin{proof}
    We first show the algorithm and the analysis for the first part. The algorithm behind the first part (and Theorem~\ref{inform:fast_algorithm}) follows the same outline as the algorithm in \cite{CKPS17,SSWZ23}:
    \begin{enumerate}
        \item We apply Lemma~\ref{lem:find_heavy_frequencies} to obtain $\ell=O(k/\epsilon)$ frequencies in $L:=\{\tilde{f}_1,\ldots,\tilde{f}_\ell\}$ such that    \begin{equation}\label{eq:pf_freq_guarantee}
            \forall f \text{ with } \int_{f-\Delta_k}^{f+\Delta_k} |\wh{x \cdot H_k}(u)|^2 \mathrm{d} u \ge \frac{\epsilon}{5k} \cdot \|y\|_{[-1,1]}^2, \exists \tilde{f} \in L \text{ with } |\tilde{f}-f| = O(\Delta_k).
        \end{equation}
        This step takes $k^2 (\log kF/\epsilon)^{O(1)}$ samples and $k^2 (\log kF/\epsilon)^{O(1)}$ time.

        \item Let $D:=\frac{k^{2.75}}{\epsilon^{1.5}} \cdot (\log k)^{O(1)}$ be the covering radius guaranteed by Theorem~\ref{thm:covering_radius_heavy_region} and the linear family $\mathcal{F}:=\{e^{2 \pi \bi \tilde{f}_j t}\cdot t^{d}: j \in [\ell], d \le C \cdot D\}$ for some constant $C=O(1)$. Then applying a linear regression algorithm (like Theorem~1.1 in \cite{CP19_colt}) finds $\tilde{x} \in \mathcal{F}$ with $\tilde{O}(\frac{\ell \cdot CD}{\epsilon})$ samples and $\tilde{O}(\frac{\ell \cdot CD}{\epsilon})^{\omega}$ time. 
    \end{enumerate}
    Let $\cC_1,\ldots,\cC_n$ be the clusters of $x$ output by Algorithm~\ref{alg:cluster1}.
    Given $L$ and $D$ defined above, let $\mathcal{R} \subset [n]$ be the clusters covered by $L$ within distance $D$ such that Theorem~\ref{thm:covering_radius_heavy_region} guarantees that $x':=\sum_{j \in \mathcal{R}} x_{\cC_j}$ satisfies
    \[
        \|H_k x - H_k x'\|^2_2 = O(\epsilon) \cdot \|H_k x\|_2^2.
    \]
    Lemma~\ref{lemm:construction_H} implies $\|x-x'\|_{[-1,1]}^2 = O(\epsilon) \cdot \|x\|_{[-1,1]}^2$ because all frequencies in $\wh{x'}$ are in $\wh{x}$.
    
    Next, Lemma~\ref{lem:low_deg_approx} implies that there are  degree-$(C \cdot D)$ polynomials $q_1,\ldots,q_\ell$ such that 
    \[
    \|x' - \sum_{j=1}^\ell e^{2 \pi \bi \tilde{f}_j t} \cdot q_j(t)\|_{[-1,1]} \le 2\epsilon \|x'\|_{[-1,1]}.    \]

    A triangle inequality shows
    \[
    \|x - \sum_{j=1}^\ell e^{2 \pi \bi \tilde{f}_j t} \cdot q_j(t)\|^2_{[-1,1]} = O(\epsilon) \cdot \|x\|_{[-1,1]}^2.
    \]

    Because $y=x+\eta$, there exists $z \in \mathcal{F}$ with $\|z - y\|^2_{[-1,1]} \le O(\epsilon \cdot \|x\|^2_{[-1,1]} + \| \eta \|^2_{[-1,1]})$. So linear regression algorithms return $\tilde{x}$ given samples in $y$ within $\|\tilde{x} - x\|_{[-1,1]}^2=O(\epsilon \cdot \|x\|_{[-1,1]}^2 + \| \eta \|_{[-1,1]}^2)$.

    The algorithm for Theorem~\ref{thm:covering_radius_conjecture} is the same except for the setting $D:=\frac{k^{2}}{\epsilon^{2}} \cdot (\log k)^{O(1)}$ from Theorem~\ref{thm:covering_radius_conjecture}. Since the analysis is the same, we omit it here.
\end{proof}

\section*{Acknowledgements}
The authors used Gemini 3.1 during the development of this work to explore
proof strategies and search for related tools in the literature. Gemini
was not used in any part of the exposition. 
The authors assume responsibility for all content. 

\bibliographystyle{alpha} 
\bibliography{bibFFT}

\appendix

\section{Proofs of Theorem~\ref{thm:net_frequency} and Corollary~\ref{cor:query_complexity_learning}}\label{sec:proofs}




\subsection{Proof of Theorem~\ref{thm:net_frequency}}




We use the same rightward-separation construction as \cite{CKPS17}. 

Suppose $f_1<\cdots<f_k$. 
We set $\eta=\epsilon/(Ck^2)$ for a large constant $C$, and define
\begin{align*}
f'_1 = \eta \cdot \lceil f_1/\eta \rceil , \qquad f'_j=\max\big\{\eta \cdot \lceil f_j/\eta \rceil ,f'_{j-1}+\eta\big\} \quad (2\le j\le k).
\end{align*}
We may assume $f_k \le F- k \eta$; otherwise, we can adjust the direction of rounding. 
Thus the new frequencies $\{f'_j\}_{j=1}^k \subset \mathcal{N} = \frac{\epsilon}{C k^2} \cdot\mathbb{Z} \cap [-F,F]$ and each moves by at most $k \eta$.

Starting with $x$, we replace $f_k$ by $f'_k$, then $f_{k-1}$ by $f'_{k-1}$, and continue down to $f_1$.  
Each time a frequency is replaced, we apply Lemma~\ref{lem:shift_one_freq} to the current signal.
With sufficiently large $C$, the triangle inequality yields
\begin{align*}
\|x'-x\|_{[-1,1]}
\le \bigl[(1+ O((k \eta))^k - 1\bigr] \|x\|_{[-1,1]}
\le (e^{\epsilon/2}-1) \cdot \|x\|_{[-1,1]}
\le \epsilon \cdot \|x\|_{[-1,1]},
\end{align*}
where the last inequality holds for $0 < \varepsilon < 1$.

\subsection{Proof of Corollary~\ref{cor:query_complexity_learning}}

For a sample sequence $S=(t_1,\ldots,t_m)$ and positive weights
$\omega=(\omega_1,\ldots,\omega_m)$, we denote
\begin{align*}
    \|u\|_{S,\omega}^2:=2\sum_{i=1}^m\omega_i|u(t_i)|^2.
\end{align*}
For a finite frequency set $A\subset\mathbb R$, let $V_A:=\operatorname{span}\{e^{2\pi\bi f t}:f\in A\}$.

Following the analysis in \cite{CP19_colt} and improved bounds in Lemma~\ref{lemma:bounds_Fourier_sparse_signals}, we can actually obtain a more general sampling lemma as described below. For completeness, we provide its proof at the end of this subsection.

\begin{lemma}\label{lem:finite_grid_embedding}
Let $\mathcal{N} \subset[-F,F]$.
There exists a explicit distribution $D_{\mathcal{F}}$ such that for $m=O\bigl(k^2 \log k \log |\mathcal{N}| \bigr)$,
independent samples $t_1,\ldots,t_m$ from $D_{\mathcal{F}}$, and weights $\omega_i=1/(m D_{\mathcal{F}}(t_i))$, 
with probability at least $0.995$ we have
\begin{equation}\label{eq:finite_grid_embedding}
\frac{1}{2}\|u\|_{[-1,1]}^2 \le \|u\|_{S,\omega}^2 \le \frac{3}{2}\|u\|_{[-1,1]}^2
\end{equation}
simultaneously for every signal $u$ having at most $k$ frequencies in $\mathcal{N}$.
\end{lemma}

Now we are ready to prove Corollary~\ref{cor:query_complexity_learning}.

\begin{algorithm}[t]
\caption{Recover k-sparse-Fourier signal}\label{alg:net_learning}
\begin{algorithmic}[1]
\Procedure{SparseFT}{$y,k,F,\epsilon$}
\State $\mathcal{N}\gets\frac{\epsilon}{10Ck^2} \cdot \mathbb Z\cap[-F,F]$
\State $m\gets \Theta(k^2 \log k \log |\mathcal{N}|)$
\State Draw $t_1, \cdots, t_m$ independently from the density $D_{\mathcal{F}}$ in Lemma~\ref{lem:finite_grid_embedding}
\State Query $y(t_1), \cdots, y(t_m)$ and set the corresponding weights
\ForAll{$A\subseteq\mathcal{N}$ with $|A|\le k$}
    \State $x_A\gets\argmin_{z\in V_A}\|y-z\|_{S,\omega}$
\EndFor
\State \Return $\widetilde x \gets\argmin_{A \subset \mathcal{N}, |A| \le k}\|y- x_A\|_{S,\omega}$
\EndProcedure
\end{algorithmic}
\end{algorithm}

With taking $\rho = \epsilon/10$, Theorem~\ref{thm:net_frequency}, applied with accuracy $\rho$, implies a signal $x'$ whose frequencies lie in
\begin{equation}
    \mathcal{N}=\frac{\rho}{Ck^2}\mathbb Z\cap[-F,F]
    \qquad\text{and}\qquad
    \|x-x'\|_{[-1,1]}\le \rho\|x\|_{[-1,1]}.
\end{equation}

We apply Lemma~\ref{lem:finite_grid_embedding} with $\mathcal{N}'=\mathcal{N}, k' = 2k$, which uses
\begin{equation}\label{eq:cp-sample-count}
m=O\left(k^2 \log k  \log |\mathcal{N}| \right) =O\left(k^2 \log k  \log\frac{kF}{\epsilon}\right)
\end{equation}
samples and guarantees
\begin{equation} \label{eq:cp-sample-bound}
    \frac{1}{2}\|\widetilde x-x'\|_{[-1,1]}^2 \le \|\widetilde x-x'\|_{S,\omega}^2 \le \frac{3}{2}\|\widetilde x-x'\|_{[-1,1]}^2,
\end{equation}
with probability 0.995.

Let $r=y-x'=\eta+(x-x')$.
Then $\|r\|_{[-1,1]} \le \|\eta\|_{[-1,1]} + \|x-x'\|_{[-1,1]}$. 
And by the definition of the norm $\| \cdot \|_{S,\omega}$,
\[
\mathbb E\bigl[\|r\|_{S,\omega}^2\bigr] =\int_{-1}^1|r(t)|^2\,\mathrm dt =\|r\|_{[-1,1]}^2.
\]
Markov's inequality therefore shows that, with probability at least $0.995$,
\begin{equation}\label{eq:noise_empirical_bound}
\|r\|_{S,\omega}\le\sqrt{200}\|r\|_{[-1,1]}.
\end{equation}
And the events in Lemma~\ref{lem:finite_grid_embedding} and \eqref{eq:noise_empirical_bound} hold simultaneously with probability at least $0.99$.

Assume both two above events happen. Then
\begin{align*}
\|\widetilde x-x'\|_{[-1,1]}
&\le\sqrt2\|\widetilde x-x'\|_{S,\omega} \tag{the lower bound in \eqref{eq:cp-sample-bound}}\\
&\le\sqrt2\bigl( \|\widetilde x-y\|_{S,\omega} +\|y-x'\|_{S,\omega}\bigr) \tag{the triangle inequality}\\
&\le 2 \sqrt2 \|y-x'\|_{S,\omega} \tag{$\|\widetilde x-y\|_{S,\omega} \le \|y-x'\|_{S,\omega}$}\\
&\le2\sqrt2\|r\|_{S,\omega} \tag{the definition of $r$}\\
&\le40\|r\|_{[-1,1]} \tag{the assumption}.
\end{align*}
Consequently,
\begin{align*}
\|\widetilde x-x\|_{[-1,1]} \le \|\widetilde x-x'\|_{[-1,1]} +\|x'-x\|_{[-1,1]}
\le O\left(\|\eta\|_{[-1,1]} +\epsilon \|x\|_{[-1,1]}\right).
\end{align*}

Finally, there are at most $(e |\mathcal{N}| / k)^k = (kF/\epsilon)^{O(k)}$ sets of at most $k$ grid frequencies. 
And each $A$ costs $\poly(m, k)$ time for linear regression.
Thus the total running time is $(kF/\epsilon)^{O(k)}$.

\begin{proof}[Proof of Lemma~\ref{lem:finite_grid_embedding}]
By replacing the original bounds with Property \ref{item:uniform_bound_on_interval} and \ref{item:leverage_bound_on_interval} from Lemma~\ref{lemma:bounds_Fourier_sparse_signals}, we can remove a log terms in Theorem~9.1 of \cite{CP19_colt}:

\begin{lemma}\label{lem:cp_importance_sampling}
There exists a constant $c=\Theta(1)$ such that the distribution whose
density with respect to the uniform distribution on $[-1,1]$ is
\begin{equation}\label{eq:def_cp_sampling_density}
D_{\mathcal{F}}(t)=
\begin{cases}
\frac{c}{(1-|t|)\log k}, & |t|\le1-\frac{1}{k},\\
c k, & |t|>1-\frac{1}{k}
\end{cases}
\end{equation}
guarantees, for every $k$-Fourier-sparse signal $x$,
\begin{equation}\label{eq:cp_importance_sampling}
\frac{|x(t)|^2}{D_{\mathcal{F}}(t)} \le O(k \log k)\|x\|_{[-1,1]}^2 \qquad (t\in[-1,1]).
\end{equation}
\end{lemma}

By Lemma~\ref{lem:cp_importance_sampling}, there is $\kappa=O(k \log k)$ such that every signal $x$ with at most $k$
frequencies satisfies
\begin{equation}\label{eq:cp-pointwise-bound}
\frac{|x(t)|^2}{D_{\mathcal{F}}(t)} \le \kappa \cdot \|x\|_{[-1,1]}^2.
\end{equation}
For a fixed nonzero $x$, we define
\[
Z_i(x):=\frac{2|x(t_i)|^2}
{D_{\mathcal{F}}(t_i)\|x\|_{[-1,1]}^2}.
\]
These variables are independent, have expectation $1$, and lie in $[0,2\kappa]$.
Moreover, the definition of the weights gives
\[
\frac{\|x\|_{S,\omega}^2}{\|x\|_{[-1,1]}^2} =\frac1m\sum_{i=1}^m Z_i(x).
\]

We state the following version of the Chernoff bound used in this proof.
\begin{lemma}[Chernoff Bound \cite{chernoff1952,tarjan09}]
\label{lem:cp_chernoff}
Let $Z_1,\ldots,Z_m$ be independent random variables such that
$0\le Z_i\le R$ and $\E[Z_i]=1$ for every $i$.  For every
$0<\theta<1/2$,
\begin{equation}\label{eq:cp_chernoff}
\Pr\left[
\left|\frac1m\sum_{i=1}^m Z_i-1\right|\ge\theta
\right]
\le2\exp\left(-\frac{\theta^2m}{3R}\right).
\end{equation}
\end{lemma}
Lemma~\ref{lem:cp_chernoff} with $\theta=1/5$ therefore implies
\begin{equation}\label{eq:fixed_signal_concentration}
\Pr\left[
\left|\frac{\|x\|_{S,\omega}^2}{\|x\|_{[-1,1]}^2}-1\right|
>\frac15
\right]
\le 2e^{- \Omega(m/\kappa)}
\end{equation}.

For every $A\subseteq\mathcal N$ with $1\le|A|\le k$, we construct a $1/10$-net $\mathcal M_A$ of the unit sphere of $V_A$ in the $\|\cdot\|_{[-1,1]}$ norm.  
The volumetric argument shows $|\mathcal M_A| \le 2^{O(k)}$.  
Hence, the total number of net points is at most $|\mathcal{N}|^{O(k)}$.
A union bound in \eqref{eq:fixed_signal_concentration} shows that a sufficiently large
\[
m=O\bigl(k^2 \log k \log |\mathcal N|\bigr)
\]
makes the estimate in \eqref{eq:fixed_signal_concentration} hold for every point of every $\mathcal M_A$ with probability at least $0.995$.
From the property of the net, for any $x$ having at most $k$ frequencies in $\mathcal{N}$,
$\|x\|^2_{S,\omega}=(1 \pm \frac{1}{2}) \|x\|_{[-1,1]}^2$.
\end{proof}

\section{Filters and Locality} \label{sec:Filters}
One may assume that all powers in the filters are rounded up to the least even integer.
In the below proofs, we use the following bounds on the $\sinc$ function:
\begin{fact} \label{fact:sinc_bounds}
    Recall that $\sinc(x):=\frac{\sin(\pi x)}{\pi x}$. We denote $a:=\frac{1.2}{\pi}$.
    \begin{enumerate}
    \item For any $|x| \ge a$, $|\sinc(x)| \le \frac{1}{\pi |x|}$.
    \item For any $|x| \le a$, $\sinc(x) \in \left[1 - \frac{\pi^2 |x|^2}{6}, 1 - \frac{\pi^2 |x|^2}{10}\right]$.
    \end{enumerate}
\end{fact}

\begin{corollary} \label{cor:sinc_bounds_plus}
    For every $p > 0$, every even integer $q \ge 2$,
    \begin{enumerate}
        \item $\int_{|x| \le a/p} \sinc(p x)^q \mathrm{d} x = \Theta\left(\frac{1}{p\sqrt{q}}\right)$;
        \item $\int_{|x| \ge t} \sinc(p x)^q \mathrm{d} x \le O\left(\frac{1}{p}(\pi p t)^{-q+1}\right)$ with $t \ge a/p$.
    \end{enumerate}
\end{corollary}

\begin{proof}
By Fact~\ref{fact:sinc_bounds},
\begin{align*}
    \sinc(p x)^q = \exp(-\Theta(p^2 q x^2)) \quad (|x|\le a/p),
    \qquad
    \sinc(p x)^q \le (\pi p |x|)^{-q} \quad (|x|\ge a/p).
\end{align*}
Therefore,
\begin{align*}
    \int_{|x| \le a/p} \sinc(p x)^q \mathrm{d} x
    = & \int_{|x| \le a/p\sqrt{q}} \sinc(p x)^q \mathrm{d} x + \int_{a/p\sqrt{q} \le |x| \le a/p} \sinc(p x)^q \mathrm{d} x \\
    = & \int_{|x| \le a/p\sqrt{q}} \Theta(1) \mathrm{d} x + \int_{a/p\sqrt{q} \le |x| \le a/p} \exp(-\Theta(p^2 q x^2)) \mathrm{d} x = \Theta\left(\frac{1}{p\sqrt{q}}\right).
\end{align*}
And
\begin{align*}
    \int_{|x| \ge t} \sinc(p x)^q \mathrm{d} x \le \int_{|x| \ge t} (\pi p |x|)^{-q} \mathrm{d} x \le O\left(\frac{1}{p}(\pi p t)^{-q+1}\right).
\end{align*}
\end{proof}

\subsection{Proof of Lemma~\ref{lemm:construction_H}}\label{sec:proof_H_delta}
We finish the proof of Lemma~\ref{lemm:construction_H} about $(H_{\ell,\delta},\wh{H_{\ell,\delta}})$ in this section. The construction of this filter originates from \cite{CP19_ICALP}. We restate it here with the parameter scale needed for our proof.

Given the sparsity $\ell$ and error $\delta$, let $C=O(1)$, $S=\ell^2 \log \ell$, $\alpha_H = 1 - \frac{\delta}{C \ell^2}$ and
\begin{equation}\label{eq:gH-def}
    g_H(t):= \sinc\left(\frac{C \ell^2}{\delta}t\right)^{C\log\frac{\ell}{\delta}}
    \cdot \prod_{j=0}^{\left\lceil2 \log \ell\right\rceil} \sinc\left(C \ell^2 2^{-j} t\right)^{C 2^j \log \ell}.
\end{equation}
Then 
\begin{align}
H_{\ell,\delta}(t) & :=s_0 \cdot g_H(t) * \rect_{2 \alpha_H}(t)
\end{align}
where $s_0 > 0$ is chosen so that $H_{\ell,\delta}(0)=1$.

\begin{claim} \label{claim:H-bounds}
For $\ell \ge 1$ and $0 < \delta<1/2$, the filter $H = H_{\ell,\delta}$ has the following properties:
\begin{enumerate}
    \item $s_0 = \Theta\left(\frac{C \ell^2}{\delta} \sqrt{C \log \frac{\ell}{\delta}}\right)$.
    \item $|1 - H(t)| \le \left(\frac{\delta}{\ell}\right)^{\Omega(C)}$ for $|t| \le 1 - \frac{2\delta}{C \ell^2}$.
    \item $|H(t)| \le O(1)$ for $1 - \frac{2\delta}{C \ell^2} \le |t| \le 1$.
    \item $|H(t)| \le \left(\frac{\delta}{\ell}\right)^{\Omega(C)} \exp\left(-\Omega(C \ell^2\log(\ell)(|t|-1))\right)$ for $1 \le |t| \le 1 + 1/C$.
    \item $|H(t)| \le \left(\frac{\delta}{\ell}\right)^{\Omega(C)} \pi^{-\Omega\left(C \ell^2\log \ell\right)}$ for $1+1/C \le |t| \le 2$
    \item $|H(t)| \le \left(\frac{\delta}{\ell}\right)^{\Omega(C)} |C \pi t|^{-\Omega\left(C \ell^2\log \ell\right)}$ for $|t| \ge 2$.
    \item $\supp (\wh{H}) \subseteq [-\Delta_H, \Delta_H]$ with $\Delta_H := C^2 \cdot (\frac{\ell^2 \log \ell/\delta}{\delta} + 2\ell^2\log^2 \ell)$.
\end{enumerate}
\end{claim}



\begin{proof}
We first determine $s_0$. Applying Corollary~\ref{cor:sinc_bounds_plus} to the first $\sinc$ factor, and using that all other factors are bounded by $1$, we have
\begin{align*}
    \int g_H(v) dv = \Theta\left(\frac{\delta}{\ell^2\sqrt{\log\frac{\ell}{\delta}}}\right)
    ~\text{and}~
    \int_{|v|\ge \delta/(C\ell^2)}g_H(v) dv
    \le \left(\frac{\delta}{\ell}\right)^{\Omega(C)}.
\end{align*}
Thus $s_0 = \Theta\left(\frac{\ell^2}{\delta}\sqrt{\log\frac{\ell}{\delta}}\right)$.

If $|t| \le 1 - 2\delta/(C\ell^2)$,
then the interval $[t-\alpha_H,t+\alpha_H]$ contains $[-\delta/(C\ell^2),\delta/(C\ell^2)]$ and misses only the above tail of $g_H$. So
\begin{align*}
    |H(t)-1| \le 2 s_0 \int_{|v|\ge \delta/(C\ell^2)} g_H(v) dv \le \left(\frac{\delta}{\ell}\right)^{\Omega(C)}.
\end{align*}
The normalization also implies $|H(t)|\le s_0\int_\R g_H(v)dv=O(1)$ for all $t$.

It remains to bound the tails. By symmetry, assume $t \ge 1$.
Then $H(t) \le 2 s_0 \cdot \int_{t-\alpha_H}^\infty g_H(v) dv$.
The first $\sinc$ factor implies, for $ v\ge \delta/(C\ell^2)$, $\sinc\left(\frac{C\ell^2}{\delta}v\right)^{C\log\frac{\ell}{\delta}} \le \left(\frac{\delta}{\ell}\right)^{\Omega(C)}$.
On $\delta/(C\ell^2)\le v\le 1/(C\ell^2)$, the extra polynomial loss from the first factor absorbs the missing $\exp(-\Omega(C\ell^2\log(\ell)v))$ factor.
And the $j$-th multiscale factor becomes active once $C \ell^2 v \ge 2^j$.  Hence, for
$\frac{2^j}{C\ell^2}\le v\le \frac{2^{j+1}}{C\ell^2}$ with
$0\le j\le \lceil2\log \ell\rceil$,
\begin{align*}
    g_H(v) 
    & \le
    \sinc\left(\frac{C\ell^2}{\delta}v\right)^{C\log\frac{\ell}{\delta}}
    \sinc\left(C \ell^2 2^{-j} v\right)^{C 2^j \log \ell} \\
    & \le \left(\frac{\delta}{\ell}\right)^{\Omega(C)} \exp(-\Omega(C 2^j \log \ell) ) \\
    & \le \left(\frac{\delta}{\ell}\right)^{\Omega(C)} \exp(-\Omega(C^2 \ell^2\log(\ell) v)),
\end{align*}
where the second step is by $C \ell^2 2^{-j} v \ge 1$ and the last step follows from $C \ell^2 v \le 2^{j+1}$.
Integrating shows
\begin{equation}
    \label{eq:bound on Ht (Appendix)}
    H(t)
    \le \left(\frac{\delta}{\ell}\right)^{\Omega(C)} \exp\left(-\Omega(C \ell^2\log(\ell)(t-1))\right),
\end{equation}
which proves the claimed near-boundary bound for $t \le 1+1/C$.
If $t \ge 1 + 1/C$,
\begin{align*}
    H(t)
    \le \left(\frac{\delta}{\ell}\right)^{\Omega(C)} (C\pi(t-1))^{-\Omega\left(C \ell^2\log \ell\right)} 
    \le \begin{cases}
        \left(\frac{\delta}{\ell}\right)^{\Omega(C)} \pi^{-\Omega\left(C \ell^2\log \ell\right)}, & 1+1/C \le t \le 2, \\
        \left(\frac{\delta}{\ell}\right)^{\Omega(C)} (C\pi t)^{-\Omega\left(C \ell^2\log \ell\right)}, & t \ge 2.
    \end{cases}
\end{align*}

The Fourier support bound follows by summing the widths of the box functions:
\begin{align*}
    \frac{\ell^2}{\delta}\log\frac{\ell}{\delta} + \sum_{j=0}^{\left\lceil2 \log \ell\right\rceil} \ell^2 \log \ell \le O\left( \frac{\ell^2}{\delta}\log\frac{\ell}{\delta} + \ell^2\log^2 \ell\right).
\end{align*}
\end{proof}

\begin{lemma} \label{lemma:accuracy-tunable-localization}
For every $\ell \ge 1$, $0 < \delta < 1/2$ and $\ell$-Fourier-sparse signal $x$,
\begin{align}
    \int_{-1}^{1}|H_{\ell, \delta}(t)x(t)|^2 dt
    & \ge (1-\delta)\int_{-1}^{1}|x(t)|^2 dt, \label{eq:inside-goal}\\
    \int_{\R\setminus[-1,1]}|H_{\ell, \delta}(t)x(t)|^2 dt
    & \le \delta\int_{-1}^{1}|x(t)|^2 dt, \label{eq:outside-goal}.
\end{align}
\end{lemma}

\begin{proof}
For convenience, denote $H_{\ell, \delta}$ by $H$.

Let $I:=\left[-1 + \frac{2\delta}{C \ell^2}, 1 - \frac{2\delta}{C \ell^2}\right]$.
By Claim~\ref{claim:H-bounds}, $|1-H(t)| \le \left(\frac{\delta}{\ell}\right)^{\Omega(C)}$ on $I$.  
Hence, $\int_I|H(t)x(t)|^2 dt \ge (1 - \left(\delta/\ell\right)^{\Omega(C)})^2 \int_I|x(t)|^2 dt$.
Since $[-1,1]\setminus I$ has length $\frac{4\delta}{C \ell^2}$,
\begin{align*}
    \int_{[-1,1]\setminus I}|x(t)|^2 dt
    \le \frac{4\delta}{C \ell^2} \sup_{|t|\le1}|x(t)|^2
    \le O\left(\frac{\delta}{C}\right) \int_{-1}^{1}|x(t)|^2 dt
\end{align*}
where the last inequality is by Property \ref{item:uniform_bound_on_interval} of Lemma~\ref{lemma:bounds_Fourier_sparse_signals}
Therefore, $\int_{-1}^{1}|H(t)x(t)|^2 dt \ge (1- \delta^2) \int_I|x(t)|^2 dt \ge (1-\delta)\int_{-1}^{1}|x(t)|^2 dt$
after increasing the constant $C$ in the definition of $H$. 

The remaining part is to bound the outside energy. 
For $1 \le |t| \le 1 + 1/C$, Property~\ref{item:exponential_bound_outside_interval} of Lemma~\ref{lemma:bounds_Fourier_sparse_signals} shows that
\begin{align*}
    |x(t)|^2 
    \le \poly(\ell)\|x\|_{[-1,1]}^2 \exp\Big(O(\ell^2 \log \ell)(|t|-1)\Big).
\end{align*}
Hence, for the near-boundary range, Claim~\ref{claim:H-bounds} implies
\begin{align*}
    |H(t)x(t)|^2
    & \le \poly(\ell)\left(\frac{\delta}{\ell}\right)^{\Omega(C)} 
    \exp(-\Omega(C \ell^2 \log \ell) (|t|-1)) \|x\|_{[-1,1]}^2.
\end{align*}
By the far-tail range estimation in Claim~\ref{claim:H-bounds} and Property~\ref{item:polynomial_bound_outside_interval} of Lemma~\ref{lemma:bounds_Fourier_sparse_signals}, 
\begin{align*}
    |H(t)x(t)|^2
    \le \poly(\ell)\left(\frac{\delta}{\ell}\right)^{\Omega(C)}
    |\pi t / 2|^{-\Omega(C \ell^2 \log \ell)} \|x\|_{[-1,1]}^2.
\end{align*}
Increasing $C$ makes the tail energy at most $\delta \cdot \|x\|_{[-1,1]}^2$.  
This proves \eqref{eq:outside-goal}. 
\end{proof}

\section{Filters and Orthogonality}

\subsection{Proof of Claim~\ref{clm:almost_orthogonal_clusters}}\label{sec:proof_clm_almost_orthogonal}

Given $1 \le \ell \le r$ and $0<\delta<1/2$, let $C=O(1), \alpha_M := 1 - \frac{\delta^2}{C\ell^2}$ and
\begin{align*}
    g_M(t):=\sinc\left(\frac{C\ell^2}{\delta^2}t\right)^{ C\left(r+\log\frac{1}{\delta}\right)}.
\end{align*}
We define the localizing filter
\begin{equation}\label{eq:L-def}
    M_{\ell, r, \delta}(t) := s_M\cdot (g_M*\rect_{2 \alpha_M})(t),
\end{equation}
where $s_M>0$ is chosen so that $M_{\ell, r, \delta}(0)=1$.

\begin{claim}\label{claim:L-bounds}
For $1 \le \ell \le r$ and $0 < \delta < 1/2$, the filter $M = M_{\ell, r, \delta}$ satisfies the following properties:
\begin{enumerate}
    \item $|1 - M(t)| \le \delta^{\Omega(C)}$ for $|t|\le 1 - 2 \delta^2 / C \ell^2$.
    \item $|M(t)|\le O(1)$ for $1 - 2 \delta^2 / C \ell^2 \le |t|\le 1$.
    \item $|M(t)| \le \delta^{\Omega(C)} \pi^{-\Omega(Cr)}$ for $1 \le |t| \le 2$.
    \item $|M(t)| \le \delta^{\Omega(C)} |C t|^{-\Omega(Cr)}$ for $|t| \ge 2$.
    \item $\supp (\wh{M}) \subseteq [-\Delta_L, \Delta_L]$ with $\Delta_L := C^2 \left(\frac{\ell^2}{\delta^2}\left(r + \log\frac1\delta\right)\right)$.
\end{enumerate}
\end{claim}

\begin{proof}
Corollary~\ref{cor:sinc_bounds_plus} implies
\begin{equation}\label{eq:L-normalization}
    \int_{|v| \le \delta^2/(C\ell^2)} g_M(v) dv
    = \Theta\left(\frac{\delta^2}{\ell^2\sqrt{r+\log(1/\delta)}}\right).
\end{equation}
\begin{align*}
    \int_{|v| \ge \delta^2/(C\ell^2)}g_M(v) dv
    \le
    O\left(\frac{\delta^2}{\ell^2}\right)
    \exp\left(-\Omega\left(r+\log\frac1\delta\right)\right).
\end{align*}
Thus $s_M=\Theta\left(\ell^2\sqrt{r+\log(1/\delta)}/\delta^2 \right)$ and $|M(t)|\le s_M\int_\R g_M(v)dv=O(1)$ for all $t$.

If $|t|\le 1-2\delta^2/(C\ell^2)$, the interval $[t-\alpha_M,t+\alpha_M]$ contains $[-\delta^2/(C\ell^2),\delta^2/(C\ell^2)]$, and hence
\[
    |M(t)-1| \le O(s_M)\int_{|v|\ge \delta^2/(C\ell^2)}g_M(v) dv \le \delta^{\Omega(C)}.
\]

For $|t| \ge 1$, Corollary~\ref{cor:sinc_bounds_plus} applied from
$|t| - \alpha_M = |t| - 1 + \delta^2/(C\ell^2)$ gives
\begin{align*}
    M(t)
    \le O\left(\sqrt{r+\log\frac1\delta}\right)
    \left(\pi\frac{C\ell^2}{\delta^2}\left(|t|-1+\frac{\delta^2}{C\ell^2}\right)\right)^{-\Omega(C(r+\log(1/\delta)))}.
\end{align*}
For $1\le |t|\le2$, the expression $\pi\frac{C\ell^2}{\delta^2}\left(|t|-1+\frac{\delta^2}{C\ell^2}\right)$ is at least $\pi$, so $M(t) \le \pi^{-\Omega(C(r+\log(1/\delta)))} \le \delta^{\Omega(C)} \pi^{-\Omega(Cr)}$.
For $|t| \ge 2$, the term $|t|-1+\delta^2/(C\ell^2)$ is at least $|t|/2$. Thus $\pi\frac{C\ell^2}{\delta^2}\left(|t|-1+\frac{\delta^2}{C\ell^2}\right) \ge C|t|$ and $M(t) \le (C|t|)^{-\Omega(C(r+\log(1/\delta)))} \le \delta^{\Omega(C)} (C|t|)^{-\Omega(Cr)}$.

The Fourier support bound follows by summing the widths of the box functions $C^2 \left(\frac{\ell^2}{\delta^2}\left(r + \log\frac1\delta\right)\right)$.
\end{proof}

\begin{lemma} \label{lem:orthogonality_by_localizer}
For $0<\delta<1/2$ and two signals of Fourier sparsity $\ell$ and $r$
separately with $\ell \le r$,
\begin{align*}
    w(t):=\sum_{j=1}^{\ell}\alpha_j e^{2\pi\bi f'_jt} \qquad \text{ and }
    \qquad
    z(t):=\sum_{j=1}^{r}\beta_j e^{2\pi\bi f_jt},
\end{align*}
if the distance between their frequencies $\min_{j,j'} |f_j-f'_{j'}| \ge C_H\frac{\ell^2 (r + \log 1/\delta)}{\delta^2}$ for some constant $C_H$, then
\begin{align*}
    |\langle w, z\rangle_{[-1, 1]}| \le \delta \cdot \|w\|_{[-1,1]} \cdot \|z\|_{[-1,1]}..
\end{align*}
\end{lemma}

\begin{proof}
Let $M:=M_{\ell, r, \delta}$ be the localizing filter from \eqref{eq:L-def}.
We decompose the inner product $\langle w,z\rangle$ into a filtered term and a flatness-error term:
\begin{align}
    \langle w, z\rangle_{[-1, 1]} = \langle M w, z\rangle_{[-1, 1]} + \langle (1-M) w, z\rangle_{[-1, 1]}.
    \label{eq:filter_decomposition}
\end{align}

We begin with bounding the first term in \eqref{eq:filter_decomposition}. By Parseval's identity,
\begin{align*}
    \langle M w, z\rangle_{[-1, 1]} = & \langle M w, z\rangle_{[-\infty, +\infty]} - \langle M w, z\rangle_{[-\infty, +\infty] \setminus [-1, 1]} \\
    = & \langle \wh{M w}, \wh{z}\rangle_{[-\infty, +\infty]} - \langle M w, z\rangle_{[-\infty, +\infty] \setminus [-1, 1]}
\end{align*}
The distribution $\widehat{w\overline{z}}$ is supported on the frequency differences $f'_{j'}-f_j$.
By Claim~\ref{claim:L-bounds}, $\langle \wh{M w}, \wh{z}\rangle_{[-\infty, +\infty]}$ therefore vanishes whenever the separation constant $C_H$ is sufficiently large.

Applying Property~\ref{item:polynomial_bound_outside_interval} of Lemma~\ref{lemma:bounds_Fourier_sparse_signals} to $w$ and $z$ implies that, for $|t| \ge 1$
\begin{align*}
    \frac{|w(t)\overline{z(t)}|}{\|w\|_{[-1,1]}\|z\|_{[-1,1]}} \le (e(|t|+1))^{r+\ell} \le \min \left\{ (3e)^{r+\ell}, (\tfrac{3}{2} e |t|)^{r+\ell} \right\}.
\end{align*}
With the tail bound in Claim~\ref{claim:L-bounds}, we have $|M(t)| \le \delta^{\Omega(C)} \cdot \max\{\pi^{-\Omega(Cr)}, |C t|^{-\Omega(Cr)}\}$.
Hence, 
\begin{align*}
    \left|\int_{\R\setminus[-1,1]}M(t)w(t)\overline{z(t)} \mathrm{d}t\right| \le & \|w\|_{[-1,1]}\|z\|_{[-1,1]} \cdot \left|\int_{|t| \ge 1} M(t) \cdot \min \left\{ (3e)^{r+\ell}, (\tfrac{3}{2} e |t|)^{r+\ell} \right\} \mathrm{d}t\right| \\
    \le & \delta^{\Omega(C)} \|w\|_{[-1,1]}\|z\|_{[-1,1]} \cdot \left[ \left|\int_{1 \le |t| \le 2} \pi^{-\Omega(Cr)} (3e)^{r+\ell} \mathrm{d}t\right| + \left|\int_{|t| \ge 2} |C t|^{-\Omega(Cr)} (\tfrac{3}{2} e |t|)^{r+\ell} \mathrm{d}t\right| \right]
\end{align*}
Thus the filtered term $|\langle M w, z\rangle_{[-1, 1]}|$ is at most $(\delta/2)\|w\|_{[-1,1]}\|z\|_{[-1,1]}$ for a large $C$.

We next bound the second term in \eqref{eq:filter_decomposition}. By Cauchy-Schwarz inequality,
\begin{equation}
    \label{eq:filter_decomposition_part2_cauchy_schwarz}
    \left|\int_{-1}^{1}[1-M(t)]w(t)\overline{z(t)} \mathrm{d}t\right|
    \le \|z\|_{[-1,1]} \cdot \left(\int_{-1}^{1}[1-M(t)]^2|w(t)|^2 \mathrm{d}t\right)^{1/2}.
\end{equation}

Let $I=[-1+2\delta^2/(C\ell^2),1-2\delta^2/(C\ell^2)]$.
For $t\in I$, Claim~\ref{claim:L-bounds} shows that $|1 - M(t)| \le \delta^{\Omega(C)}$. On the two boundary intervals, Property~\ref{item:uniform_bound_on_interval} in Lemma~\ref{lemma:bounds_Fourier_sparse_signals} gives $|w(t)|^2 \le O(\ell^2) \|w\|_{[-1,1]}^2$.
Consequently,
\begin{equation}
    \label{eq:filter_decomposition_part2_property1}
    \int_{-1}^{1}[1-M(t)]^2|w(t)|^2 \mathrm{d}t
    \le \delta^{\Omega(C)}\|w\|_{[-1,1]}^2 +O\left(\frac{\delta^2}{C\ell^2}\right) O(\ell^2)\|w\|_{[-1,1]}^2
    \le O(\delta^2)\|w\|_{[-1,1]}^2.
\end{equation}
So \eqref{eq:filter_decomposition_part2_cauchy_schwarz} and \eqref{eq:filter_decomposition_part2_property1} gives
\begin{align*}
    |\langle (1-M) w, z\rangle_{[-1, 1]}| = \left|\int_{-1}^{1}[1-M(t)]w(t)\overline{z(t)} \mathrm{d}t\right|
    \le \frac{\delta}{2}\|w\|_{[-1,1]}\|z\|_{[-1,1]}.
\end{align*}
Combining the above bounds proves the claim.
\end{proof}

\begin{lemma}
\label{lem:filtered_orthogonality_by_localizer}
For $\ell \le r\le k$ and two signals of Fourier sparsity $\ell$ and $r$ separately,
\begin{align*}
    w(t):=\sum_{j=1}^{\ell}\alpha_j e^{2\pi\bi f'_jt} \qquad \text{ and }
    \qquad
    z(t):=\sum_{j=1}^{r}\beta_j e^{2\pi\bi f_jt},
\end{align*}
if the distance between their frequencies $\min_{j,j'} |f_j-f'_{j'}| \ge \min\left\{ C_H\frac{\ell^2 (r + \log 1/\delta) \log^2 k}{\delta^2}, 2\Delta_k \right\}$ for some constant $C_H$, then
\begin{align*}
    |\langle H_k w,H_k z\rangle|
    \le \delta \cdot \|H_k w\|_2 \cdot \|H_k z\|_2.
\end{align*}
\end{lemma}

\begin{proof}
Let $d:=\min_{j,j'}|f_j-f'_{j'}|$.  
If $d \ge 2 \Delta_k$, then $\widehat{H_k w}$ and $\widehat{H_k z}$ have disjoint supports, so Parseval's identity gives $\langle H_kw,H_kz\rangle=0$.
It remains to consider the case $d < 2\Delta_k$. 
Recall that $\Delta_k = C^2 (\frac{k^2 \log k/\epsilon}{\epsilon} + k^2 \log k \cdot \log (k^2 \log k) )$.
In this case, $d \ge C_H\frac{\ell^2r}{\delta^2}\log^2 k$ implies
\begin{align}
    \delta^2 \ge \frac{C_H \epsilon\ell^2r}{2 C^2 k^2}.
\end{align}

We first compare $\langle H_k w,H_k z\rangle$ and $\langle w, z\rangle$ in $[-1, 1]$.
Let $I = \left[-1 + \frac{2 \epsilon}{C k^2}, 1 - \frac{2 \epsilon}{C k^2}\right]$
One can decompose the difference into two parts:
\begin{align}
    |\langle H_k w,H_k z\rangle_{[-1, 1]} - \langle w, z\rangle_{[-1, 1]}|
    = & |\langle (H_k^2 - 1) w, z\rangle_{[-1, 1]}| \notag \\
    \le & |\langle (H_k^2 - 1) w, z\rangle_{I}| + |\langle (H_k^2 - 1) w, z\rangle_{[-1, 1] \setminus I}| \notag \\
    \le & \| (H_k^2 - 1) w \|_{I} \| z \|_{I} + \| (H_k^2 - 1) w \|_{[-1, 1] \setminus I} \| z \|_{[-1, 1] \setminus I}, \label{eq:H_k_inner_decomposition}
\end{align}
where the last inequality is by Cauchy-Schwarz inequality.

By Claim~\ref{claim:H-bounds}, $|H_k(t)^2 - 1| = (\epsilon / k)^{\Omega(C)} \le O(\delta^2)$ for $t \in I$. Thus the first term of $\eqref{eq:H_k_inner_decomposition}$ is at most $O(\delta) \| w \|_{[-1, 1]} \| z \|_{[-1, 1]}$. 
As the property~\ref{item:uniform_bound_on_interval} of Lemma~\ref{lemma:bounds_Fourier_sparse_signals} shows that $w(t) \le O(\ell) \| w \|_{[-1, 1]}$ for $t \in [-1, 1]$, we bound $\| (H_k^2 - 1) w \|_{[-1, 1] \setminus I}^2$ by $\frac{4 \epsilon}{C k^2} \cdot O(\ell^2) \| w \|_{[-1, 1]}^2 \le O(\delta^2) \| w \|_{[-1, 1]}^2$. So the later term of \eqref{eq:H_k_inner_decomposition} is at most $O(\delta) \| w \|_{[-1, 1]} \| z \|_{[-1, 1]}$. Thus,
\begin{equation}
    |\langle H_k w,H_k z\rangle_{[-1, 1]} - \langle w, z\rangle_{[-1, 1]}| \le O(\delta) \| w \|_{[-1, 1]} \| z \|_{[-1, 1]}. \label{eq:wz_inner_product_H_k_difference}
\end{equation}

Tail bounds in Lemma~\ref{lemma:bounds_Fourier_sparse_signals} and Claim~\ref{claim:H-bounds} also imply that $\|H_k w\|_{(-\infty,\infty) \setminus [-1, 1]}^2 \le \poly(\ell) (\epsilon / k)^{\Omega(C)} \| w \|_{[-1, 1]}^2$ and $\|H_k z\|_{(-\infty,\infty) \setminus [-1, 1]}^2 \le \poly(r) (\epsilon / k)^{\Omega(C)} \| z \|_{[-1, 1]}^2$. 
So by the Cauchy--Schwarz inequality,
\begin{equation}
    |\langle H_k w,H_k z\rangle_{(-\infty,\infty) \setminus [-1, 1]}| \le \poly(\ell) \poly(r) (\epsilon / k)^{\Omega(C)} \cdot \| w \|_{[-1, 1]} \| z \|_{[-1, 1]} \le O(\delta) \cdot \| w \|_{[-1, 1]} \| z \|_{[-1, 1]}. \label{eq:H_k_wz_inner_product_outside_bound}
\end{equation}

Therefore,
\begin{align*}
    & |\langle H_k w,H_k z\rangle_{(-\infty,\infty)}| \\
    \le & |\langle H_k w,H_k z\rangle_{[-1, 1]}| + |\langle H_k w,H_k z\rangle_{(-\infty,\infty) \setminus [-1, 1]}| \\
    \le & |\langle H_k w,H_k z\rangle_{[-1, 1]} - \langle w,z\rangle_{[-1, 1]}| + |\langle w,z\rangle_{[-1, 1]}| + |\langle H_k w,H_k z\rangle_{(-\infty,\infty) \setminus [-1, 1]}| \\
    \le & O(\delta) \| w \|_{[-1, 1]} \| z \|_{[-1, 1]} + |\langle w,z\rangle_{[-1, 1]}| \tag{\ref{eq:wz_inner_product_H_k_difference} and \ref{eq:H_k_wz_inner_product_outside_bound}} \\
    \le &  O(\delta) \| w \|_{[-1, 1]} \| z \|_{[-1, 1]}. \tag{Lemma~\ref{lem:orthogonality_by_localizer}} 
\end{align*}

Applying Lemma~\ref{lemma:accuracy-tunable-localization} to $w$ and $z$, we have
\begin{align*}
    \|w\|_{[-1,1]} \le O(1)\|H_k w\|_2,
    \qquad
    \|z\|_{[-1,1]} \le O(1)\|H_k z\|_2.
\end{align*}
That proves the lemma.
\end{proof}

\subsection{Proof of Claim~\ref{clm:almost_orthogonal_clusters_under_conjecture}} \label{sec:almost_orthogonal_clusters_under_conjecture}
Given the sparsity $1 \le \ell \le r$ and error $\delta < 1 / 2$, let $C=O(1)$, $\alpha_{M'} := 1 - \frac{\delta^2}{C \ell^2}$ and
\begin{align*}
    g_{M'}(t)
    := \sinc \left(\frac{C \ell^2}{\delta^2}t\right)^{C\log(r/\delta)}
    \prod_{i=0}^{\lceil\log r\rceil}
    \sinc \left( \frac{Ct}{\delta^2/\ell^2+4^i/r^2} \right)^{C 2^i}
\end{align*}
The square-root localizer is
\begin{equation}\label{eq:sqrt-localizer-def}
    M'_{\ell, r, \delta}(t)  := s_{M'} \cdot \bigl(g_{M'}*\rect_{2\alpha_{M'}}\bigr)(t),
\end{equation}
	where $s_{M'} > 0$ is chosen so that $M'_{\ell, r, \delta}(0) = 1$.

\begin{claim}
\label{claim:sqrt-localizer-bounds}
For $M' = M'_{\ell, r, \delta}$, the following
properties hold:
\begin{enumerate}
    \item $|1-M'(t)| \le (\delta / r)^{\Omega(C)}$ for $|t| \le 1 - \frac{2\delta^2}{C \ell^2}$.
    \item For $1 \le |t| \le 1 + 1 / C$, $M'(t) \le (\delta / r)^{\Omega(C)} \cdot \exp\left(-\Omega (C r \sqrt{|t| - 1})\right)$.
    \item For $ |t| \ge 1 + 1 / C$, $M'(t) \le (\delta / r)^{\Omega(C)} \cdot |\pi t / 2|^{-\Omega(C r)}$.
    \item $\supp(\widehat{M'}) \subseteq [-\Delta_{M'}, \Delta_{M'}]$ with $\Delta_{M'} := C^2 \left(\frac{4 \ell r}{\delta} + \frac{\ell^2 \log(r/\delta)}{\delta^2}\right) $. 
\end{enumerate}
\end{claim}

\begin{proof}
For the Fourier support, $\widehat{M'}$ is supported in an interval whose radius is at most
\begin{align*}
    C^2 \left( \frac{\ell^2}{\delta^2} \cdot \log\frac{r}{\delta} + \sum_{i=0}^{\lceil\log r\rceil} \frac{2^i}{\delta^2/\ell^2+4^i/r^2} \right).
\end{align*}
For the latter summation, by $\frac{1}{a + b} \le \min\{\frac{1}{a}, \frac{1}{b}\}$ for $a, b > 0$,
\begin{align*}
    \sum_{i=0}^{\lceil\log r\rceil} \frac{2^i}{\delta^2/\ell^2+4^i/r^2} \le \sum_{i=0}^{\lceil\log r\rceil} \min \left\{ \frac{2^i \ell^2}{\delta^2}, \frac{r^2}{2^i} \right\}. 
\end{align*}
Let $j$ be the largest integer such that $\frac{2^j \ell^2}{\delta^2} \le \frac{r^2}{2^j}$. If $j < 0$, then $\frac{\ell}{\delta} \ge r$ and
\begin{align*}
    \sum_{i=0}^{\lceil\log r\rceil} \min \left\{ \frac{2^i \ell^2}{\delta^2}, \frac{r^2}{2^i} \right\} = \sum_{i=0}^{\lceil\log r\rceil} \frac{r^2}{2^i} \le 2 r^2 \le \frac{2 \ell r}{\delta}.
\end{align*}
Otherwise, since $2^j \le \frac{\delta r}{\ell} \le 2^{j+1}$,
\begin{align*}
    \sum_{i=0}^{\lceil\log r\rceil} \min \left\{ \frac{2^i \ell^2}{\delta^2}, \frac{r^2}{2^i} \right\} = \sum_{i=0}^{j} \frac{2^i \ell^2}{\delta^2} + \sum_{i= j+1}^{\lceil\log r\rceil} \frac{r^2}{2^i} \le \frac{2^{j+1} \ell^2}{\delta^2} + \frac{r^2}{2^j} \le \frac{4 \ell r}{\delta}.
\end{align*}
Combining the two cases proves the stated support bound.


By Corollary~\ref{cor:sinc_bounds_plus}, $\int_{|v|\ge \delta^2/(2\ell^2)} g_{M'}(v)dv \le \left(\delta / r\right)^{\Omega(C)}$. Moreover, Fact~\ref{fact:sinc_bounds} implies that 
\begin{align*}
    g_{M'}(t) \ge \exp \left[ - O(1) \cdot C^4 \left(  \frac{\ell^2\log(r/\delta)}{\delta^2} +\sum_{i=0}^{\lceil\log r\rceil} \frac{2^i}{\delta^2/\ell^2+4^i/r^2} \right)^2 t^2 \right] \ge \exp \Big[ - O(\Delta_{M'}^2 t^2) \Big],
\end{align*}
for $|t| = O(1/\Delta_{M'})$.
Hence, $\int_{\R} g_{M'}(v)dv \ge \int_{-1/\Delta_{M'}}^{1/\Delta_{M'}} g_{M'}(v)dv = \Omega(1 / \Delta_{M'})$ and $s_{M'} = O(\Delta_{M'}) = C^2 \cdot (r/\delta)^{O(1)}$.

If $|t| \le 1 - 2\delta^2 / C\ell^2$, then $[t-\alpha_{M'},t+\alpha_{M'}]$ contains $[-\delta^2/C\ell^2,\delta^2/C \ell^2]$, and hence $|1-M'(t)| \le (\delta / r)^{\Omega(C)}$.
The normalization guarantees $|M'(t)|\le s_{M'}\int_\R g_{M'}(v)dv=O(1)$ for all $t$.

Next we check the outside tail. 
By symmetry take $t \ge 1$.  
From \eqref{eq:sqrt-localizer-def},
\begin{align*}
    M'(t) \le s_{M'}\int_{t-\alpha_{M'}}^{\infty} g_{M'}(v) dv.
\end{align*}
For the $i$ satisfying $\delta^2/\ell^2 + 4^i/r^2 \le C v \le \delta^2/\ell^2 + 4^{i + 1}/r^2$, the $i$-th dyadic factor is active and contributes $\exp(-\Omega(C 2^i)) = \exp(-\Omega(C r \sqrt{C v - \delta^2 / \ell^2}))$. 
Hence, for $1 \le t \le 1 + 1/C$,
\begin{align*}
    M'(t) \le (\delta / r)^{\Omega(C)} \exp(-\Omega(C r \sqrt{C (t - \alpha_{M'}) - \delta^2 / \ell^2})) \le (\delta / r)^{\Omega(C)} \exp(-\Omega(C r \sqrt{t - 1})).
\end{align*}
For $|t| \ge 1 + 1/C$, all dyadic factors are active and a similar calculation gives
\begin{align*}
    M'(t) \le (\delta / r)^{\Omega(C)} \cdot |\pi t / 2|^{-\Omega(C r)}.
\end{align*}
\end{proof}







\begin{lemma} \label{lem:orthogonality_by_localizer_under_conjecture}
Assume that Conjecture~\ref{conj:growth_rate} holds.
For two signals of Fourier sparsity $\ell$ and $r$ separately with $\ell \le r$,
\begin{align*}
    w(t):=\sum_{j=1}^{\ell}\alpha_j e^{2\pi\bi f'_jt} \qquad \text{ and }
    \qquad
    z(t):=\sum_{j=1}^{r}\beta_j e^{2\pi\bi f_jt},
\end{align*}
if the distance between their frequencies $\min_{j,j'} |f_j-f'_{j'}| \ge C_H \left(\frac{4 \ell r}{\delta} + \frac{\ell^2 \log(r/\delta)}{\delta^2}\right)$ for some constant $C_H$, then
\begin{align*}
    |\langle w, z\rangle_{[-1, 1]}| \le \delta \cdot \|w\|_{[-1,1]} \cdot \|z\|_{[-1,1]}..
\end{align*}
\end{lemma}

\begin{proof}
Let $M':=M'_{\ell,r,\delta}$.  We use the same decomposition as in
\eqref{eq:filter_decomposition}:
\begin{align*}
    \langle w,z\rangle_{[-1,1]} =\langle M'w,z\rangle_{[-1,1]} +\langle(1-M')w,z\rangle_{[-1,1]}.
\end{align*}
By Property~4 of Claim~\ref{claim:sqrt-localizer-bounds} and the separation hypothesis, $\int_{\R}M'(t)w(t)\overline{z(t)}\,dt=0$.
Thus the absolute value of the filtered term in $[-1,1]$ equals that of its tail in $\R\setminus[-1,1]$.

Under Conjecture~\ref{conj:growth_rate}, Property~\ref{item:uniform_bound_on_interval} and \ref{item:polynomial_bound_outside_interval} of Lemma~\ref{lemma:bounds_Fourier_sparse_signals} implies, for $|t| \ge 1$,
\begin{align*}
    |w(t)z(t)|
    \le \poly(r)\|w\|_{[-1,1]}\|z\|_{[-1,1]}
    \begin{cases}
       \exp(O(r\sqrt{|t|-1})),&1\le |t|\le1+1/C,\\
       |e(t+1)|^{O(r)},&|t|\ge1+1/C.
    \end{cases}
\end{align*}

By Properties~2 and~3 of Claim~\ref{claim:sqrt-localizer-bounds}, we have
\begin{align*}
    |\langle M'w,z\rangle_{[-1,1]}|
    &\le \poly(r)\left(\frac\delta r\right)^{\Omega(C)} \|w\|_{[-1,1]}\|z\|_{[-1,1]}\\
    & \quad\cdot\left( \int_1^{1+1/C} e^{-\Omega(Cr\sqrt{t-1})+O(r\sqrt{t-1})}\,dt +\int_{1+1/C}^{\infty}t^{-\Omega(Cr)+O(r)}\,dt \right)\\
    & \le \frac{\delta}{2}\|w\|_{[-1,1]}\|z\|_{[-1,1]}.
\end{align*}

For the flatness-error term, let $I=\left[-1+\frac{2\delta^2}{C\ell^2}, 1-\frac{2\delta^2}{C\ell^2}\right]$.
Property~\ref{item:uniform_bound_on_interval}
of Lemma~\ref{lemma:bounds_Fourier_sparse_signals} then yields
\begin{align*}
    \int_{-1}^{1}|1-M'(t)|^2|w(t)|^2\,dt
    \le \left(\frac\delta r\right)^{\Omega(C)} \|w\|_{[-1,1]}^2 + O\left(\frac{\delta^2}{C\ell^2}\right) O(\ell^2)\|w\|_{[-1,1]}^2
    \le \frac{\delta^2}{4}\|w\|_{[-1,1]}^2.
\end{align*}
Cauchy--Schwarz bounds the flatness-error term by $(\delta/2)\|w\|_{[-1,1]}\|z\|_{[-1,1]}$.  Combining the two terms proves the lemma.
\end{proof}

\begin{lemma}
\label{lem:filtered_orthogonality_by_localizer_under_conjecture}
Assume that Conjecture~\ref{conj:growth_rate} holds.
For $\ell \le r\le k$ and two signals of Fourier sparsity $\ell$ and $r$ separately,
\begin{align*}
    w(t):=\sum_{j=1}^{\ell}\alpha_j e^{2\pi\bi f'_jt} \qquad \text{ and }
    \qquad
    z(t):=\sum_{j=1}^{r}\beta_j e^{2\pi\bi f_jt},
\end{align*}
if the distance between their frequencies $\min_{j,j'} |f_j-f'_{j'}| \ge \min\left\{ C_H \left(\frac{4 \ell r}{\delta} + \frac{\ell^2 \log (r/\delta)}{\delta^2}\right) \log^2 k, 2\Delta_k \right\}$ for some constant $C_H$, then
\begin{align*}
    |\langle H_k w,H_k z\rangle|
    \le \delta \cdot \|H_k w\|_2 \cdot \|H_k z\|_2.
\end{align*}
\end{lemma}

\begin{proof}
Let $d:=\min_{j,j'}|f_j-f'_{j'}|$.  If $d\ge2\Delta_k$, then $\widehat{H_kw}$ and $\widehat{H_kz}$ have disjoint supports, and the claim follows from Parseval's identity.

Suppose that $d < 2\Delta_k$.  
Since $\Delta_k = O(k^2\log^2 k)$, the separation hypothesis implies
$\delta^2=\Omega\left(\frac{\epsilon \ell^2}{k^2}\right)$.
And Lemma~\ref{lem:orthogonality_by_localizer_under_conjecture} shows $|\langle w,z\rangle_{[-1,1]}| \le O(\delta)\|w\|_{[-1,1]}\|z\|_{[-1,1]}$.

Similarly to the proof of Lemma~\ref{lem:filtered_orthogonality_by_localizer},
we consider the same interval $I=[-1+2\epsilon/(Ck^2),1-2\epsilon/(Ck^2)]$.
Claim~\ref{claim:H-bounds} and Lemma~\ref{lemma:bounds_Fourier_sparse_signals} implies that
\begin{align*}
    |\langle H_kw,H_kz\rangle_{\R} -\langle w,z\rangle_{[-1,1]}|
    &\le O(\delta)\|w\|_{[-1,1]}\|z\|_{[-1,1]}.
\end{align*}
This estimate follows by splitting at $I$: on $I$,$|H_k^2-1|\le(\epsilon/k)^{\Omega(C)}\le O(\delta^2)$; on $[-1,1]\setminus I$, its $O(\epsilon/k^2)$ width and the uniform bound on $w$ give an $O(\delta^2)\|w\|_{[-1,1]}^2$ contribution; outside $[-1,1]$, Claim~\ref{claim:H-bounds} and Lemma~\ref{lemma:bounds_Fourier_sparse_signals} give an $O(\delta)\|w\|_{[-1,1]}\|z\|_{[-1,1]}$  contribution for sufficiently large $C$.

Finally, by Lemma~\ref{lemma:accuracy-tunable-localization}, we have
$\|w\|_{[-1,1]}\le O(1)\|H_kw\|_2$, and $\|z\|_{[-1,1]}\le O(1)\|H_kz\|_2$.

Combining the above bounds proves the lemma.
\end{proof}

\end{document}